\documentclass[a4paper,11pt]{amsart}
\usepackage{amsmath,amsthm,amssymb,amsfonts,enumerate,color,esint,bm}
\usepackage[pdftex]{graphicx}
\usepackage{float}
\usepackage{tikz}
\usepackage{soul}
\usepackage{tabu}
\usetikzlibrary{decorations.pathreplacing}
\usepackage{subcaption}
\usepackage{placeins}

\usepackage[colorlinks=true, allcolors=blue]{hyperref}
\usepackage{amsrefs}

\newcommand{\R}{\mathbb{R}}
\renewcommand{\H}{\mathcal{H}}

\def\({\left(}
\def\){\right)}

\renewcommand{\u}{\bm{u}}
\newcommand{\f}{\bm{f}}
\newcommand{\x}{\bm{x}}
\renewcommand{\b}{\bm{b}}

\newtheorem{thm}{Theorem}[section]
\newtheorem*{thm*}{Theorem}

\newtheorem{lem}[thm]{Lemma}

\theoremstyle{definition}

\newtheorem{rem}[thm]{Remark}

\usepackage[normalem]{ulem}
\usepackage{soul}
\definecolor{nicegreen}{rgb}{0.0, 0.5, 0.0}

\numberwithin{equation}{section}

\allowdisplaybreaks

\author[P. Seleson]{\href{https://pabloseleson.ornl.gov/}{Pablo Seleson}\textsuperscript{*}}

\address[P. Seleson]{Computer Science and Mathematics Division\\
Oak Ridge National Laboratory\\
P.O. Box 2008, MS-6013\\
Oak Ridge, TN 37831\\ USA}
\email{selesonpd@ornl.gov}

\author[P. R. Stinga]{\href{https://pabloraulstinga.github.io/}{Pablo Ra\'ul Stinga}\textsuperscript{*}}

\address[P. R. Stinga]{Department of Mathematics\\
Iowa State University\\
396 Carver Hall\\
 Ames, IA 50011\\ USA}
\email{stinga@iastate.edu}

\author[M. Vaughan]{\href{https://maryvaughan.github.io/}{Mary Vaughan}\textsuperscript{*}}

\address[M. Vaughan]{Department of Mathematics\\
Texas State University\\
MCS 470\\ 
San Marcos, TX 78666 \\ USA}
\email{vaughan@txstate.edu}

\thanks{{*}Authors are listed alphabetically.\\
This manuscript has been authored in part by UT-Battelle, LLC, under contract DE-AC05-00OR22725 with the US Department of Energy (DOE). The US government retains and the publisher, by accepting the article for publication, acknowledges that the US government retains a nonexclusive, paid-up, irrevocable, worldwide license to publish or reproduce the published form of this manuscript, or allow others to do so, for US government purposes. DOE will provide public access to these results of federally sponsored research in accordance with the DOE Public Access Plan ( https://www.energy.gov/doe-public-access-plan ).
}

\keywords{
Peridynamics, bond failure, bond breaking, crack propagation, fracture}

\subjclass[2020]{Primary: 
74A70, 
74R10, 
74A45. 
Secondary: 
45K05, 
65R20. 
}

\begin{document}

 \title[Bond failure in peridynamics]{Bond failure in peridynamics: Nonequivalence of critical stretch and critical energy density criteria}

\begin{abstract}
This paper rigorously analyzes bond failure in the peridynamic theory of solid mechanics, which is a fundamental component of fracture modeling. 
We compare analytically and numerically two common bond-failure criteria:~{\em critical stretch} and~{\em critical energy density}. In the former, bonds fail when they stretch to a critical value, 
 whereas in the latter, bonds fail  
when the bond energy density exceeds a threshold. By focusing the analysis on bond-based models, 
we prove mathematically that the critical stretch criterion and the critical energy density criterion are not equivalent in general and result in different bond-breaking and fracture phenomena. 
Numerical examples showcase the striking differences between the effect of the two criteria on crack dynamics, including the crack tip evolution, crack propagation, and crack branching.  
\end{abstract}

\maketitle

\section{Introduction}

Modeling material failure and damage is an ongoing challenge in computational science and engineering. 
Computational solid mechanics has been remarkably successful in the modeling of material behavior, especially based on the finite element method~\cite{Liu-et-al-2022}. However, the continuity assumptions of the basic equations of classical continuum mechanics require special techniques to model material discontinuities, such as propagating cracks.

Peridynamics~\cite{Silling, Sillingetal2007} is a nonlocal formulation of classical continuum mechanics, where the mathematical equations are reformulated to accommodate 
material  
discontinuities. The nonlocality of peridynamics is expressed in   
that material points interact with surrounding material points within a finite neighborhood. The connections between them are referred to as {\em bonds}. Peridynamic governing equations replace spatial differential operators 
with 
integral operators that do not require 
differentiability assumptions, 
 thus allowing the representation of material separation. 
Connections between peridynamics and classical theory  have been established. 
For elastic materials, peridynamics was shown to converge to classical elasticity under suitable differentiability assumptions as nonlocality vanishes~\cite{Silling-Lehoucq-2008}, and a link between peridynamics and classical linear elastic fracture mechanics was presented in~\cite{Lipton-2014}. 

To model material separation, such as fracture, peridynamics introduces the concept of {\em bond failure}, which means that a bond is no longer able to transfer forces between points. This occurs when the bond reaches or exceeds the threshold defined by the failure criterion, which results in bond breaking. 
It is common to assume that bonds break irreversibly:~once a bond is broken, it remains broken for the remainder of the simulation.  
The incorporation of bond breaking into peridynamic computations enables the modeling of crack initiation and propagation without the need for an external crack law. Therefore, the choice of bond-failure criterion is essential for properly modeling fracture 
in peridynamics. 

The most commonly used bond-failure criteria 
 in peridynamics are the~{\em critical stretch}~\cite{SillingAskari} and the~{\em critical energy density}~\cite{FosterSillingChen}. 
The critical stretch criterion  
establishes that bonds fail 
once they stretch to a critical value.\footnote{In~\cite{SillingAskari}, bonds remain unbroken as long as their stretch is below the critical stretch, implying that a bond breaks once it reaches this limit. 
However, in later works, the critical stretch criterion often assumes that bonds break once their stretch exceeds this limit.
}
The critical energy density criterion, in contrast,  
states that bonds fail 
once their energy density exceeds a threshold.  
The critical value for bond failure  
for each criterion can be determined 
by calculating the total energy per unit fracture area of bonds crossing a fracture surface (assuming all those bonds are in their critical failure state) and comparing this quantity to the fracture energy; see Sections \ref{sec:criteria3D} and \ref{sec:criteria2D} for calculations in three and two dimensions, respectively. 
Because one criterion relies on bond stretch and 
the other relies on bond energy density, 
the question of which criterion should be implemented into a model naturally arises. 

In this paper, we meticulously investigate the analytical differences between the critical stretch and critical energy density criteria and present numerical comparisons. 
In particular, we prove that the two criteria are not equivalent in general. 
To the best of the authors' knowledge, 
this is the first mathematically rigorous comparison of the two criteria to appear in the literature. 
A numerical comparison between the critical stretch criterion and the critical energy density criterion was presented in~\cite{Dipasqualeetal2017} for two-dimensional mixed-mode I-II fracture cases, but 
neither analytical investigations nor deep numerical comparisons between these criteria were studied. 

We focus our work on bond-based peridynamics~\cite{Silling} and employ a model for brittle elastic materials based on a generalization of the commonly used prototype microelastic brittle (PMB) material model~\cite{SillingAskari}.    
The model utilizes a material-dependent influence function $\omega(r)$ that depends only on the bond length ${r>0}$ and helps determine how much each bond contributes to the peridynamic internal force density. 
Moreover, we consider both two- and three-dimensional peridynamic formulations. Details on the model are presented in Section~\ref{sec:model}. 

We recast the critical energy density criterion as a critical stretch criterion with a bond-dependent critical stretch; see Lemmas \ref{lem:energy-to-stretch} and \ref{lem:energy-to-stretch2D} for three- and two-dimensional formulations,  respectively. 
This enables direct comparison between the two criteria 
via the 
 corresponding critical stretch expressions and leads to our main results, Theorems \ref{thm:main3D} and~\ref{thm:main2D}. 
 Here is an informal version of those results. 

\begin{thm}\label{thm:intro}
In the generalized PMB model, 
the critical stretch criterion and the critical energy density criterion coincide if and only if the influence function is $\omega(r) = \beta r^{-1}$ for some $\beta>0$.
\end{thm} 

In particular, we prove that the two criteria are \emph{not} equivalent in general.
This includes the original PMB model for which $\omega(r) \equiv 1$ \cite{SillingAskari}. 
Indeed, the \emph{only} case in which they are equivalent is when the influence function is exactly $\omega(r) = \beta r^{-1}$ for some $\beta>0$. 

Remarkably, when they are not equivalent, the two criteria yield  different bond-breaking and fracture behaviors. 
As an illustrative example, we consider influence functions given by
\begin{equation}\label{eq:omega}
\omega(r) = r^{-\alpha} \quad \hbox{for some}~\alpha < d+1,
\end{equation} 
where $d = 2,3$ is the spatial dimension. 
We show analytically that the choices of bond-failure criterion and of influence function dictate if either shorter or longer bonds will break first; see 
Theorem~\ref{thm:3Dbehavior} for the three-dimensional case and Theorem~\ref{thm:2Dbehavior} for the two-dimensional case. 
We informally write those results here. 

\begin{thm}\label{thm:intro2}
In the generalized PMB model, let $\omega(r) = r^{-\alpha}$ for $\alpha \not=1$.
\begin{enumerate}
\item Assume that $\alpha <1$. Shorter bonds that break under the critical stretch criterion may not break under the critical energy density criterion. Longer bonds that break under the critical energy density criterion may not break under the critical stretch criterion. 
\item Assume that $1< \alpha < d+1$. Longer bonds that break under the critical stretch criterion may not break under the critical energy density criterion. Shorter bonds that break under the critical energy density criterion may not break under the critical stretch criterion. 
\end{enumerate}
\end{thm}

Recasting the critical energy density criterion as a bond-dependent critical stretch criterion also allows for straightforward numerical implementation. We simply modify the expression of the critical stretch in the bond-breaking condition in peridynamics codes that already implement the critical stretch criterion. 

We confirm our mathematical results through several numerical experiments on two-dimensional problems, comparing the critical stretch and critical energy density criteria, using influence functions of the form \eqref{eq:omega}. 
First, we consider a simple isotropic extension, which showcases the strikingly different bond-breaking patterns for various influence functions, i.e.,~for different values of $\alpha$ (see Section~\ref{sec:isotropicextension}).
Then, we illustrate the notably distinct bond-breaking behavior exhibited by the two bond-failure criteria during the evolution of a crack tip for given influence functions (see Section~\ref{sec: Crack tip evolution}).  
Both of these examples substantiate Theorem \ref{thm:intro2}. 
Finally, we study crack propagation and branching, using an example of a pre-notched plate subjected to traction loading as an illustration.  
Our results confirm that there are undeniable differences when implementing the critical stretch criterion versus the critical energy density criterion. 
Indeed, taking $\alpha =2$, for the critical stretch criterion, the crack evolves as a single horizontal line, but for the critical energy density criterion, the crack branches. After increasing the traction, both criteria result in crack branching, but the crack paths are significantly different (see Section~\ref{sec:example3}). 
Taking instead $\alpha=0$, the different behaviors are apparent but more subtle. For instance, we find that the crack tip propagates slower under the critical stretch criterion than under the critical energy density criterion. 

Our analytical results and numerical simulations indicate major implications in modeling that need to be understood. 
It is not clear a priori which bond-failure criteria should be utilized and thus deeper mathematical and numerical investigations are required. 

\subsection{Literature review}

The original peridynamic formulation was based on pairwise interactions and is referred to as {\em bond-based peridynamics}~\cite{Silling}. This formulation was later extended to multibody interactions under {\em state-based peridynamics}~\cite{Sillingetal2007}, which eliminates modeling restrictions imposed by the bond-based formulation such as a fixed Poisson's ratio (see Remark~\ref{rem:poisson}). As indicated above, we focus on comparing the critical stretch and critical energy density criteria in bond-based peridynamics, a formulation that models the response of a bond independently of the state of other bonds, thus enabling the establishment of single-bond rules for bond breaking for both criteria. 
 
Variations of the critical stretch and critical energy density criteria exist in the literature.   
The commonly used critical stretch criterion assumes that the critical stretch of a given bond is independent of the state of 
other bonds. 
A generalization was presented in~\cite{SillingAskari}, in which the critical stretch depends on the minimum stretch among all bonds connected to a given point. 
Moreover, while the critical stretch criterion was initially proposed for a bond-based model, derivations of a critical stretch for a state-based model were presented in~\cite{Madenci-Oterkus-2014, Zhang-Qiao-2018, Aguiar-Patriota-2023}, and various approaches for incorporating critical stretch-based bond failure 
within a state-based model were discussed in~\cite{Karpenko-et-al-2020}. 
To account for material anisotropy, an orientation-dependent critical stretch criterion 
was presented in~\cite{Ghajari-et-al-2014} for the modeling of orthotropic media,  
 with applications to unidirectional fiber-reinforced composites, cortical bone, and polycrystalline microstructures. 
This criterion was also employed in~\cite{Ren-et-al-2018} for the modeling of fiber-reinforced composite laminates, and a similar approach was discussed in~\cite{Ren-et-al-2022}. 
Furthermore, a critical energy density criterion, similar to the one proposed in~\cite{FosterSillingChen} but based on the $J$-integral,  
 was presented in~\cite{Madenci-Oterkus-2016} with applications involving plasticity. 

For mathematical analysis, continuous damage formulations were proposed that enable the weakening of bond forces without resulting in immediate bond breaking. 
A modified critical stretch criterion that incorporates partial bond failure through linear bond degradation was proposed for a bond-based model in~\cite{Emmrich-Puhst-2016}, 
and a similar approach was presented in~\cite{Du-et-al-2018}. A bond-based model with softening behavior, but without irreversible bond breaking, was presented in~\cite{Lipton-2016}. This softening approach was extended to a state-based model that enables irreversible damage in~\cite{Lipton-et-al-2018}. State-based models extending those from~\cite{Du-et-al-2018} and~\cite{Lipton-et-al-2018} were presented in~\cite{Aguiar-Patriota-2023} with a focus on applications rather than on mathematical analysis; these models are related to one of the approaches discussed in~\cite{Karpenko-et-al-2020}.

Other bond-failure criteria have also been proposed. For example, several works featured 
 the critical shear angle criterion, which is based on bond rotation. 
This criterion was used for Mode~II (in-plane shear) fracture analysis in~\cite{Zhang-Qiao-2019} (where it is referred to as the critical skew criterion) and  for Mode~III (anti-plane shear) and torsional fracture analysis in~\cite{Oterkus-Madenci-2015}. 
To model fiber-reinforced composites, \cite{Oterkus-Madenci-2012} combined the critical stretch criterion for normal interlayer bonds with the critical shear angle criterion for shear bonds to simulate the interaction between neighboring plies in a laminate, thereby enabling Mode~I (tensile) and Mode~II failure, respectively.
To distinguish tensile cracks from shear cracks in mixed-mode~I-II fracture analysis of rock-like materials, the critical stretch and critical shear angle criteria were augmented, respectively, with dilatational and deviatoric deformation information in~\cite{Wang-et-al-2023}. 
In~\cite{Wang-et-al-2018}, a conjugated bond-pair-based peridynamic model was presented that uses the critical stretch criterion for Mode~I failure and the critical bond shear energy density criterion for Mode~II and~III failure. 
A combined bond-failure criterion was presented in~\cite{Madenci-et-al-2021}, under which a bond fails if it satisfies either the critical stretch criterion or the critical shear angle criterion. 
A similar concept, but based on the bond strain energy density, was presented in~\cite{Li-et-al-2025}:~a bond fails if the strain energy density associated with its normal or tangential deformations exceeds a threshold. 
An approach that compensates the critical energy density criterion with the critical shear angle criterion, to account for shear bonds that might otherwise be  overlooked, was presented in~\cite{Yang-et-al-2024}.  
Bond failure based on a critical angle between bond pairs was used in~\cite{O'Grady-Foster-2014} for a peridynamic Euler--Bernoulli beam model.  
 In the context of micropolar peridynamics~\cite{Gerstle-et-al-2007}, 
  bond-failure 
 criteria based on (tensile and compressive) critical stretch, 
critical shear angle, 
 and (tensile and compressive) frictional sliding were discussed in~\cite{Diana-Casolo-2019}. 
 To account for both tensile and shear damage, the critical deviatoric bond strain criterion was introduced in~\cite{Ren-et-al-shear}.
 
 Bond-failure criteria that incorporate classical concepts of stress or strain have also been proposed. 
 Criteria based on strain invariants---the equivalent strain (which measures the shearing strain) and the average volumetric strain---were discussed in~\cite{Warren-et-al-2009}. 
Stress-based criteria appeared in several works, as follows. 
The maximum  principal stress criterion for tensile failure and the Mohr-Coulomb failure criterion for shear failure were presented in~\cite{Zhou-Wang-2016}. The maximum principal stress criterion for tensile failure and the twin shear strength criterion for shear failure were employed in~\cite{Shou-et-al-2019}. The maximum and minimum principal stress criteria for tensile and compressive failure, respectively, were applied in~\cite{Dipasquale-et-al-2022}. Other proposed stress-based criteria include the Tsai-Hill criterion~\cite{Hattori-et-al-2018}, the incubation time criterion~\cite{Ignatiev-et-al-2020, Ignatiev-et-al-2021, Ignatiev-et-al-2023}, the mean stress criterion~\cite{Ignatiev-et-al-2023}, and the remote stress criterion~\cite{Ignatev-Oterkus-2025}. 
 An approach combining a stress-based failure criterion with a bond-failure criterion was proposed in~\cite{Song-et-al-2025}: the model checks whether the stress-based condition is met at a point before assessing the bond-failure condition for bonds connected to that point. 
Some of the approaches described in this paragraph were presented within the peridynamic correspondence modeling framework; this framework enables the direct incorporation of classical continuum mechanics constitutive models in peridynamics~\cite{Sillingetal2007}. 
An approach to integrate classical continuum damage laws within the peridynamic correspondence modeling framework was presented in~\cite{Tupek-etal-2013} 
by modifying the influence function based on the state of accumulated damage, demonstrated using the Johnson–Cook damage model for ductile damage in metals.   

Despite the wide variety of bond-failure criteria put forward in the literature, the critical stretch and critical energy density criteria remain the most broadly used approaches in peridynamic fracture simulations and thus constitute the focus of our analysis. 
 
\subsection{Organization of the paper}
 
The organization of this paper is as follows. 
In Section~\ref{sec:model}, we review the bond-based peridynamic formulation and a brittle elastic material model for two- and three-dimensional problems. 
In Sections \ref{sec:3D} and \ref{sec:2D}, for three- and two-dimensional problems respectively, we discuss the critical stretch and critical energy density criteria and present our main results: Theorems \ref{thm:main3D} and \ref{thm:main2D}. 
In Section~\ref{sec: numerical examples}, we present a numerical comparison between the two bond-failure criteria with various examples. A summarizing discussion is presented in Section~\ref{sec: discussion}.

\section{Bond-based peridynamic model for generalized PMB materials}\label{sec:model}

We consider two- and three-dimensional peridynamics problems. 
For the three-dimensional formulation, we let $\Omega \subset \mathbb{R}^3$ be a bounded domain. 
For the two-dimensional formulation, we still work in three spatial dimensions and consider a plane stress or plane strain body of thickness $h$ in the out-of-plane direction, represented by a two-dimensional domain $\Omega \subset \mathbb{R}^2$. 
We will refer to these as 3D and 2D problems, respectively.

Let $\u(\x,t)$ denote the displacement of a point $\x \in \Omega$ at time $t$. 
The bond-based peridynamic equation of motion introduced in \cite{Silling} is 
\begin{equation}\label{eq:pd}
\rho \ddot{\u}(\x,t) = \int_{\H_{\x}} \f(\u(\x',t) - \u(\x,t), \x',\x,t)\, dV_{\x'} + \b(\x,t),
\end{equation}
for the three-dimensional case, and 
\begin{equation}\label{eq:pd2D}
\rho \ddot{\u}(\x,t) = h\int_{\H_{\x}} \f(\u(\x',t) - \u(\x,t), \x',\x,t)\, dA_{\x'} + \b(\x,t),
\end{equation}
for the two-dimensional case (see~\cite{SelesonEtAl}), where $\rho$ is the mass density, $\ddot{\u}$ is the second derivative of ${\u}$ with respect to time, i.e., the acceleration, 
$\H_{\x}$ is a neighborhood of $\x$, 
$\f$ is the pairwise force function (with units of force per volume squared), and
$\b$ is a prescribed body force density.
We assume that
\[
\H_{\x} = \{\x' \in \Omega: \| \x' - \x \| \leq \delta \},
\]
where $\delta>0$ is the peridynamic horizon and $\|\cdot\|$ denotes the Euclidean norm.
We deal only with homogeneous bodies, but it is necessary to indicate separate dependence on $\x$, $\x'$, and~$t$ in the pairwise force function in \eqref{eq:pd} and \eqref{eq:pd2D} to account for bond failure. 

As is customary, we denote the relative position vector in the reference configuration and the relative displacement vector, respectively, by 
\[
\bm{\xi} := \x' - \x \quad \hbox{and} \quad \bm{\eta} := \u(\x',t) - \u(\x,t).
\] 
The vector $\bm{\xi}$ is also known as the \emph{bond}. 
With this, we can write $\f =\f(\bm{\eta}, \x',\x,t)$ in the new variables. 
The following additional restrictions on $\f$ are based on linear and angular momentum \cite{Silling}:
\[
\f(-\bm{\eta}, \x,\x',t) = - \f(\bm{\eta}, \x',\x,t)
\quad \hbox{and} \quad
(\bm{\xi}+ \bm{\eta}) \times \f(\bm{\eta}, \x',\x,t) = \bm{0}
\]
for all $\bm{\eta}, \x',\x \in \R^d$ and $t\geq0$.
Noting that $\bm{\xi} + \bm{\eta}$ is the relative position vector in the deformed configuration, the \emph{bond stretch} is a measure of strain defined as
\begin{equation}\label{eq:stretch}
s(\bm{\eta},\bm{\xi})
	:= \frac{\|\bm{\xi} + \bm{\eta}\| - \|\bm{\xi}\|}{\|\bm{\xi}\|}.
\end{equation}

\begin{rem}[Isotropic extension]\label{rem:isotropic}
Several times, we will consider a static (i.e., time-independent) problem involving a body under an isotropic extension given by $\u(\x) =  \bar{s} \x$ for some $\bar{s}>0$ constant. 
This results in the relation $\bm{\eta} = \bar{s} \bm{\xi}$ between the bond and its relative displacement as well as a constant stretch, $s(\bm{\eta},\bm{\xi}) = \bar{s}$, for all bonds~$\bm{\xi}\in \H\setminus\{\bm{0}\}$, where $\H:=\H_{\bm{0}}$.
\end{rem}

A bond-based peridynamic material is microelastic if there is a scalar-valued function $w = w(\bm{\eta},\x',\x,t)$, called the pairwise potential or micropotential, such that
\[
\f(\bm{\eta}, \x',\x,t) = \frac{\partial w}{\partial \bm{\eta}}(\bm{\eta}, \x',\x,t).
\] 
The micropotential $w(\bm{\eta},\x',\x,t)$ is the energy density (energy per unit volume squared) stored in the bond at time $t$. 
It can be shown that 
$w$ only depends on $\bm{\eta}$ through $\|\bm{\eta} + \bm{\xi}\|$ and
 there is a scalar-valued function $f$ such that
\[
\f(\bm{\eta},\x',\x,t) 
	= f(\bm{\eta}, \x',\x,t) \frac{\bm{\xi} + \bm{\eta}}{\|\bm{\xi} + \bm{\eta}\|},
\]
see~\cite{SelesonEtAl, Silling}.
The macroelastic energy density   
is given by
\begin{equation}\label{eq:W0}
W = W(\x,t) = \frac{1}{2} \int_{\H_{\x}} w(\bm{\eta}, \x',\x,t) \, dV_{\x'},  
\end{equation}
for the three-dimensional case, and 
\begin{equation}\label{eq:W02D}
W = W(\x,t) = \frac{1}{2} h \int_{\H_{\x}} w(\bm{\eta}, \x',\x,t) \, dA_{\x'},  
\end{equation}
for the two-dimensional case (see~\cite{SelesonEtAl}).

We will only be concerned with brittle elastic materials described by the generalized prototype microelastic brittle (GPMB) constitutive model introduced in \cite{SelesonParks,Seleson} (a generalization of the PMB model in \cite{SillingAskari})  given by 
\begin{equation}\label{eq:f-PMB}
f(\bm{\eta}, \x',\x,t) = 
\mu(\x',\x,t)\, c\, \omega(\|\bm{\xi}\|) s(\bm{\eta},\bm{\xi}),
\end{equation}
where $c$ is the micromodulus constant (or bond elastic constant), $\omega = \omega(r):(0,\delta] \to [0,\infty)$  is the influence function,
$s$ is the bond stretch given in \eqref{eq:stretch}, and $\mu \in \{0,1\}$ is a history-dependent, Boolean-valued bond-breaking function. 
In particular, $\mu(\x',\x,t) = 1$ if the bond is intact for all $0  < \tilde{t} \leq t$ and $\mu(\x',\x,t) = 0$ if the bond broke at some time $\tilde{t} \leq t$ (see \eqref{eq:mu-stretch} and \eqref{eq:mu-energy} below).
The corresponding micropotential is given by
\begin{equation}\label{eq:micropotential}
w(\bm{\eta}, \x',\x,t) = 
\frac{1}{2} \mu(\x',\x,t) \,c\,\omega(\|\bm{\xi}\|) s^2(\bm{\eta},\bm{\xi}) \|\bm{\xi}\|.
\end{equation}
We assume  $\omega(r)>0$ for $r \in (0,\delta]$. 
Note that $\omega \equiv 1$ corresponds to the PMB model in \cite{SillingAskari}.  

We will now determine the micromodulus constant $c$ in terms of material properties. For this, we separately consider two- and three-dimensional problems.

\begin{rem}\label{rem:poisson}
Bond-based peridynamics is restricted to modeling of materials with a Poisson's ratio of $ \nu=\frac14$ in three dimensions (see \cite{Silling}). In two-dimensional problems, bond-based peridynamic models are limited to a Poisson's ratio of $ \nu=\frac14$ for plain strain and $\nu=\frac13$ for plane stress (see \cite{Gerstle,Trageser-Seleson-2020}).
\end{rem}

\subsection{Three-dimensional problems}

For more details on the following discussion, see \cite{SillingAskari}. 

Consider a three-dimensional problem of a body under an isotropic extension, so recalling Remark~\ref{rem:isotropic}, 
$s(\bm{\eta},\bm{\xi}) = \bar{s}$ is constant for some $\bar{s}>0$. 
Assume also that no bonds are broken at time $t$, so that $\mu \equiv 1$ in \eqref{eq:micropotential}. 
With the explicit form of $w$ in \eqref{eq:micropotential}, we use spherical coordinates in \eqref{eq:W0} to find 
\begin{align*}
W 
	&= \frac{c \bar{s}^2}{4} \int_{\H}  \omega(\|\bm{\xi}\|)\|\bm{\xi}\|  \, dV_{\bm{\xi}} 
	=c \bar{s}^2 \pi \int_0^\delta  \omega(r)r^3 \, dr,
\end{align*}
for any point in the bulk of the body (i.e., farther than $\delta$ from the domain boundary), where $\H$ is defined in Remark~\ref{rem:isotropic}.  
On the other hand, the strain energy density in classical linear elasticity for the same isotropic extension is $W^{C} = 9K\bar{s}^2/2$, where $K$ is the bulk modulus.  
Assuming $W= W^{C}$, we get 
\begin{equation}\label{eq:c-3D}
c = \frac{9K}{\displaystyle 2 \pi \int_0^\delta  \omega(r)r^3 \, dr} \qquad \hbox{for a 3D problem}. 
\end{equation}

\subsection{Two-dimensional problems}\label{sec:2Dc}

For more details on the following discussion, see \cite{SelesonEtAl}. 

Consider a two-dimensional problem of a body in a state of plane strain or plane stress under an isotropic extension, so recalling Remark \ref{rem:isotropic}, 
$s(\bm{\eta},\bm{\xi}) = \bar{s}$ is constant for some $\bar{s}>0$. 
As above, but now using polar coordinates in \eqref{eq:W02D}, we get
\begin{align*}
W 
	&=  \frac{h c \bar{s}^2}{4} \int_{\H}  \omega(\|\bm{\xi}\|)\|\bm{\xi}\|  \, dA_{\bm{\xi}} 
	=\frac{h c \bar{s}^2 \pi}{2} \int_0^\delta  \omega(r)r^2 \, dr,
\end{align*}
for any point in the bulk of the body.

\subsubsection*{Plane strain} 
For a body in a state of plane strain and under an isotropic extension given by $\u(\x) =  \bar{s} \x$  for some $\bar{s}>0$, it is known from classical linear elasticity that the strain energy density 
for materials with a Poisson's ratio of $\nu=\frac14$ (recall Remark~\ref{rem:poisson})
is $W^{C\varepsilon} = 8E \bar{s}^2/5$, where $E$ is the Young's modulus.  
Assuming $W= W^{C \varepsilon}$, we find that
\begin{equation}\label{eq:c-planestrain}
c =\frac{16E}{5 \pi h \displaystyle \int_0^\delta  \omega(r)r^2 \, dr }
\qquad \hbox{for a 2D problem in a state of plane strain}.
\end{equation}

\subsubsection*{Plane stress} 
For a body in a state of plane stress and under an isotropic extension given by $\u(\x) =  \bar{s} \x$  for some $\bar{s}>0$, it is known from classical linear elasticity that the strain energy density
for materials with a Poisson's ratio of $\nu=\frac13$ (recall Remark~\ref{rem:poisson}) 
 is $W^{C\sigma} = 3E\bar{s}^2/2$.  
Assuming $W= W^{C \sigma}$, we find that
\begin{equation}\label{eq:c-planestress}
c=\frac{3E}{\pi h \displaystyle  \int_0^\delta  \omega(r)r^2 \, dr}
\qquad \hbox{for a 2D problem in a state of plane stress}.
\end{equation}

\begin{rem}\label{rem: h independence of model}
As noted in~\cite{SelesonEtAl}, although the expressions for the micromodulus constant~$c$ in~\eqref{eq:c-planestrain} and~\eqref{eq:c-planestress} depend on $h$,  in practice, $h$ does not appear in~\eqref{eq:pd2D} and~\eqref{eq:W02D} because, after substituting the expression of the GPMB pairwise force function (see~\eqref{eq:f-PMB}) and micropotential (see~\eqref{eq:micropotential}) into these equations, respectively, the product $hc$ emerges, which is $h$-independent.
\end{rem}

\section{Bond-failure criteria for three-dimensional problems}\label{sec:3D}

In this section, we analyze the two bond-failure criteria---critical stretch and critical energy density---in three dimensions.  
For consistency, for both criteria, we assume that bonds break once they reach a critical threshold. 

\subsection{Bond-failure criteria}\label{sec:criteria3D}

The bond-failure criteria are distinguished by their choice of Boolean-valued function $\mu$ in \eqref{eq:f-PMB} and \eqref{eq:micropotential}. 
In the critical stretch criterion \cite{SillingAskari}, there is a constant $s_0>0$ (the critical stretch) such that $\mu$ is given by
\begin{equation}\label{eq:mu-stretch}
\mu(\x',\x,t) = \mu_0(\x',\x,t):= \begin{cases}
1 & \hbox{if}~s(\widetilde{\bm{\eta}},\bm{\xi})   < s_0~\hbox{for all}~0 \leq \tilde{t} \leq t, \\
0 & \hbox{otherwise}, 
\end{cases}
\end{equation} 
where $\widetilde{\bm{\eta}} := \u(\x',\tilde{t}) - \u(\x,\tilde{t})$.  
On the other hand, in the critical energy density criterion \cite{FosterSillingChen}, there is a constant $w_c>0$ (the critical energy density) such that $\mu$ is given by 
\begin{equation}\label{eq:mu-energy}
\mu(\x',\x,t) = \mu_c(\x',\x,t): = \begin{cases}
1 & \hbox{if}~w(\widetilde{\bm{\eta}}, \x',\x,\tilde{t})  < w_c~\hbox{for all}~0 \leq \tilde{t} \leq t, \\
0 & \hbox{otherwise}. 
\end{cases}
\end{equation}
In both cases, we determine the critical values $s_0$ and $w_c$ in terms of material constants. 
Towards this end, we first describe the work required to break all bonds across a fracture surface.

\begin{figure}[htb]
\begin{tikzpicture}[scale=3.5]
\draw[fill=gray!20, line width=1.5pt] plot[domain=20:160] ({cos(\x)}, 0) -- plot[domain=20:160] ({cos(\x)}, {sin(\x)-sin(20)}) -- cycle; 
\draw[thick] (-1.15,0) -- (1.15,0); 
\draw (1.6,.01) node {\small Fracture surface}; 
\draw ({0}, {-sin(20)}) node[label={-5:$\x$}, circle, fill, scale=.3] {}; 
\draw[-latex, thick] ({0}, {-sin(20)}) -- ({cos(20)},0); 
	\draw (.7,-.15) node {$\delta$}; 
\draw[-latex, thick] ({0}, {-sin(20)}) -- ({.65*cos(30)}, {.65*sin(30)-sin(20)}); 
	\draw (.425,-.15) node {$r$}; 
\draw[thick, densely dotted] ({0}, {0}) -- (0,{-1}); 
	\draw [thick ,decoration={brace,amplitude=6pt}, decorate,] (.02,-.02)--(.02,{-sin(20)+.01});  
	\draw (.12,{-.5*sin(20)}) node {$z$};
\draw [thick ,decoration={brace,mirror,amplitude=6pt}, decorate,] (-.02,-.02)--(-.02,{-.99});  
	\draw (-.12,-.5) node {$\delta$};
\draw[thick,dashed] (0,0) -- (0,1); 
\draw [-latex, thick] (.155,.8) arc [start angle=0, end angle=340, x radius=4.5pt, y radius=1.25pt]; 
	\draw (.1,.9) node {$\theta$}; 
\draw[-latex, thick] plot[domain=90:32] ({.65*cos(\x)}, {.65*sin(\x)-sin(20)}); 
\draw (.45,.32) node {$\cos^{-1}\left(\frac{z}{r}\right)$}; 
\draw[-latex, thick] ({0}, {-sin(20)}) --  ({.75*cos(130)}, {.75*sin(130)-sin(20)}); 
	\draw ({.75*cos(130)}, {.75*sin(130)-sin(20)})  node[label={180:$\x'$}, circle, fill, scale=.3] {}; 
	\draw (-.25,-.15) node {$\bm{\xi}$}; 
\draw (-.2,.4) node {$\Omega_z$}; 
\end{tikzpicture}
\caption{Description of all bonds $\bm{\xi}$ connecting a point $\x$ located at a distance $z \in (0,\delta)$ below the fracture surface to points  $\x'$ in the spherical cap $\Omega_z$ above the fracture surface.}
\label{fig:integration}
\end{figure}
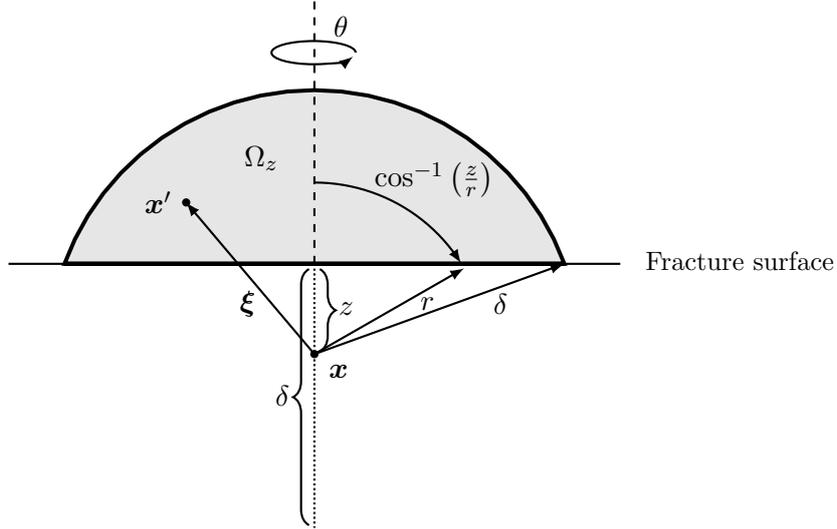

Split the body into two halves across a flat~fracture surface (see Figure \ref{fig:integration}). 
To physically separate the halves, we require that all bonds across the fracture surface break (i.e.,~reach the critical stretch $s_0$ or the critical energy density $w_c$). 
Consider a line segment of length $\delta$ below the fracture surface that is perpendicular to the fracture surface. 
Fix $z \in (0,\delta)$. 
Let $\x = \x(z)$ denote the unique point on the line at a distance $z$ from the fracture surface, and let $\Omega_z$ denote the spherical cap of radius $\delta$ centered at $\x$ above the fracture surface (see Figure \ref{fig:integration}). 
For fracture to occur, all bonds connecting $\x= \x(z)$ to points $\x'$ in $ \Omega_z$ must break for every $z \in (0,\delta)$. 
The corresponding energy per unit fracture area required to break all those bonds, or the energy release rate, is given by 
\begin{equation}\label{eq:G0}
G_0 
= \int_{0}^\delta \int_{\Omega_z} w_{\text{critical}}(\x',\x(z))  \, d V_{\bm{\x}'} \, d z,
\end{equation}
where $w_{\text{critical}} = w_{\text{critical}}(\x',\x)$ is the energy density stored in the bond $\bm{\xi} = \x'-\x$ when it reaches the critical stretch or the critical energy density according to $\mu$ in \eqref{eq:mu-stretch} or \eqref{eq:mu-energy}, respectively.

\subsubsection{Critical stretch criterion.}
Consider first the critical stretch criterion. 
If a bond is at the critical stretch $s_0$, then from \eqref{eq:micropotential}, the energy density $w_{\text{critical}}$ stored in the bond   (assuming the bond has not broken) is
\begin{equation}\label{eq:w0-critical}
w_{\text{critical}}(\x',\x) = w_0(\|\bm{\xi}\|) := \frac{1}{2} \,c\, \omega(\|\bm{\xi}\|)\,s_0^2\, \|\bm{\xi}\|. 
\end{equation}
Using \eqref{eq:G0} together with  \eqref{eq:w0-critical}, we conduct a standard computation using spherical coordinates  to find (see~\cite{SillingAskari}):
\begin{align*}
G_0
	&= \int_{0}^\delta \int_0^{2\pi} \int_z^\delta \int_0^{\cos^{-1}(z/r)} w_0(r) r^2 \sin (\phi )\, d \phi \, d r \, d \theta \, dz
	= \frac{\pi cs_0^2}{2} \int_{0}^\delta \omega(r)r^4 \, dr.
\end{align*}
Recalling \eqref{eq:c-3D}, we conclude that $s_0$ can be written in terms of material constants as
\begin{equation}\label{eq:s0}
s_0^2 = \frac{4G_0}{9K}
 \frac{\displaystyle  \int_0^\delta  \omega(r)r^3 \, dr}{ \displaystyle \int_{0}^\delta \omega(r)r^4 \, dr} \quad \hbox{for a 3D problem.}
\end{equation}
For $\omega(r)\equiv 1$, this  recovers the expression in \cite[Equation (27)]{SillingAskari}. 

\subsubsection{Critical energy density criterion.}
Next, consider the critical energy density  criterion. 
In this case, $w_{\text{critical}}\equiv w_c$ is constant.
Using $w_c$ in \eqref{eq:G0}, we apply a standard computation using spherical coordinates  to find (see~\cite{FosterSillingChen}):
\begin{align*}
G_0
	&=\int_{0}^\delta \int_0^{2\pi} \int_z^\delta \int_0^{\cos^{-1}(z/r)} w_c\, r^2 \sin (\phi) \, d \phi \, d r \, d \theta \, dz
	= \frac{\pi w_c \delta^4}{4}.
\end{align*}
Therefore, $w_c$ can be written in terms of material constants as
\begin{equation}\label{eq:wc}
w_c = \frac{4G_0}{\pi \delta^4} \quad \hbox{for a 3D problem.} 
\end{equation}

In the next lemma, we write the critical energy density criterion equivalently as a bond-dependent critical stretch criterion in which the critical stretch depends only on the  bond length.

\begin{lem}[3D critical energy density criterion as a critical stretch criterion]\label{lem:energy-to-stretch}
Given the pairwise force function \eqref{eq:f-PMB}, 
a bond $\bm{\xi} = \x'-\x \in \H \setminus \{\bm0\}$, and a time $t>0$, 
it holds that 
\begin{equation}\label{eq:energy-to-stretch}
w(\bm{\eta}, \x' ,\x,t)< w_c \quad \hbox{if and only if} \quad
s^2(\bm{\eta}, \bm{\xi}) < s_c^2(\|\bm{\xi}\|),
\end{equation}
 where 
\begin{equation}\label{eq:sc}
s_c^2(\|\bm{\xi}\|) 
	= \frac{16G_0 }{9K \delta^4 \omega(\|\bm{\xi}\|)\|\bm{\xi}\|} \int_0^\delta \omega (r) r^3 \, dr. 
\end{equation}
Consequently, $\mu_c$ in \eqref{eq:mu-energy} can be written as
\begin{equation}\label{eq:mu-energy2}
\mu_c(\x' ,\x,t) = \begin{cases}
1 & \hbox{if}~s^2(\widetilde{\bm{\eta}}, \bm{\xi}) < s_c^2(\|\bm{\xi}\|)~\hbox{for all}~0 \leq \tilde{t} \leq t, \\
0 & \hbox{otherwise}. 
\end{cases}
\end{equation}
\end{lem}

\begin{proof}
Fix a bond $\bm{\xi} \in \H \setminus \{\bm0\}$ and a time $t>0$, and
recall that $\omega(\|\bm{\xi}\|) >0$. Consider the condition 
\begin{equation}\label{eq:w-wc-inequality}
w(\bm{\eta}, \x',\x,t) < w_c.
\end{equation}
Recalling  \eqref{eq:micropotential}, \eqref{eq:c-3D}, \eqref{eq:wc}, and that $\mu=1$ for an intact bond, we have that the inequality \eqref{eq:w-wc-inequality} holds if and only if
\begin{align*}
\frac{1}{2}  \frac{9K}{\displaystyle 2 \pi \int_0^\delta  \omega(r)r^3 \, dr} \omega(\|\bm{\xi}\|) s^2(\bm{\eta},\bm{\xi})  \|\bm{\xi}\| < \frac{4G_0}{\pi \delta^4}. 
\end{align*}
Equivalently,
\[
s^2(\bm{\eta},\bm{\xi}) <  \frac{16 G_0}{9K \delta^4 \omega(\|\bm{\xi}\|)  \|\bm{\xi}\|}  \int_0^\delta  \omega(r)r^3 \, dr  = s_c^2(\|\bm{\xi}\|). 
\]
\end{proof}

\begin{rem}
From \eqref{eq:sc} and \eqref{eq:s0}, the quantities $s_c$ and $s_0$ are related in the following way:
\begin{equation}\label{eq:scs0}
s_c^2(\|\bm{\xi}\|)
	= \frac{4}{\delta^4  \omega(\|\bm{\xi}\|)\|\bm{\xi}\|} \left( \frac{4G_0 }{9K} \int_0^\delta \omega (r) r^3 \, dr\right)
	=\left( \frac{4}{\delta^4  \omega(\|\bm{\xi}\|)\|\bm{\xi}\|} \displaystyle \int_{0}^\delta \omega(r)r^4 \, dr \right) s_0^2.
\end{equation}
\end{rem}

\medskip

\subsection{Critical stretch vs.~critical energy density criteria}\label{sec:3Dcomparison}

We now present our main result. 

\begin{thm}[Nonequivalence of bond-failure criteria for 3D problems]\label{thm:main3D} 
 Given the pairwise force function \eqref{eq:f-PMB}, the bond-breaking functions $\mu_0$ and $\mu_c$ coincide if and only if 
there exists a constant $\beta>0$ such that $\omega(\|\bm\xi\|) = \beta \|\bm{\xi}\|^{-1}$ for all $\bm{\xi} \in \H \setminus \{\bm{0}\}$. 
In this case, 
the bond-breaking function $\mu$ can be written as
\[
\mu(\x',\x,t) = \mu_0(\x',\x,t) = \mu_c(\x',\x,t)
	= \begin{cases}
	1 & \hbox{if}~\displaystyle s^2(\widetilde{\bm{\eta}},\bm{\xi}) < s_0^2~\hbox{for all}~0 \leq \tilde{t} \leq t, \\[.5em]
	0 
	& \hbox{otherwise},
	\end{cases}
\]
where
\[
s_0^2 = s_c^2(\|\bm\xi\|) = \frac{16G_0 }{27K \delta}.
\]
\end{thm}

\begin{proof}
By Lemma \ref{lem:energy-to-stretch}, the critical stretch and critical energy density criteria coincide if and only if $s_0^2 = s_c^2(\|\bm{\xi}\|)$ for all $\bm{\xi} \in \H \setminus \{\bm0\}$. 
Thus, we will prove 
\[
s_0^2 = s_c^2(\|\bm{\xi}\|)~\hbox{for all}~\bm{\xi} \in \H \setminus \{\bm0\} \quad \hbox{if and only if} \quad \omega(\|\bm{\xi}\|) =  \beta \|\bm{\xi}\|^{-1}~\hbox{for all}~\bm{\xi} \in \H \setminus \{\bm0\}
\]
for some constant $\beta>0$.
Assume first that there exists some constant $\beta>0$ such that $\omega(\|\bm\xi\|) = \beta \|\bm{\xi}\|^{-1}$ for all $\bm{\xi}\in \H \setminus \{\bm{0}\}$.
From \eqref{eq:scs0}, we have
\[
s_c^2(\|\bm{\xi}\|)
	= \frac{4s_0^2}{\delta^4  \beta} \displaystyle \int_{0}^\delta \beta r^3 \, dr
	=  s_0^2   \quad \hbox{for all}~\bm{\xi} \in \H \setminus \{\bm{0}\},
\]
as desired. 
Conversely, assume that $s_0^2 = s_c^2(\|\bm{\xi}\|)$ for all $\bm{\xi} \in \H \setminus \{\bm0\}$. 
By \eqref{eq:scs0}, we have
\[
1= \frac{4}{\delta^4  \omega(\|\bm{\xi}\|)\|\bm{\xi}\|} \displaystyle \int_{0}^\delta \omega(r)r^4 \, dr \quad \hbox{for all}~\bm{\xi} \in \H \setminus \{\bm0\}.
\]
Rearranging, we write
\[
\omega(\|\bm{\xi}\|)\|\bm{\xi}\| = \frac{4}{\delta^4} \int_0^\delta \omega(r)  r^4 \, dr \quad \hbox{for all}~\bm{\xi} \in \H \setminus \{\bm0\}
\]
to find that $\omega(\|\bm{\xi}\|)\|\bm{\xi}\|$ is constant in $\H \setminus \{\bm{0}\}$, i.e.,~there exists some constant $\beta>0$ such that $\omega(\|\bm{\xi}\|) = \beta \|\bm{\xi}\|^{-1}$ for all $\bm{\xi} \in \H\setminus \{\bm{0}\}$. 

Lastly, note that if $\omega(\|\bm\xi\|) \|\bm{\xi}\| = \beta$, then from \eqref{eq:s0}, we have
\begin{align*}
s_c^2(\|\bm{\xi}\|) = s_0^2 
	&= \frac{4G_0}{9K}
 \frac{\displaystyle  \int_0^\delta \beta r^2 \, dr}{ \displaystyle \int_{0}^\delta \beta r^3 \, dr}
 	= \frac{16G_0 }{27K \delta}.
\end{align*}
\end{proof}

Theorem \ref{thm:main3D} establishes that the critical stretch and critical energy density criteria are \emph{not} equivalent in general.
For example, consider the influence function 
\begin{equation}\label{eq:alpha}
\omega(\|\bm{\xi}\|) = \|\bm{\xi}\|^{-\alpha} \quad \hbox{for  some}~\alpha \in (-\infty,1) \cup (1,4).
\end{equation}
By Theorem \ref{thm:main3D}, if $\alpha =1$, then the two criteria are equivalent, so we omit that case. 
On the other hand, by \eqref{eq:c-3D}, it must be that $\alpha < 4$. 
In this setting, one has, by \eqref{eq:s0}, 
\begin{equation}\label{eq:finite}
s_0^2 =  \frac{4G_0}{9K}
 \frac{ \displaystyle \int_0^\delta  r^{3-\alpha} \, dr}{ \displaystyle \int_{0}^\delta r^{4-\alpha} \, dr}
 	= \frac{4G_0(5-\alpha)}{9K(4-\alpha) \delta},
\end{equation}
and by \eqref{eq:scs0}, 
\begin{equation}\label{scs0-alpha}
s_c^2(\|\bm{\xi}\|)
	=\left( \frac{4}{\delta^4  \|\bm{\xi}\|^{1-\alpha}} \displaystyle \int_{0}^\delta r^{4-\alpha} \, dr \right) s_0^2
	= \frac{4 }{5-\alpha  } \left( \frac{\delta}{\|\bm{\xi}\|}\right)^{1-\alpha} s_0^2,
\end{equation}
demonstrating that $s_c^2$ and $s_0^2$ do not coincide, except for $\alpha = 1$.

Moreover,
\begin{equation}\label{eq:long-short}
s_c^2(\|\bm\xi\|) > s_0^2 \quad \hbox{if and only if} \quad
\frac{4\delta^{1-\alpha}}{5-\alpha} >\|\bm{\xi}\|^{1-\alpha}.
\end{equation}
Therefore, if $\alpha<1$, then, from  \eqref{eq:long-short}, we see that ``sufficiently short'' bonds break first under the critical stretch criterion and ``sufficiently long'' bonds break first under the critical energy density criterion. 
On the other hand, if $1 < \alpha < 4$, then we see from \eqref{eq:long-short} that ``sufficiently long'' bonds break first under the critical stretch criterion and ``sufficiently short'' bonds break first under the critical energy density criterion. 
With these ingredients, we make precise Theorem \ref{thm:intro2} for 3D problems: 

\begin{thm}[Bond-breaking behavior for 3D problems]\label{thm:3Dbehavior}
Consider the pairwise force function \eqref{eq:f-PMB} with influence function \eqref{eq:alpha}. 
\begin{enumerate}
\item Assume that $\alpha <1$. A  bond $\bm{\xi} = \x'-\x$ satisfying
\[
\|\bm\xi\|< \left(\frac{4}{5-\alpha}\right)^{\frac{1}{1-\alpha}}\delta
\]
and $s^2(\bm{\eta},\bm{\xi}) = s_0^2$ at time $t>0$ breaks at time $t$ under the critical stretch criterion but does not break under the critical energy density criterion.
On the other hand, 
a bond satisfying
\[
\left(\frac{4}{5-\alpha}\right)^{\frac{1}{1-\alpha}} \delta <\|\bm{\xi}\| \leq \delta
\]
and $s^2(\bm{\eta},\bm{\xi})=s_c^2(\|\bm\xi\|)$ at time $t>0$ breaks at time $t$ under the critical energy density criterion but does not break under the critical stretch criterion. 

\item Assume that $1 < \alpha < 4$. A bond $\bm{\xi} = \x'-\x$ satisfying
\[
\left(\frac{5-\alpha}{4}\right)^{\frac{1}{\alpha-1}} \delta <\|\bm{\xi}\|  \leq \delta
\]
and $s^2(\bm{\eta},\bm{\xi}) = s_0^2$ at time $t>0$ breaks at time $t$ under the critical stretch criterion but does not break under the critical energy density criterion.
On the other hand, 
a bond satisfying
\[
\|\bm{\xi}\| <\left(\frac{5-\alpha}{4}\right)^{\frac{1}{\alpha-1}} \delta
\]
and $s^2(\bm{\eta},\bm{\xi})=s_c^2(\|\bm\xi\|)$ at time $t>0$ breaks at time $t$ under the critical energy density criterion but does not break under the critical stretch criterion. 

\end{enumerate}
\end{thm}

\subsubsection{Isotropic extension}
  
Consider still the influence function \eqref{eq:alpha}, and now let us assume an isotropic extension. 
Recalling Remark \ref{rem:isotropic},  
$s \equiv \bar{s}$ for some constant $\bar{s}>0$.
Under the critical stretch criterion, all bonds $\bm{\xi} \in \H \setminus \{\bm0\}$ break as long as $\bar{s} \geq s_0$. 
On the other hand, under the critical energy density criterion,  a bond $\bm{\xi}$ breaks when
$
\bar{s}^2 \geq s_c^2(\|\bm{\xi}\|) = 4 s_0^2 \delta^{1-\alpha}/((5-\alpha)\|\bm{\xi}\|^{1-\alpha})
$
(see~\eqref{scs0-alpha})
or, equivalently,
\begin{equation}\label{eq:bond-bound}
\|\bm{\xi}\|^{1-\alpha} \geq \frac{4 \delta^{1-\alpha} s_0^2}{(5-\alpha) \bar{s}^2}.
\end{equation}
Note that for the longest bonds (i.e., bonds $\bm{\xi} \in \H \setminus \{\bm0\}$ such that $\|\bm{\xi}\|=\delta$) to break,  
$\bar{s}$ must satisfy (see \eqref{eq:finite})
\begin{equation}\label{eq:bars-lower}
\bar{s}^2 \geq \frac{4 s_0^2}{5-\alpha} = \frac{16G_0}{9K(4-\alpha) \delta}.
\end{equation}
Therefore, we conclude the following: 
\begin{enumerate}
\item Assume $\alpha <1$. 
As long as \eqref{eq:bars-lower} is satisfied, 
we use \eqref{eq:bond-bound} to conclude that all bonds $\bm{\xi}\in \H \setminus \{\bm0\}$ in the spherical shell
\[
\left\{ \bm{\xi} :  \left( \frac{4s_0^2}{(5-\alpha) \bar{s}^2}\right)^{ \frac{1}{1-\alpha}} \delta \leq \|\bm{\xi}\| \leq \delta \right\}
\]
break under the critical energy density criterion.  
That is, only sufficiently long bonds break. 
\item Assume $1 < \alpha < 4$. 
We use \eqref{eq:bond-bound} to conclude that all bonds $\bm{\xi}\in \H \setminus \{\bm0\}$ 
in the closed ball
\[
\left\{ \bm\xi : \|\bm{\xi}\| \leq \left(\frac{(5-\alpha) \bar{s}^2}{4  s_0^2}\right)^{\frac{1}{\alpha-1}} \delta\right\}
\]
break under the critical energy density criterion.  
That is, only sufficiently short bonds break. Moreover, for all bonds $\bm{\xi}\in \H \setminus \{\bm0\}$ to break, $\bar{s}$ must satisfy \eqref{eq:bars-lower}. 
\end{enumerate}
 
\section{Bond-failure criteria for two-dimensional problems}\label{sec:2D}

In this section, we analyze the bond-failure criteria as in Section \ref{sec:3D}, but for two-dimensional problems.

\subsection{Bond-failure criteria}\label{sec:criteria2D}

As in Section \ref{sec:criteria3D}, we will use the energy release rate to determine the critical parameters $s_0$ and $w_c$  in \eqref{eq:mu-stretch} and \eqref{eq:mu-energy}, respectively, in terms of material constants, but now for two-dimensional problems in a state of plane strain or plane stress. In analogy to~\eqref{eq:G0}, the energy per unit fracture area for a two-dimensional problem can be computed as
\begin{equation}\label{eq:G02D}
G_0 
= h \int_{0}^\delta \int_{\Omega_z} w_{\text{critical}}(\x',\x(z))  \, d A_{\bm{\x}'} \, d z,
\end{equation}
where $\Omega_z$ now represents a circular segment (see Figure~\ref{fig:integration}).

\subsubsection{Critical stretch criterion}

Assume first the critical stretch criterion, namely \eqref{eq:mu-stretch}. 
Using~\eqref{eq:G02D} together with~\eqref{eq:w0-critical}, we conduct a standard computation using polar coordinates (see \cite[Appendix C.2]{SelesonEtAl}) to get:
\begin{align*}
G_0
	&= h \int_{0}^\delta  \int_z^\delta \int_{-\cos^{-1}(z/r)}^{\cos^{-1}(z/r)} w_0(r) r \, d \theta \, d r \, dz
	= hcs_0^2  \int_{0}^\delta   \omega(r) r^3  \, d r.
\end{align*}
Recalling \eqref{eq:c-planestrain} and \eqref{eq:c-planestress}, we find
\begin{equation}\label{eq:s0-2D}
s_0^2 
= \frac{G_0}{hc \displaystyle \int_0^\delta \omega(r) r^3 \, dr}
= \begin{cases}
\displaystyle\frac{5 \pi G_0}{16E} \frac{ \displaystyle \int_0^\delta \omega(r) r^2 \, dr}{ \displaystyle \int_0^\delta  \omega(r)r^3 \, dr }
& \hbox{for a 2D problem in a state of plane strain,}\\[2.5em]
\displaystyle\frac{\pi G_0}{3E} \frac{ \displaystyle  \int_0^\delta  \omega(r)r^2 \, dr}{\displaystyle \int_0^\delta \omega(r) r^3 \, dr}
 & \hbox{for a 2D problem in a state of plane stress}.
\end{cases}
\end{equation}

\subsubsection{Critical energy density criterion}

Assume now the critical energy density criterion, namely \eqref{eq:mu-energy}. 
Using $w_{\text{critical}} \equiv w_c$ in~\eqref{eq:G02D}, we apply a standard computation using polar coordinates to find:
\begin{align*}
G_0
	&= h  \int_{0}^\delta  \int_z^\delta \int_{-\cos^{-1}(z/r)}^{\cos^{-1}(z/r)}  w_c \,r \, d \theta \, d r \, dz
	= \frac{2 h w_c \delta^3}{3}. 
\end{align*}
Therefore, we have 
\begin{equation}\label{eq:wc-2D}
w_c = \frac{3G_0}{2 h \delta^3} \quad \hbox{for a 2D problem.}
\end{equation}
Note that the critical energy density is the same for plane strain and plane stress. 

\begin{rem}\label{rem: h independence of critical stretch}
The critical energy density in~\eqref{eq:wc-2D} depends on $h$, in contrast to the critical stretch, which does not depend on $h$ (see~\eqref{eq:s0-2D}). Nevertheless, recasting the critical energy density criterion as a critical stretch criterion removes the dependence on~$h$ (see~Lemma~\ref{lem:energy-to-stretch2D}).
\end{rem}

The following result, analogous to Lemma \ref{lem:energy-to-stretch} for two-dimensional problems, holds. 

\begin{lem}[2D critical energy density criterion as a critical stretch criterion]
\label{lem:energy-to-stretch2D}
Consider a two-dimensional problem involving a body in a state of plane strain or plane stress.
Given the pairwise force function \eqref{eq:f-PMB},  
a bond $\bm{\xi} = \x'-\x \in \H \setminus \{\bm0\}$, and a time $t>0$, 
it holds that 
\[
w(\bm{\eta}, \x',\x,t) < w_c \quad \hbox{if and only if} \quad
s^2(\bm{\eta}, \bm{\xi}) < s_c^2(\|\bm{\xi}\|), 
\]
 where 
\begin{equation}\label{eq:sc-2D}
s_c^2(\|\bm{\xi}\|) 
= \begin{cases}
\displaystyle \frac{15 \pi G_0}{16E\delta^3 \omega(\|\bm{\xi}\|)  \|\bm{\xi}\| }   \int_0^\delta  \omega(r)r^2 \, dr 
& \hbox{for a 2D problem in a state of plane strain,}\\[1.5em]
\displaystyle\frac{\pi G_0}{E\delta^3 \omega(\|\bm{\xi}\|)  \|\bm{\xi}\| } \displaystyle  \int_0^\delta  \omega(r)r^2 \, dr
 & \hbox{for a 2D problem in a state of plane stress}.
\end{cases}
\end{equation}
\end{lem}

\begin{proof}
Recalling the proof of Lemma \ref{lem:energy-to-stretch}, it is enough to note that, for a fixed bond $\bm{\xi} \in \H\setminus\{\bm0\}$ and a fixed time $t > 0$,  by \eqref{eq:micropotential} and \eqref{eq:wc-2D} (recall  that $\mu = 1$ for an intact bond and $\omega(\|\bm{\xi}\|)>0$),
\[
w(\bm{\eta}, \x',\x,t) < w_c 
\]
holds if and only if 
\[
\frac{1}{2} c \, \omega(\|\bm{\xi}\|)\, s^2(\bm{\eta}, \bm{\xi}) \|\bm{\xi}\| <  \frac{3G_0}{2h \delta^3}, 
\]
where $c$ is given in \eqref{eq:c-planestrain} for plane strain and \eqref{eq:c-planestress} for plane stress.
Equivalently,
\begin{equation}\label{eq:sc-in-proof}
s^2(\bm{\eta}, \bm{\xi})< \frac{3G_0}{hc\delta^3 \omega(\|\bm{\xi}\|)  \|\bm{\xi}\| } = s_c^2(\|\bm{\xi}\|).
\end{equation}
The expression \eqref{eq:sc-2D} follows from \eqref{eq:sc-in-proof} together with \eqref{eq:c-planestrain} and \eqref{eq:c-planestress}. 
\end{proof}

\begin{rem}
Note from \eqref{eq:sc-in-proof} and \eqref{eq:s0-2D} that 
\begin{equation}\label{eq:scs0-2D}
s_c^2(\|\bm{\xi}\|)
	=\frac{3}{\delta^3 \omega(\|\bm{\xi}\|)  \|\bm{\xi}\| } \frac{G_0}{hc}
	=\left(\frac{3}{\delta^3 \omega(\|\bm{\xi}\|)  \|\bm{\xi}\| }  \displaystyle \int_0^\delta \omega(r) r^3 \, dr\right)s_0^2.
\end{equation}
Comparing \eqref{eq:scs0} and \eqref{eq:scs0-2D}, we conclude:
\[
s_c^2(\|\bm{\xi}\|)
	=\left(\frac{d+1}{\delta^{d+1} \omega(\|\bm{\xi}\|)  \|\bm{\xi}\| }  \displaystyle \int_0^\delta \omega(r) r^{d+1} \, dr\right)s_0^2, \quad \hbox{for}~d=2,3.
\]
\end{rem}

\subsection{Critical stretch vs.~critical energy density criteria}\label{sec:2Dcomparision}

We have the following two-dimensional analogue of Theorem \ref{thm:main3D}.

\begin{thm}[Nonequivalence of bond-failure criteria for 2D problems]\label{thm:main2D} 
Consider a two-dimensional problem involving a body in a state of plane strain or plane stress. 
Given the pairwise force function \eqref{eq:f-PMB}, 
the bond-breaking functions $\mu_0$ and $\mu_c$ coincide
 if and only if 
there exists a constant $\beta>0$ such that $\omega(\|\bm\xi\|) = \beta \|\bm{\xi}\|^{-1}$ for all $\bm\xi \in \H \setminus \{\bm0\}$. 
In this case, 
the bond-breaking function $\mu$ can be written as
\[
\mu(\x',\x,t) = \mu_0(\x',\x,t) = \mu_c(\x',\x,t)
	= \begin{cases}
	1 & \hbox{if}~\displaystyle s^2(\widetilde{\bm{\eta}},\bm{\xi}) < s_0^2~\hbox{for all}~0 \leq \tilde{t} \leq t,\\[.5em]
	0 
	& \hbox{otherwise}, 
	\end{cases}
\]
where
\[
s_0^2= s_c^2(\|\bm\xi\|)  = \begin{cases}
\displaystyle \frac{15 \pi G_0}{32E\delta}
& \hbox{for a 2D problem in a state of plane strain,}\\[1.5em]
\displaystyle \frac{\pi G_0}{2E \delta}
 & \hbox{for a 2D problem in a state of plane stress}.
 \end{cases}
\]
\end{thm}

\begin{proof}
By Lemma \ref{lem:energy-to-stretch2D}, the critical stretch and critical energy density criteria coincide if and only if $s_0^2 = s_c^2(\|\bm{\xi}\|)$ for all $\bm{\xi} \in \H \setminus \{\bm0\}$.
Recalling the proof of Theorem \ref{thm:main3D}, 
it is enough to observe that, if there is a $\beta>0$ such that $\omega(\|\bm{\xi}\|) = \beta \|\bm{\xi}\|^{-1}$ for all $\bm{\xi} \in \H\setminus \{\bm0\}$, then from~\eqref{eq:scs0-2D}
\[
s_c^2(\|\bm{\xi}\|)
	=\left(\frac{3}{\delta^3 \beta}  \displaystyle \int_0^\delta \beta r^2 \, dr\right)s_0^2 = s_0^2 
\]
and from \eqref{eq:s0-2D}
\[
s_0^2 
=\begin{cases}
\displaystyle\frac{5 \pi G_0}{16E} \frac{ \displaystyle \int_0^\delta \beta r \, dr}{ \displaystyle \int_0^\delta \beta r^2 \, dr }
& \hbox{for plane strain,}\\[2.5em]
\displaystyle\frac{\pi G_0}{3E} \frac{ \displaystyle  \int_0^\delta \beta r \, dr}{\displaystyle \int_0^\delta \beta r^2 \, dr}
 & \hbox{for plane stress}
 \end{cases}
 =\begin{cases}
\displaystyle\frac{15 \pi G_0}{32E\delta} 
& \hbox{for plane strain,}\\[1.5em]
\displaystyle\frac{\pi G_0}{2E \delta}
 & \hbox{for plane stress}.
 \end{cases}
\]
\end{proof}

As observed in the previous section, Theorem \ref{thm:main2D} implies that the critical stretch and critical energy density criteria are \emph{not}  equivalent in general. 
Similar to Section~\ref{sec:3Dcomparison}, consider the influence function
\begin{equation}\label{eq:alpha-2D}
\omega(\|\bm{\xi}\|) = \|\bm{\xi}\|^{-\alpha} \quad \hbox{for  some}~\alpha \in (-\infty,1) \cup (1,3).
\end{equation}
By Theorem \ref{thm:main2D}, if $\alpha =1$, then the two criteria are equivalent, so we omit that case. On the other hard, by \eqref{eq:c-planestrain} and \eqref{eq:c-planestress}, it must be that $\alpha < 3$. 
In this setting, by \eqref{eq:s0-2D}, one can write
\begin{equation}\label{eq: s0_2D}
s_0^2 = \begin{cases}
\displaystyle\frac{5 \pi G_0(4-\alpha)}{16E(3-\alpha)\delta} 
& \hbox{for a 2D problem in a state of plane strain,}\\[1.5em]
\displaystyle\frac{\pi G_0(4-\alpha)}{3E(3-\alpha)\delta} 
 & \hbox{for a 2D problem in a state of plane stress},
\end{cases}
\end{equation}
and by \eqref{eq:scs0-2D}, 
\begin{equation}\label{eq:omega1-2D}
s_c^2(\|\bm{\xi}\|)
	= \left(\frac{3}{\delta^3 \|\bm{\xi}\|^{1-\alpha} }  \displaystyle \int_0^\delta  r^{3-\alpha} \, dr \right) s_0^2
	= \frac{3}{4-\alpha}  \left( \frac{\delta}{\|\bm{\xi}\|}\right)^{1-\alpha}s_0^2,
\end{equation}
demonstrating that, also in the two-dimensional case, $s_c^2$ and $s_0^2$ do not coincide, except for $\alpha =1$.

Moreover, 
\begin{equation}\label{eq:long-short-2D}
s_c^2(\|\bm\xi\|) > s_0^2 \quad \hbox{if and only if} \quad
\frac{3 \delta^{1-\alpha} }{4-\alpha  } > \|\bm{\xi}\|^{1-\alpha}.
\end{equation}
Recalling the discussion before Theorem \ref{thm:3Dbehavior} for 3D problems, we draw similar conclusions from \eqref{eq:long-short-2D} for 2D problems. In particular, we make precise Theorem \ref{thm:intro2} for 2D problems: 

\begin{thm}[Bond-breaking behavior for 2D problems]\label{thm:2Dbehavior}
Consider the pairwise force function \eqref{eq:f-PMB} with influence function \eqref{eq:alpha-2D} for a two-dimensional problem involving a body in a state of plane strain or plane stress. 
\begin{enumerate}
\item Assume that $\alpha <1$. A bond $\bm{\xi} \in \H\setminus \{\bm0\}$ satisfying
\[
\|\bm\xi\|< \left(\frac{3}{4-\alpha}\right)^{\frac{1}{1-\alpha}}\delta
\]
and $s^2(\bm{\eta},\bm{\xi}) = s_0^2$ at time $t>0$ breaks at time $t$ under the critical stretch criterion but does not break under the critical energy density criterion.
On the other hand, a bond $\bm{\xi} \in \H\setminus \{\bm0\}$ satisfying
\[
\left(\frac{3}{4-\alpha}\right)^{\frac{1}{1-\alpha}} \delta <\|\bm{\xi}\| \leq \delta
\]
and $s^2(\bm{\eta},\bm{\xi}) = s_c^2(\|\bm{\xi}\|)$ at time $t>0$ breaks at time $t$ under the critical  energy density criterion but does not break under the critical stretch criterion.
\item Assume $1 < \alpha < 3$. A bond $\bm{\xi} \in \H\setminus \{\bm0\}$ satisfying
\[
\left(\frac{4-\alpha}{3}\right)^{\frac{1}{\alpha-1}} \delta <\|\bm{\xi}\|  \leq \delta
\]
and $s^2(\bm{\eta},\bm{\xi}) = s_0^2$ at time $t>0$ breaks at time $t$ under the critical stretch criterion but does not break under the critical energy density criterion.
On the other hand, a bond $\bm{\xi} \in \H\setminus \{\bm0\}$ satisfying
\[
\|\bm{\xi}\| <\left(\frac{4-\alpha}{3}\right)^{\frac{1}{\alpha-1}} \delta
\]
and $s^2(\bm{\eta},\bm{\xi}) = s_c^2(\|\bm{\xi}\|)$ at time $t>0$ breaks at time $t$ under the critical  energy  density criterion but does not break under the critical stretch criterion
\end{enumerate}
\end{thm}

\subsubsection{Isotropic extension}\label{sec:2Diso}

Consider still the influence function \eqref{eq:alpha-2D}, and now let us assume
an isotropic extension.
Recalling Remark \ref{rem:isotropic},  
$s \equiv \bar{s}$ for some constant $\bar{s}>0$. 
Under the critical stretch criterion, all bonds $\bm{\xi} \in \H \setminus \{\bm0\}$ break as long as $\bar{s} \geq s_0$.
On the other hand, under the critical energy density criterion, a bond $\bm{\xi}$ breaks when
$\bar{s}^2 \geq s_c^2(\|\bm{\xi}\|) = 3 s_0^2 \delta^{1-\alpha} /((4-\alpha) \|\bm{\xi}\|^{1-\alpha})$ (see \eqref{eq:omega1-2D}) or, equivalently,
\begin{equation}\label{eq:iso-break-2D}
\|\bm{\xi}\|^{1-\alpha} \geq \frac{3\delta^{1-\alpha} s_0^2}{(4-\alpha) \bar{s}^2}.
\end{equation}
Note that for the longest bonds (i.e., bonds $\bm{\xi} \in \H \setminus \{\bm0\}$ such that $\|\bm{\xi}\|=\delta$) 
 to break, $\bar{s}$ must satisfy (see~\eqref{eq: s0_2D})
\begin{equation}\label{eq:bars-2D}
\bar{s}^2 \geq \frac{3s_0^2}{4-\alpha} =
\begin{cases}
\displaystyle\frac{15 \pi G_0}{16E(3-\alpha)\delta} 
& \hbox{for a 2D problem in a state of plane strain,}\\[1.5em]
\displaystyle\frac{\pi G_0}{E(3-\alpha)\delta} 
 & \hbox{for a 2D problem in a state of plane stress}.
\end{cases}
\end{equation}
Therefore, we conclude the following:
\begin{enumerate}
\item Assume  $\alpha <1$. As long as~\eqref{eq:bars-2D} is satisfied, we use~\eqref{eq:iso-break-2D} to conclude that all bonds $\bm{\xi} \in \H \setminus \{\bm0\}$ in the annulus
\[
\left\{\bm{\xi}: \left(\frac{3 s_0^2}{(4-\alpha) \bar{s}^2} \right)^{\frac{1}{1-\alpha}} \delta  \leq \|\bm{\xi}\| \leq \delta \right\}
\]
break under the critical energy density criterion. That is, only sufficiently long bonds
break.  
\item Assume $1 < \alpha < 3$. We use \eqref{eq:iso-break-2D} to conclude that all bonds $\bm{\xi} \in \H \setminus \{\bm0\}$ in the closed disc
\[
\left\{\bm{\xi}:  \|\bm{\xi}\| \leq \left(\frac{(4-\alpha) \bar{s}^2}{3 s_0^2} \right)^{\frac{1}{\alpha-1}}\delta \right\}
\]
break under the critical energy density  criterion. That is, only sufficiently short bonds break.
Moreover, for all bonds $\bm{\xi} \in \H \setminus \{\bm0\}$ to break, $\bar{s}$ must satisfy \eqref{eq:bars-2D}.  
\end{enumerate}
 This behavior will be observed numerically in Section \ref{sec:isotropicextension}.

\section{Numerical examples}\label{sec: numerical examples}

We now present three numerical examples in two dimensions. 
The material model for each example is the GPMB model given by~\eqref{eq:f-PMB} with the influence function $\omega(\|\bm{\xi}\|) = 1/\|\bm{\xi}\|^\alpha$, where $\alpha < 3$. 

\subsection{Example 1: Isotropic extension} 
\label{sec:isotropicextension}

In this example, we demonstrate the difference between the critical stretch and critical energy density  criteria, including the effect of the choice of influence function on bond breaking, using a simple scenario with an imposed displacement field. 

We consider a two-dimensional square plate with domain $\Omega = (-0.5,0.5)\times(-0.5,0.5)$ 
and a horizon of ${\delta = 0.2}$. The  expressions for the elastic and bond-breaking parameters, derived in Section~\ref{sec:2Dc} (see~\eqref{eq:c-planestrain} and~\eqref{eq:c-planestress}) and Section~\ref{sec:2Dcomparision} (see~\eqref{eq:s0-2D}, \eqref{eq:wc-2D}, and~\eqref{eq:omega1-2D}), respectively,
are summarized 
 in Table~\ref{table: Example 1 expressions}. 

\begin{table}[h!]
\centering
{\tabulinesep=1.2mm
\begin{tabu}{||c | c | c | c||} 
 \hline 
  $c$ & $s_0^2$ & $w_c$ & $s_c^2(\|\bm{\xi}\|)$ \\ [0.5ex] 
     \hline\hline
      $\begin{array}{cc}
          \displaystyle \frac{16(3-\alpha)E}{5 \pi h \delta^{3-\alpha}} & \mbox{plane strain}\\
          \vspace*{-0.02in}\\
          \displaystyle \frac{3(3-\alpha)E}{\pi h \delta^{3-\alpha}} & \mbox{plane stress}\\
          \end{array}
          $  
      & $ \displaystyle
      \begin{array}{cc}
      	\displaystyle 
	 \frac{5(4-\alpha)\pi G_0}{16(3-\alpha)E \delta}
          & \mbox{plane strain}\\
          \vspace*{-0.02in}\\
         \displaystyle
          \frac{(4-\alpha)\pi G_0}{\displaystyle 3(3-\alpha)E\delta}
         & \mbox{plane stress}\\
          \end{array}
          $ 
      & $\displaystyle\frac{3G_0}{2h\delta^3}$ 
      &$\displaystyle
       \frac{3}{4-\alpha} \left(\frac{\delta}{\|\bm{\xi}\|}\right)^{1-\alpha} s_0^2$\\  
 \hline
\end{tabu}
}
\caption{
Elastic and bond-breaking parameters for $\omega(\|\bm{\xi}\|) = 1/\|\boldsymbol{\xi}\|^\alpha$ in 2D.}
\label{table: Example 1 expressions}
\end{table}

 The plate is subjected to an isotropic extension given by $\u(\x) = \bar{s} \x$ with $\bar{s}$ constant, resulting in a constant stretch $s \equiv \bar{s}$ for every bond $\bm{\xi}$ (see Remark~\ref{rem:isotropic}).
We assume a critical stretch $s_0 = 0.25$, which implies that bonds break under the critical stretch criterion when their extension reaches $25\%$ of their original (undeformed) length. We  compare numerically the plate response for varying imposed stretch $\bar{s}$ and for different choices of the influence function parameter $\alpha$ using the critical energy density criterion. As proved in Theorem \ref{thm:main2D}, the case $\alpha = 1$ coincides with the critical stretch criterion, i.e., $s_c(\|\bm{\xi}\|)=s_0$.

The computations employ the PDMATLAB2D code~\cite{SelesonEtAl}, using the meshfree discretization from~\cite{SillingAskari} with a uniform grid spacing of  $\Delta x = \Delta y = \delta/3$, which results in a 
$15\times 15$ grid of computational nodes
and an $m$-ratio of $m := \delta/\Delta x = 3$. 
A relatively coarse discretization is employed in this example to enable clear visualization of results (including computational nodes and  bonds).  
In this setting, each computational node in the bulk of the plate has $28$ neighbors, as illustrated in Figure~\ref{fig:Example1 neighborhood}.

\begin{figure}[htbp!]
\centering
\begin{subfigure}{.48\textwidth}
    \centering
    \includegraphics[width=\linewidth, trim=3cm 1.5cm 3cm 1.5cm,clip]{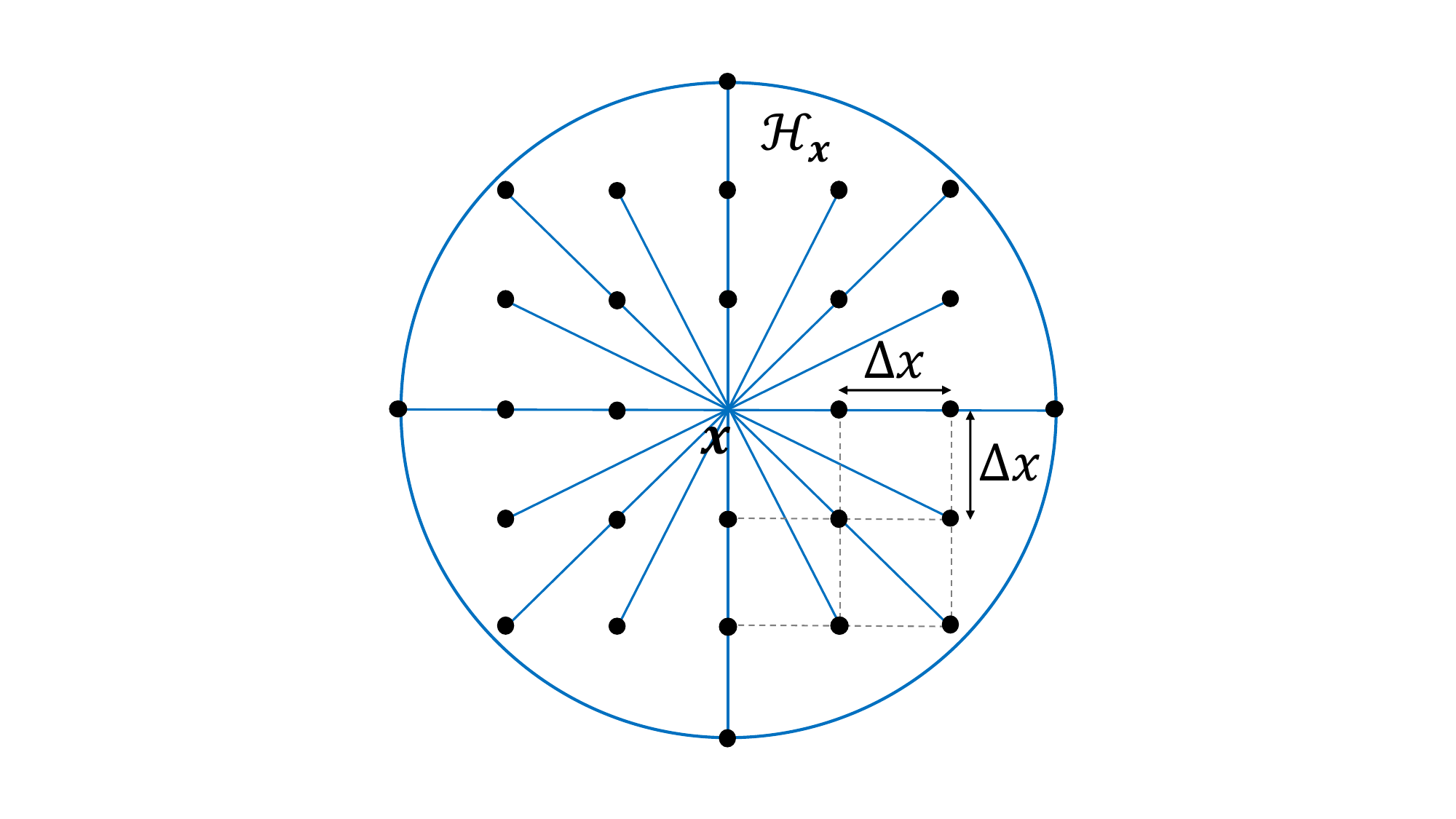}
    \medskip
    \caption{Illustration of the 2D neighborhood $\mathcal{H}_{\x}$ of a computational node at ${\x}$.}
    \label{fig:Example1 neighborhood}
   \end{subfigure}
   \hfill
   \begin{subfigure}{.47\textwidth}
   \centering
{\tabulinesep=1.2mm
\begin{tabu}{|c || c | c | c | c||} 
 \hline
   $ \|\boldsymbol{\xi}\|$   & $\alpha = 0$ & $\alpha = 1$ & $\alpha = 2$  \\ [0.5ex] 
 \hline\hline
   $\Delta x$ & $\frac{9}{4}s_0^2$ & $s_0^2$ & $\frac{1}{2}s_0^2$ \\
    $\sqrt{2}\Delta x$ & $\frac{9}{4\sqrt{2}}s_0^2$ & $s_0^2$ & $\frac{\sqrt{2}}{2}s_0^2$ \\
  $2\Delta x$     & $\frac{9}{8}s_0^2$ & $s_0^2$ &$s_0^2$  \\
   $\sqrt{5}\Delta x$ & $\frac{9}{4\sqrt{5}}s_0^2$ & $s_0^2$ & $\frac{\sqrt{5}}{2}s_0^2$ \\ 
   $2\sqrt{2}\Delta x$ & $\frac{9}{8\sqrt{2}}s_0^2$  & $s_0^2$ &  $\sqrt{2}s_0^2$ \\ 
    $3\Delta x$ & $\frac{3}{4}s_0^2$  & $s_0^2$ & $\frac{3}{2}s_0^2$ \\ 
 \hline
\end{tabu}
}
\caption{Expressions for $s_c^2(\|\bm{\xi}\|)$ as functions of $s_0^2$ for  $\alpha = 0,1,2$ and for the different bond lengths.}
\label{table: Example 1 critical stretch from energy}
      \end{subfigure}
 \caption{Uniform discretization with a grid spacing of $\Delta x = \Delta y = \delta/3$ in Example 1.}
\end{figure}

We note from Figure~\ref{fig:Example1 neighborhood} that, for the chosen discretization, a full neighborhood contains bonds with lengths 
$\|\boldsymbol{\xi}\| = \Delta x,\sqrt{2}\Delta x,2\Delta x,\sqrt{5}\Delta x,2\sqrt{2}\Delta x,$ or $3\Delta x$.
In Figure~\ref{table: Example 1 critical stretch from energy}, we present the expressions for $s_c^2(\|\bm{\xi}\|)$ for each of these bond lengths as functions of $s_0^2$ ({\em cf.}~Table~\ref{table: Example 1 expressions}) 
for $\delta = 3\Delta x$ and different values of~$\alpha$. 
We observe that for $\alpha = 1$, $s_c$ remains constant regardless of the bond length; for $\alpha = 0$, $s_c$ decreases as the bond length increases; and for $\alpha = 2$,  $s_c$ increases  as the bond length increases.
This confirms our analytical results in Section~\ref{sec:2Diso}. 
Indeed, for an isotropic extension with imposed stretch $\bar{s}$ of increasing value, under the critical energy density criterion, longer bonds break first for $\alpha = 0$, whereas shorter bonds break first for $\alpha = 2$. 
In contrast, for $\alpha = 1$, all bonds break at the same 
 critical stretch.
 
In Figure~\ref{Fig: Example 1 results}, we present numerical results demonstrating these different bond-breaking patterns. Specific values of imposed stretch $\bar{s}$ were selected based on Figure~\ref{table: Example 1 critical stretch from energy}. All the computational nodes and bonds are plotted in cyan. 
 For better visualization, we also select a given source node and draw its neighborhood, including its boundary (blue circle) and the intact bonds (blue lines) and corresponding neighboring nodes (black filled circles). 
Examining the results from the top row to the bottom row, the computations confirm that longer bonds break first for $\alpha = 0$ (left column), whereas shorter bonds break first for $\alpha = 2$ (right column), as proved.

\def\width{0.55}
\def\wwidth{0.54}

\begin{figure*}[htbp!]
    \centering
    \begin{subfigure}[t]{0.32\textwidth}
        \centering
       \boldsymbol{$\alpha = 0$}\\[0.1in]
        \includegraphics[width=\width\textwidth, trim=22cm 0cm 26cm 2cm,clip]{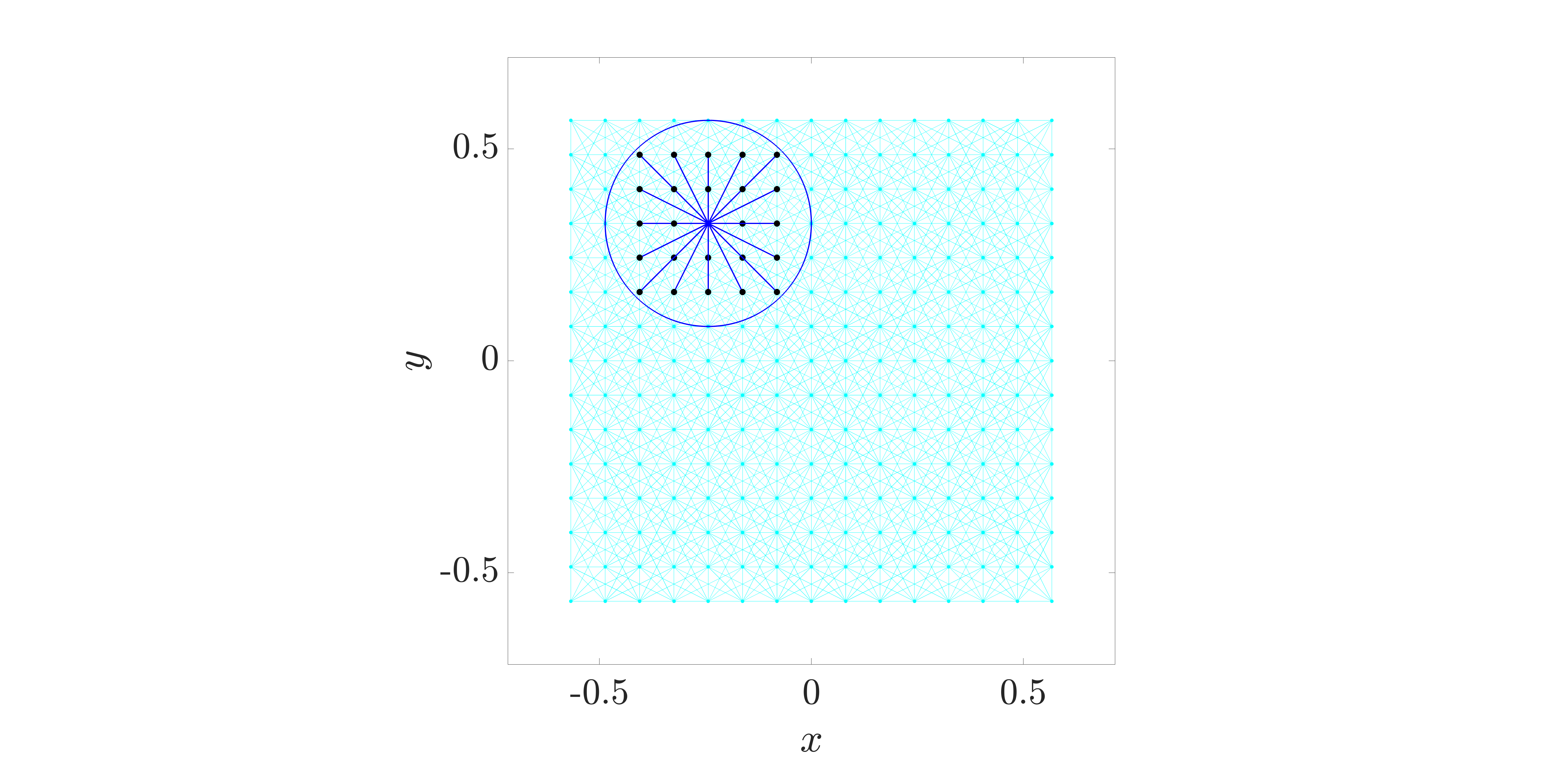}
        \caption{$\bar{s}^2 = \frac{3}{4}s_0^2$}
    \end{subfigure}%
     \vspace{0.5pc}
    \begin{subfigure}[t]{0.32\textwidth}
        \centering
       \boldsymbol{$\alpha = 1$}\\[0.1in]
    \end{subfigure}
    \begin{subfigure}[t]{0.32\textwidth}
        \centering
        \boldsymbol{$\alpha = 2$}\\[0.1in]
        \includegraphics[width=\width\textwidth, trim=22cm 0cm 26cm 2cm,clip]{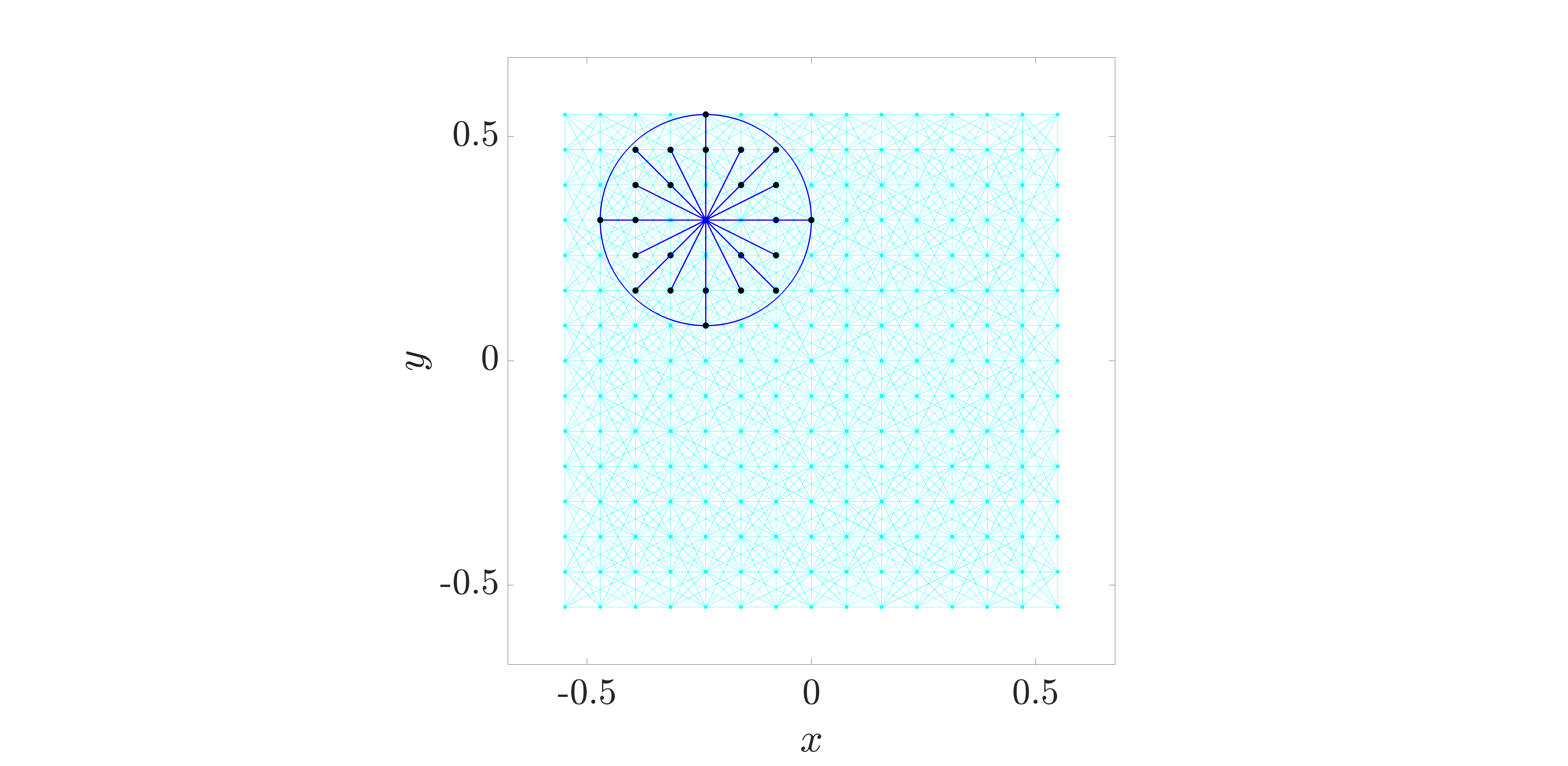}
        \caption{$\bar{s}^2 = \frac{1}{2}s_0^2$}
    \end{subfigure}
       \centering
    \begin{subfigure}[t]{0.32\textwidth}
        \centering
        \includegraphics[width=\width\textwidth,  trim=22cm 0cm 26cm 2cm,clip]{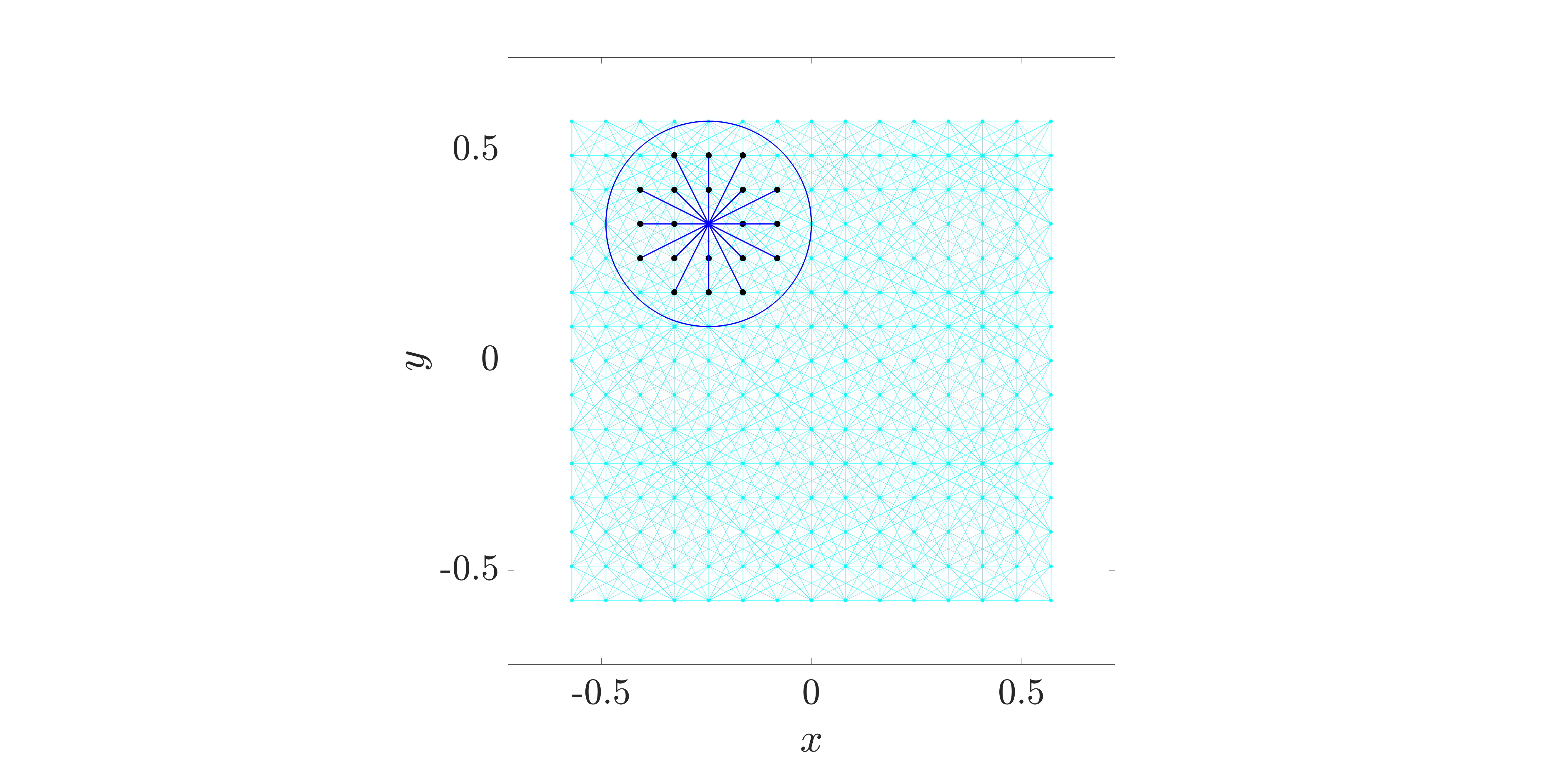}
        \caption{$\bar{s}^2 = \frac{9}{8\sqrt{2}}s_0^2$}
    \end{subfigure}
     \vspace{0.5pc}
    \begin{subfigure}[t]{0.32\textwidth}
        \centering
        {\color{white}$\alpha = 2$}
    \end{subfigure}
    \begin{subfigure}[t]{0.32\textwidth}
        \centering
        \includegraphics[width=\width\textwidth,  trim=22cm 0cm 26cm 2cm,clip]{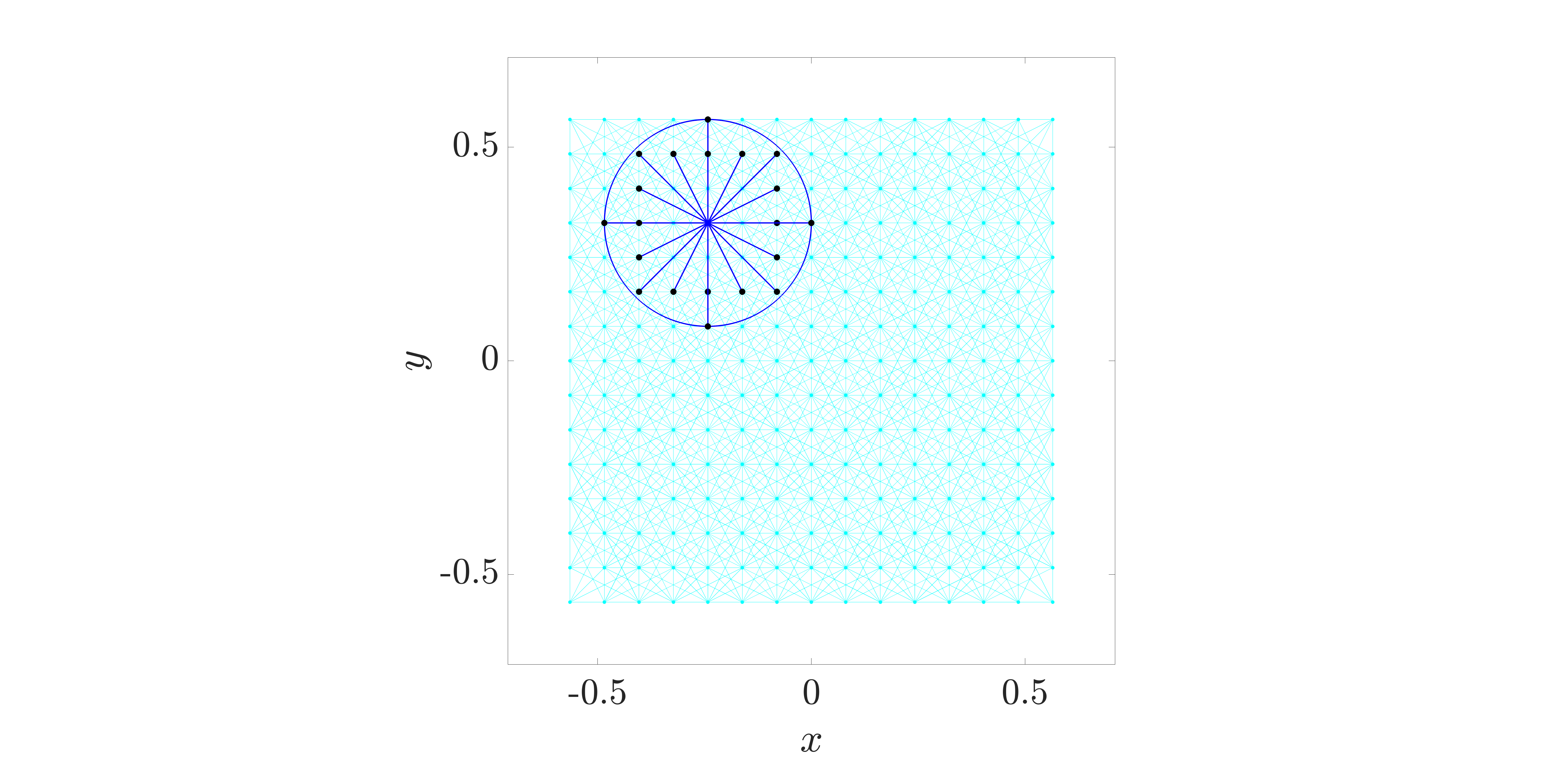}
        \caption{$\bar{s}^2 = \frac{\sqrt{2}}{2}s_0^2$}
    \end{subfigure} 
       \centering
    \begin{subfigure}[t]{0.32\textwidth}
        \centering
\includegraphics[width=\width\textwidth, trim=22cm 0cm 26cm 2cm,clip]{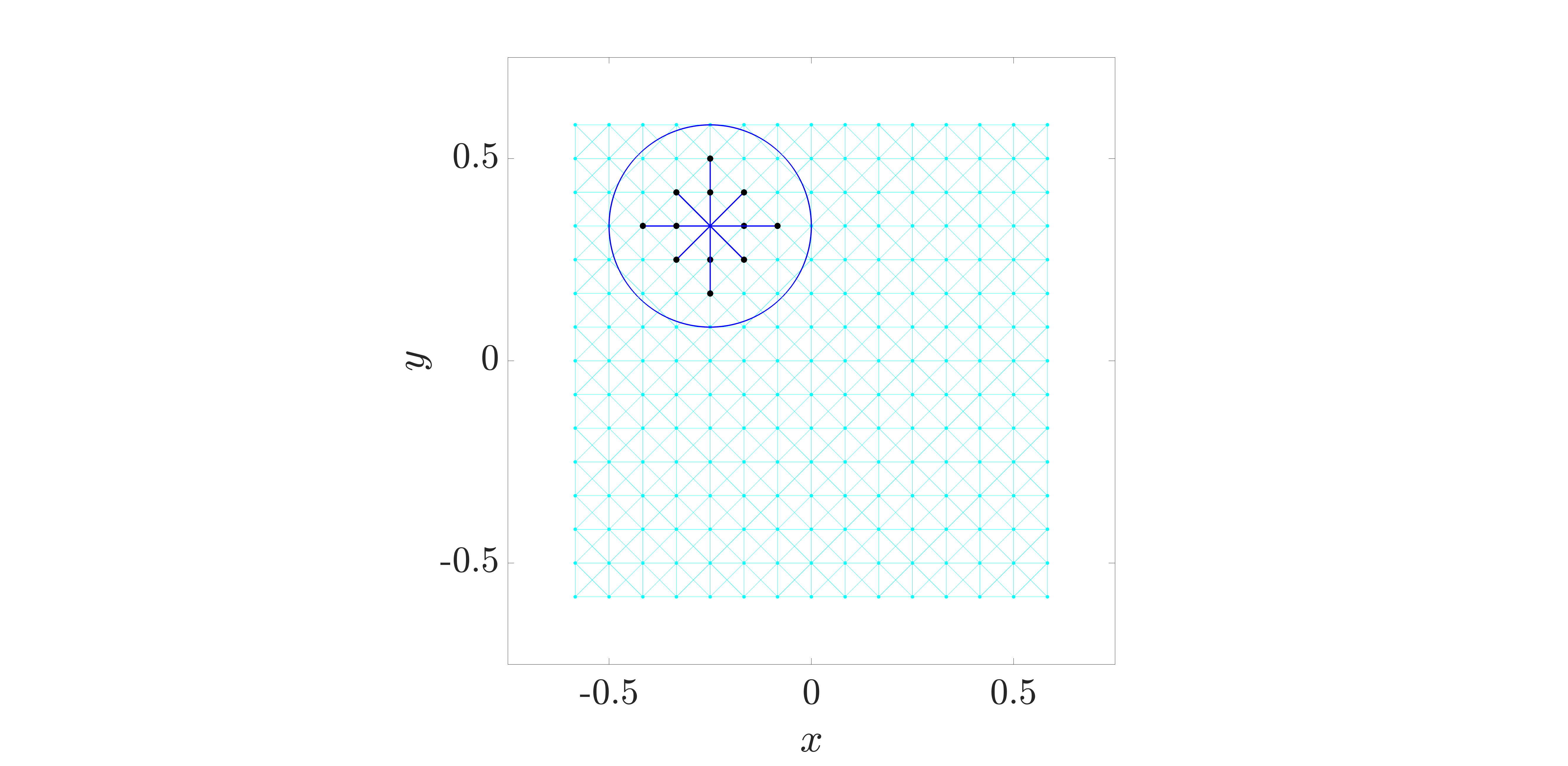}
        \caption{$\bar{s}^2 = \frac{9}{4\sqrt{5}}s_0^2$}
    \end{subfigure}
     \vspace{0.5pc}
    \begin{subfigure}[t]{0.32\textwidth}
        \centering
         {\color{white}$\alpha = 2$}
    \end{subfigure}
    \begin{subfigure}[t]{0.32\textwidth}
        \centering
        \includegraphics[width=\width\textwidth, trim=22cm 0cm 26cm 2cm,clip]{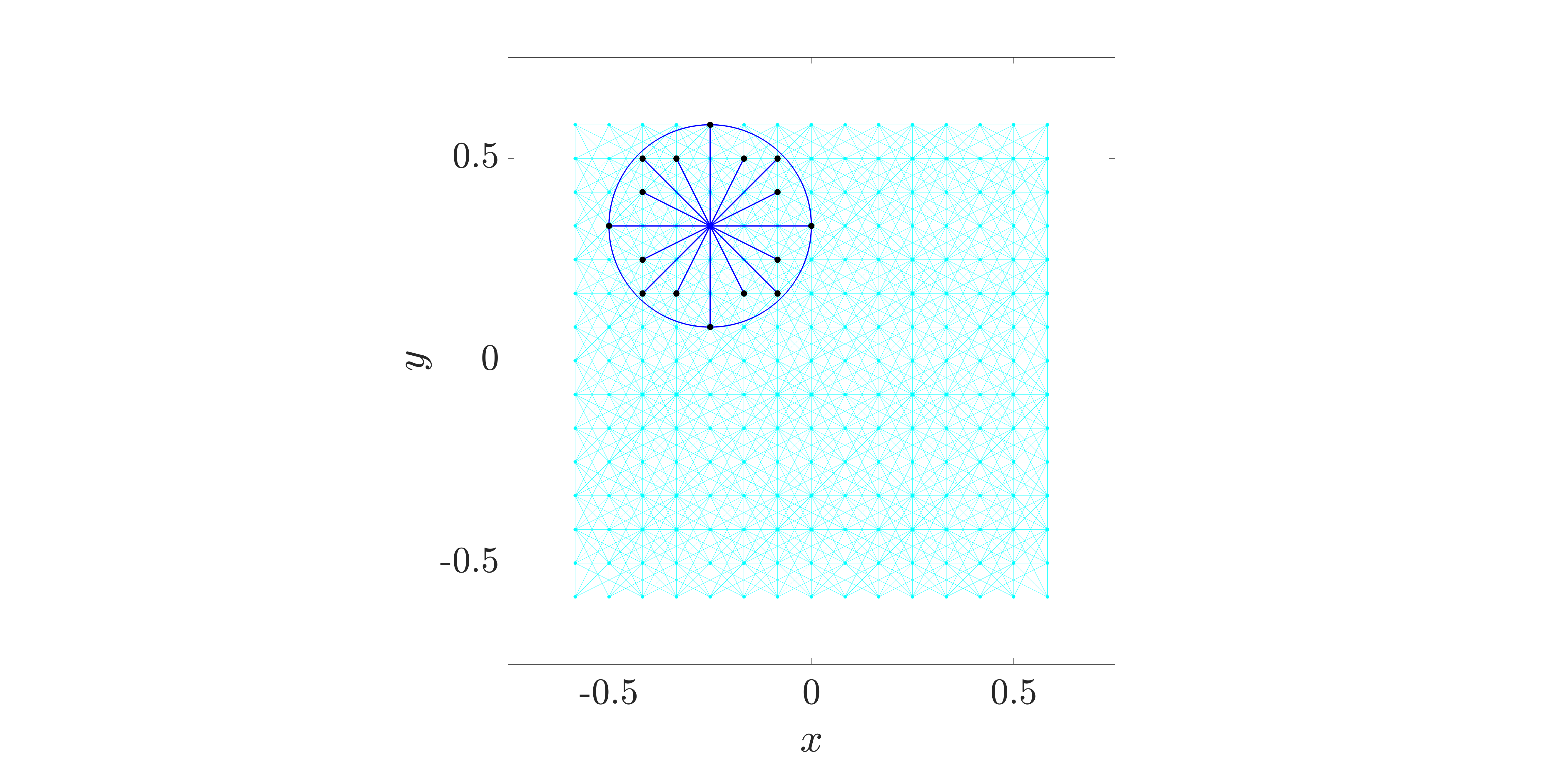}
        \caption{$\bar{s}^2 =s_0^2$}
    \end{subfigure} 
       \centering
    \begin{subfigure}[t]{0.32\textwidth}
        \centering
                \includegraphics[width=\width\textwidth,  trim=22cm 0cm 26cm 2cm,clip]{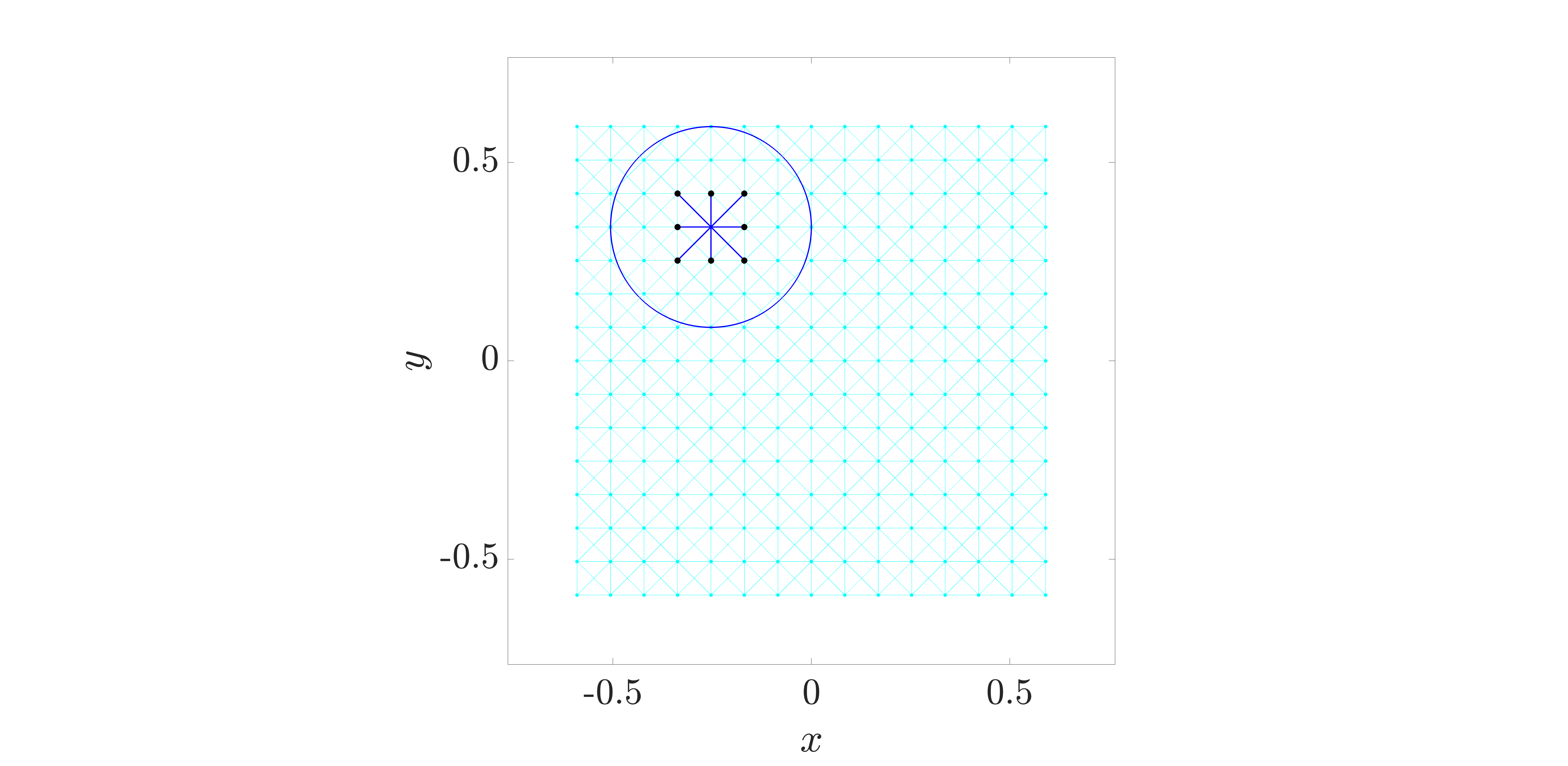}
        \caption{$\bar{s}^2 = \frac{9}{8}s_0^2$}
    \end{subfigure}
     \vspace{0.5pc}
    \begin{subfigure}[t]{0.32\textwidth}
        \centering
       {\color{white}$\alpha = 2$}
    \end{subfigure}
    \begin{subfigure}[t]{0.32\textwidth}
        \centering
        \includegraphics[width=\width\textwidth,  trim=22cm 0cm 26cm 2cm,clip]{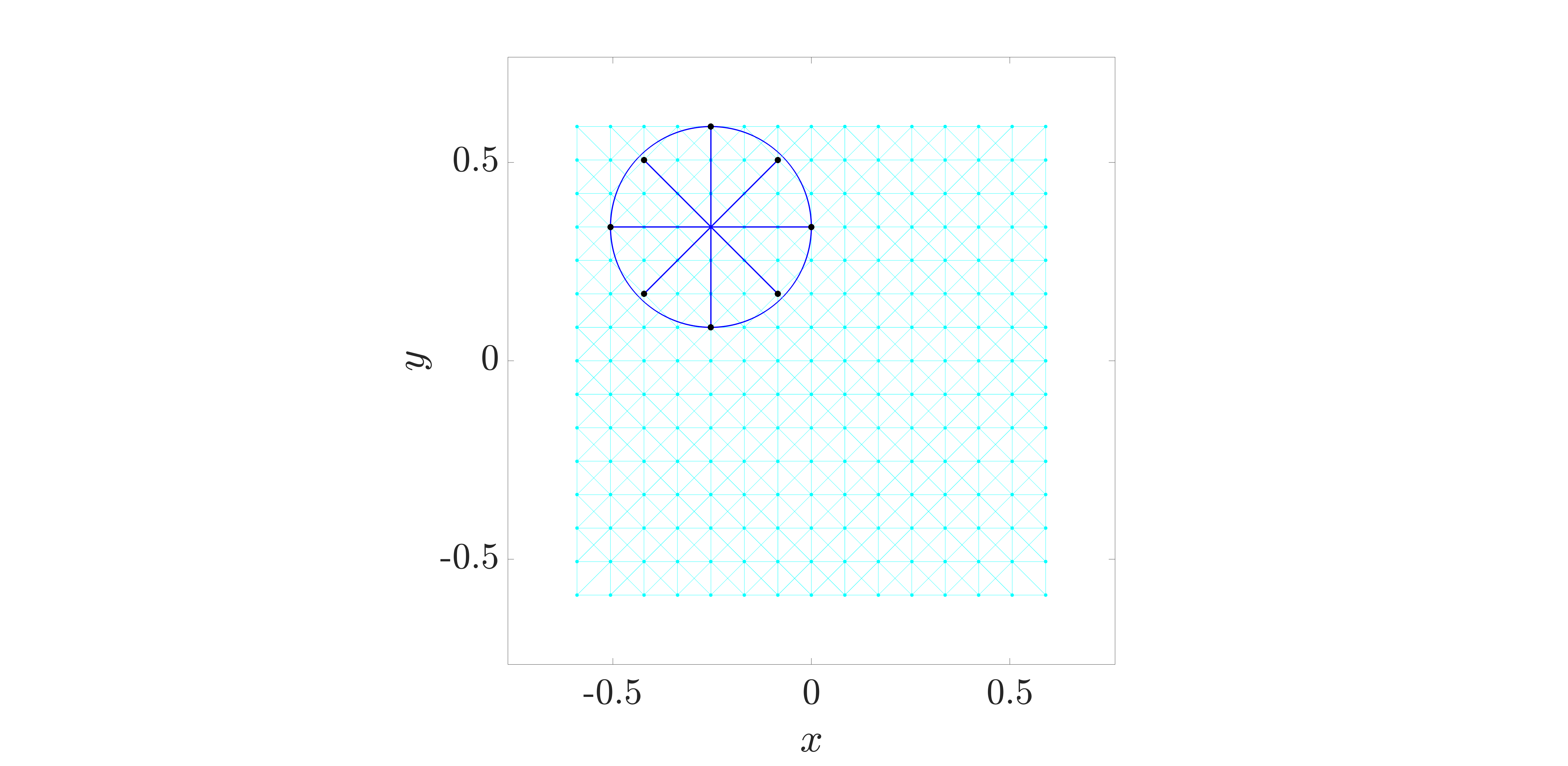}
        \caption{$\bar{s}^2= \frac{\sqrt{5}}{2}s_0^2$}
    \end{subfigure} 
       \centering
    \begin{subfigure}[t]{0.32\textwidth}
        \centering
               \includegraphics[width=\width\textwidth, trim=22cm 0cm 26cm 2cm,clip]{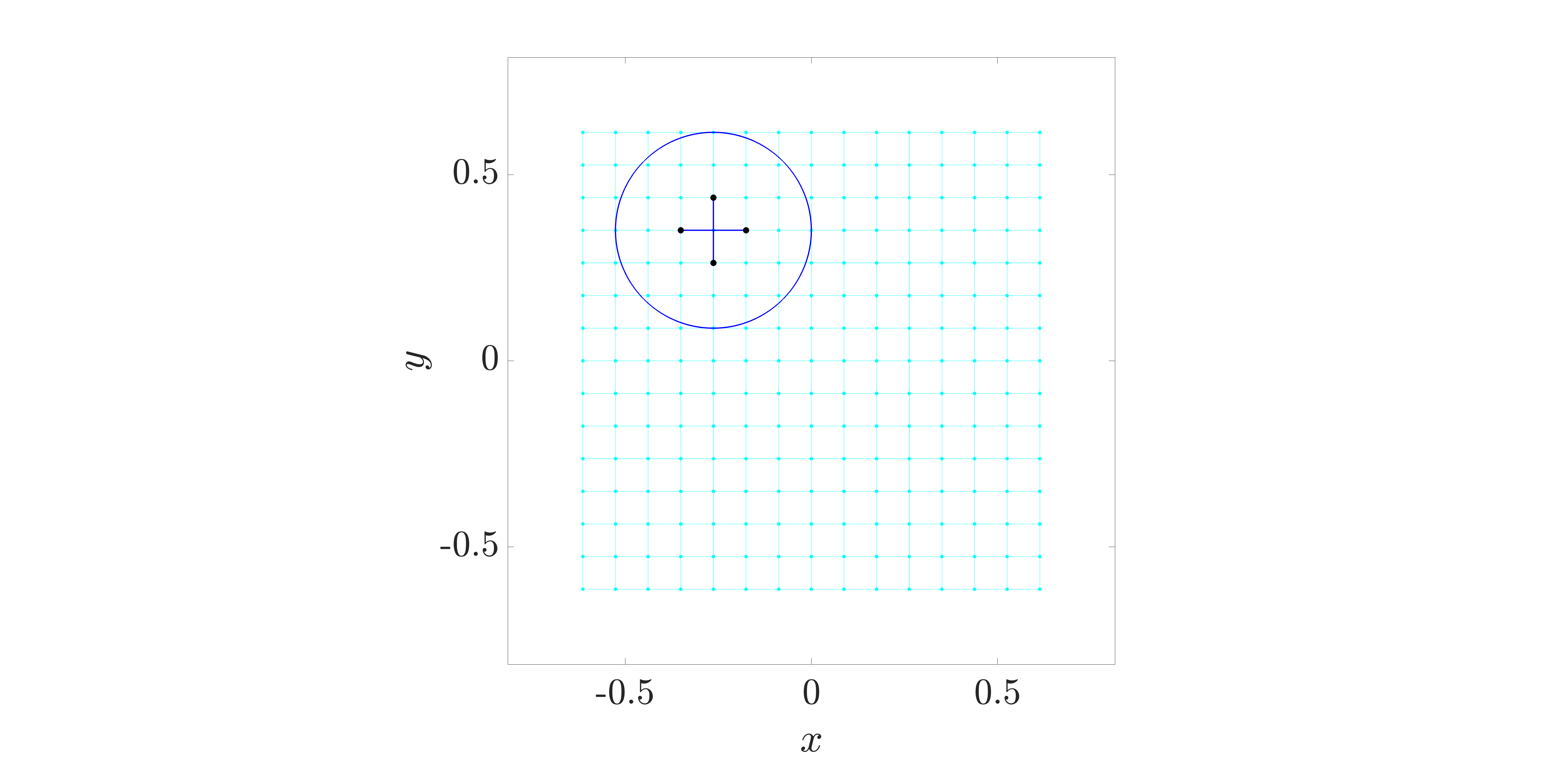}
        \caption{$\bar{s}^2 = \frac{9}{4\sqrt{2}}s_0^2$} 
    \end{subfigure}   
     \vspace{0.5pc}
    \begin{subfigure}[t]{0.32\textwidth}
        \centering
      {\color{white}$\alpha = 2$}
    \end{subfigure}
    \begin{subfigure}[t]{0.32\textwidth}
        \centering
        \includegraphics[width=\width\textwidth, trim=22cm 0cm 26cm 2cm,clip]{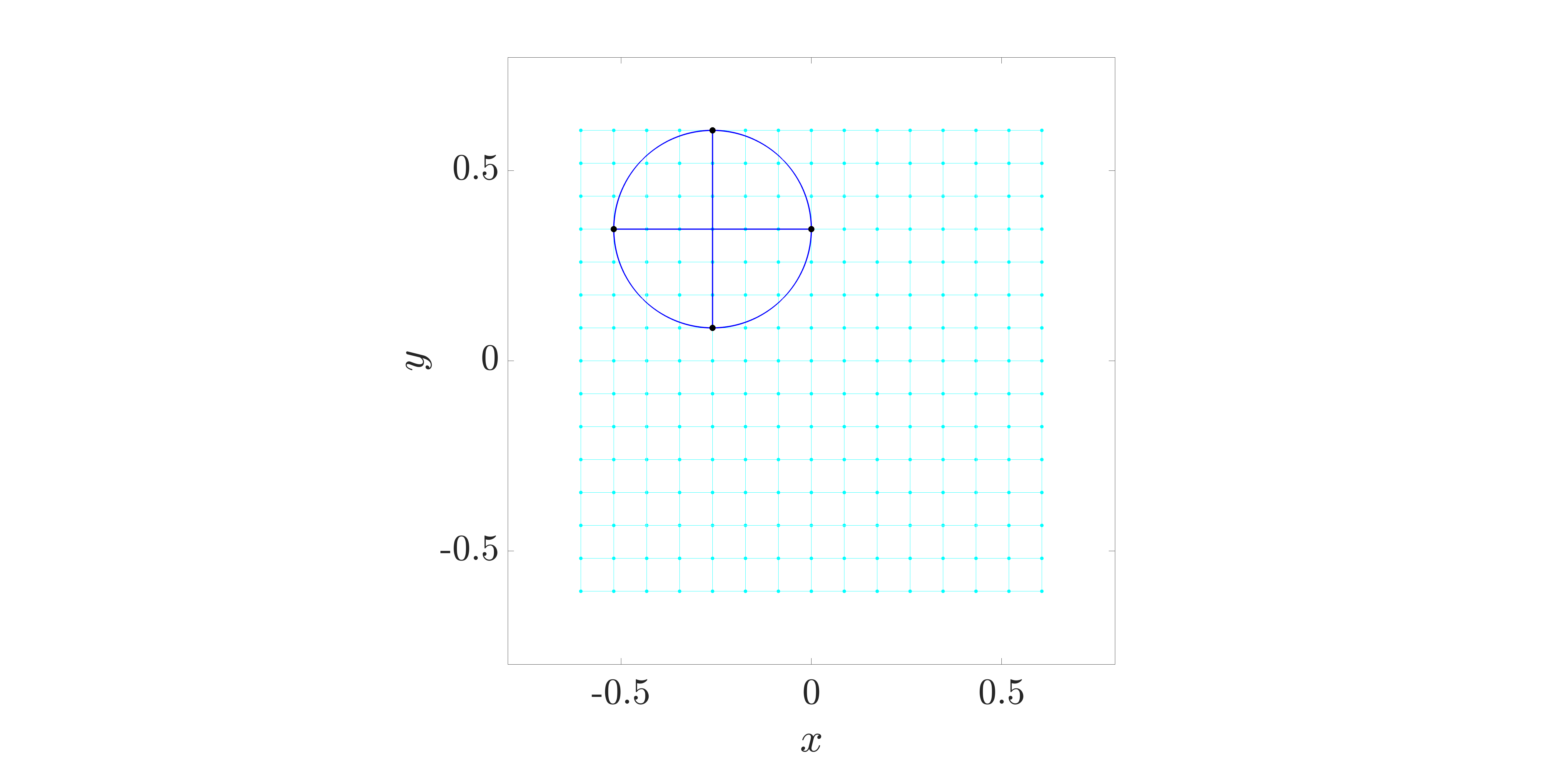}
        \caption{$\bar{s}^2 = \sqrt{2}s_0^2$}
    \end{subfigure} 
       \centering
    \begin{subfigure}[t]{0.32\textwidth}
        \centering
         \includegraphics[width=\width\textwidth, trim=22cm 0cm 26cm 2cm,clip]{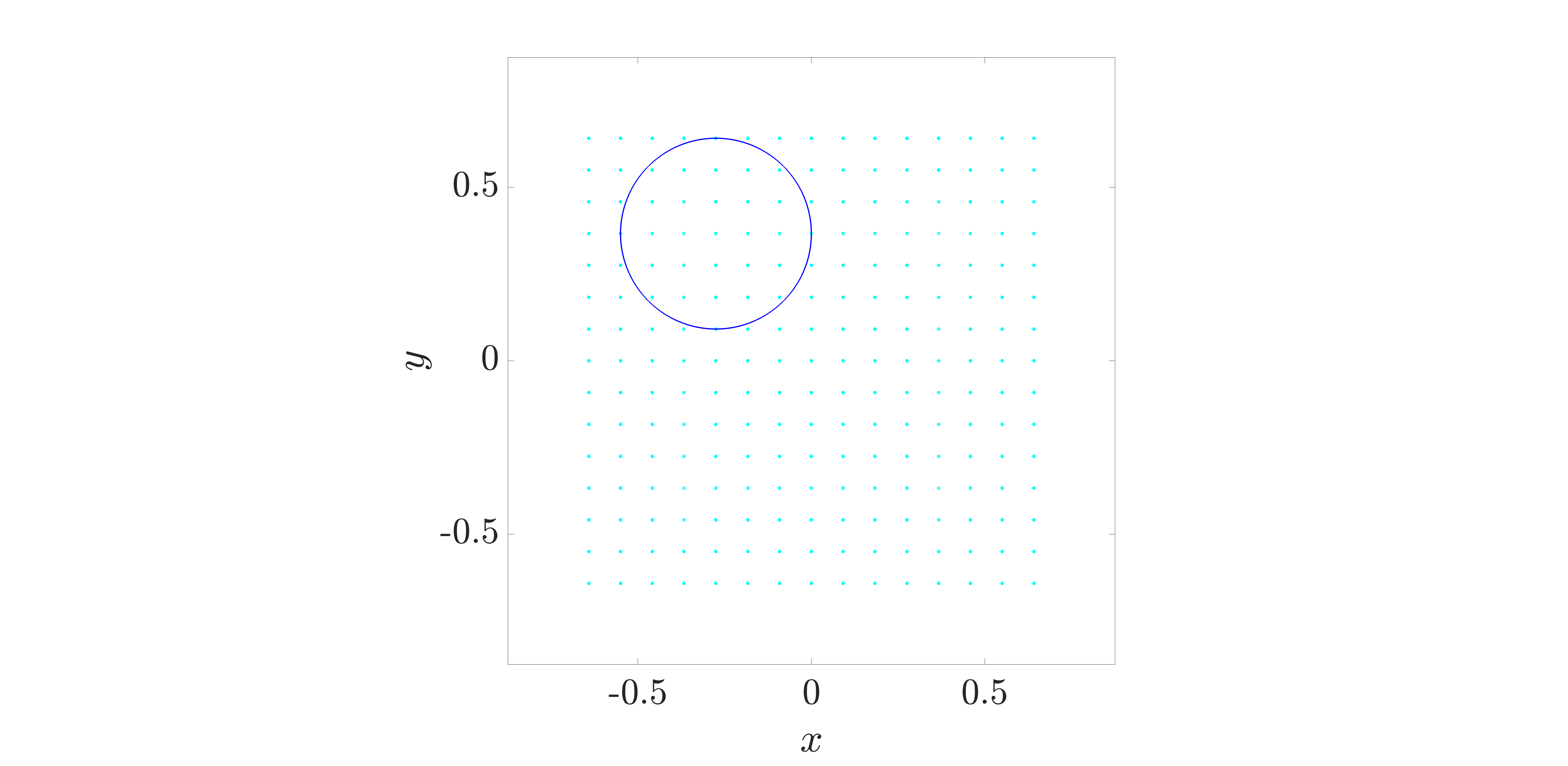}
        \caption{$\bar{s}^2 = \frac{9}{4}s_0^2$}
    \end{subfigure}
    \begin{subfigure}[t]{0.32\textwidth}
        \centering
        \includegraphics[width=\width\textwidth, trim=22cm 0cm 26cm 2cm,clip]{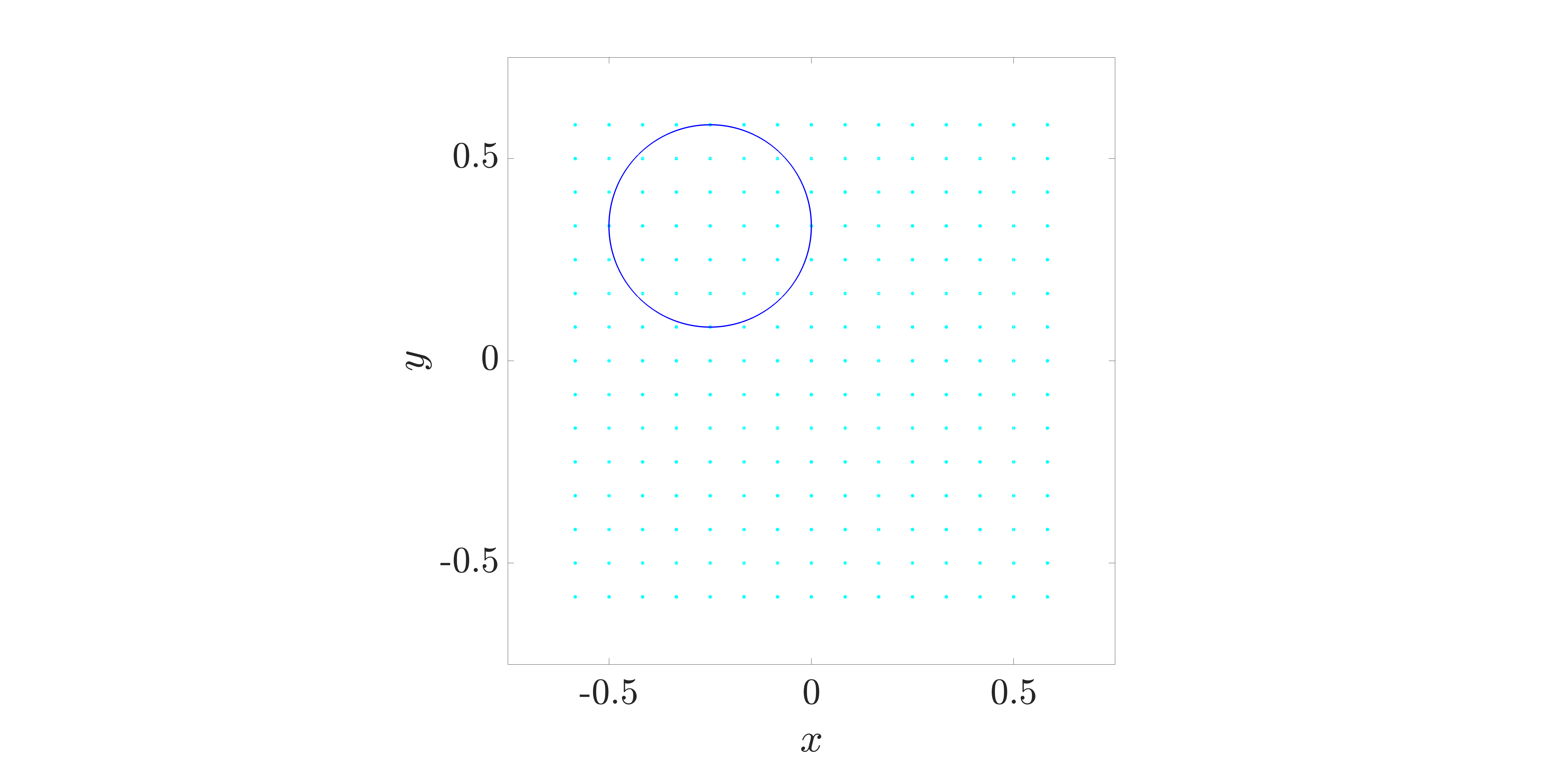}
        \caption{$\bar{s}^2 = s_0^2$}
    \end{subfigure}
    \begin{subfigure}[t]{0.32\textwidth}
        \centering
        \includegraphics[width=\width\textwidth, trim=22cm 0cm 26cm 2cm,clip]{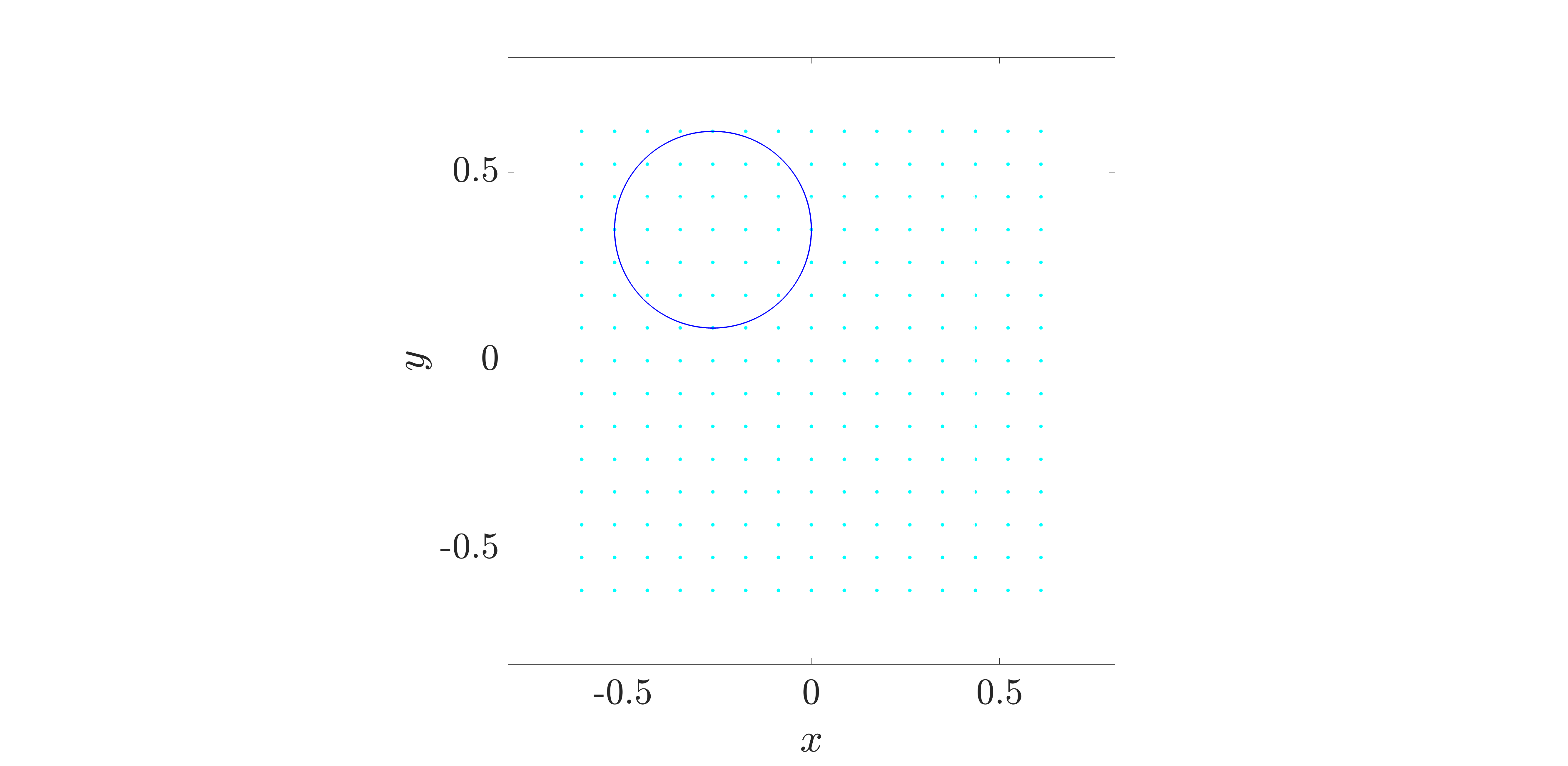}
        \caption{$\bar{s}^2 = \frac{3}{2}s_0^2$}
    \end{subfigure} 
    \caption{Illustration of the bond-breaking patterns for the critical energy density bond-failure criterion for different values of the influence function parameter $\alpha$ in Example~1. 
    }
    \label{Fig: Example 1 results}
\end{figure*}

\subsection{Example 2: Crack tip evolution}
\label{sec: Crack tip evolution}

In this example, we compare the effect of the bond-failure criteria on the evolution of a crack tip. 
Accordingly, we consider the problem of 
crack propagation in 
a pre-notched soda-lime glass thin plate subjected to traction loading~\cite{Bobaru-Zhang-2015}. The details of this problem, including the discretization parameters, are provided in the next section, where we study the full behavior of the system. 
Here, however, 
we focus on the evolution of the crack tip during the initial steps that lead to crack propagation. Our goal is to illustrate 
the sequence of bond breaking for different choices of the influence function parameter $\alpha$ for the critical stretch and critical energy density criteria. 

As opposed to the previous section, where the deformation is prescribed, here the deformation results from the material response to the imposed loading and, as such, depends on the material model chosen, in particular the influence function.  
Indeed, different choices of influence function can result in different elastic and fracture behavior~\cite{SelesonParks}. It was observed experimentally for Homalite that crack propagation is affected by stress wave interactions~(see, e.g.,~\cite{Ravi-ChandarKnauss});  this phenomenon was confirmed numerically through peridynamic simulations in both soda-lime glass and Homalite~(see, e.g.,~\cite{Bobaru-Zhang-2015}). 
Therefore, in this section, rather than comparing simulation results for different influence functions, we study the difference between the bond-failure criteria for given choices of the influence function, i.e., for given values of $\alpha$.  

Here, we consider the case of a traction magnitude of  $\sigma = 0.2$ MPa.  
The results are reported for $\alpha = 0$ in Figures~\ref{Fig: Example 2a results} and \ref{Fig: Example 2b results}, and for $\alpha =2$ in Figures~\ref{Fig: Example 2c results} and~\ref{Fig: Example 2d results}. For clarity, we specify the bonds that break at each step in Table~\ref{table: Example 2 broken bonds alpha 0.} for $\alpha = 0$ and in Table~\ref{table: Example 2 broken bonds alpha 2.} for $\alpha = 2$.

\def\width{0.6}


\begin{figure*}[htbp!]
    \centering
     {\bf The case}~\boldsymbol{$\alpha = 0$}\\[0.1in]
    \begin{subfigure}[t]{0.49\textwidth}
        \centering
      {\bf Critical stretch}\\[0.1in]
        \includegraphics[width=\width\textwidth, trim=0cm 0cm 0cm 0cm,clip]{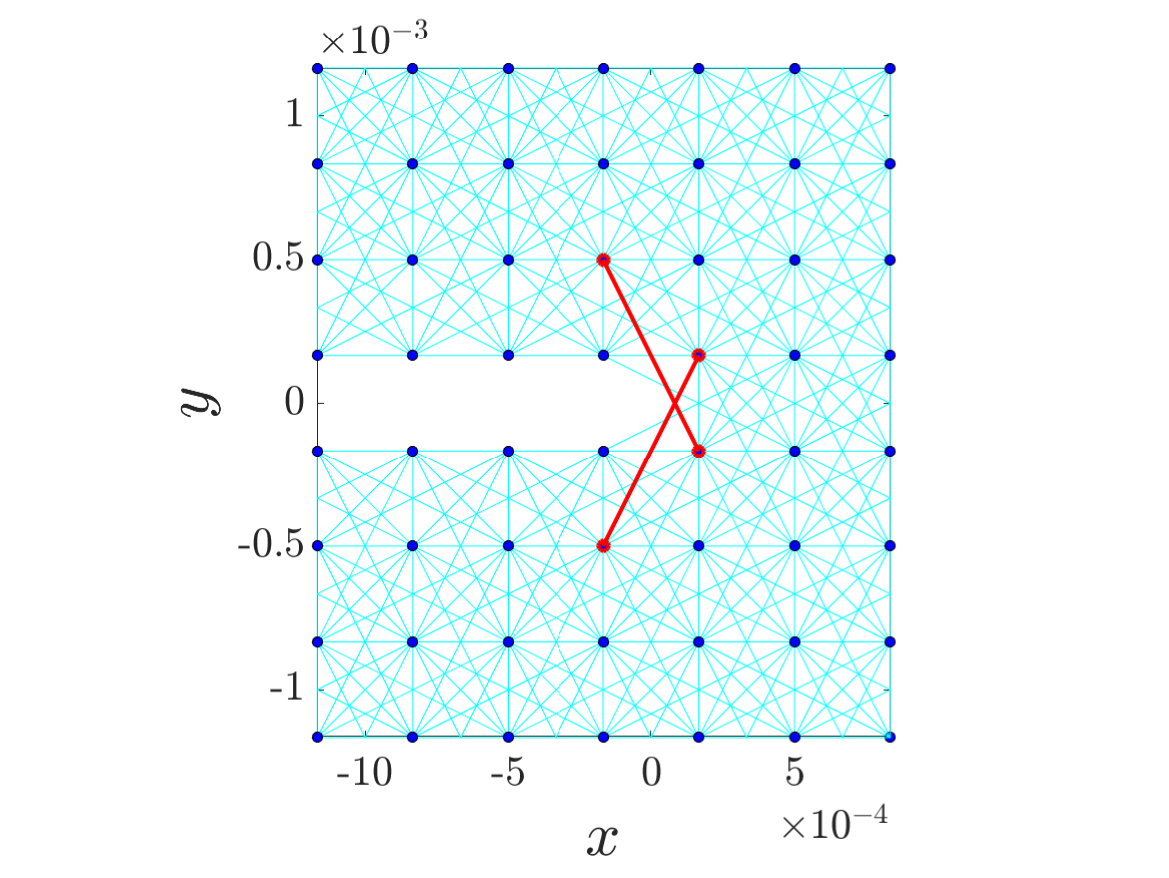}
        \caption{step $= 928$}
         \label{Fig: Example 2a results (a)}
    \end{subfigure}%
     \vspace{0.5pc}
    \begin{subfigure}[t]{0.49\textwidth}
        \centering
    {\bf Critical energy density}\\[0.1in]
       \includegraphics[width=\width\textwidth, trim=0cm 0cm 0cm 0cm,clip]{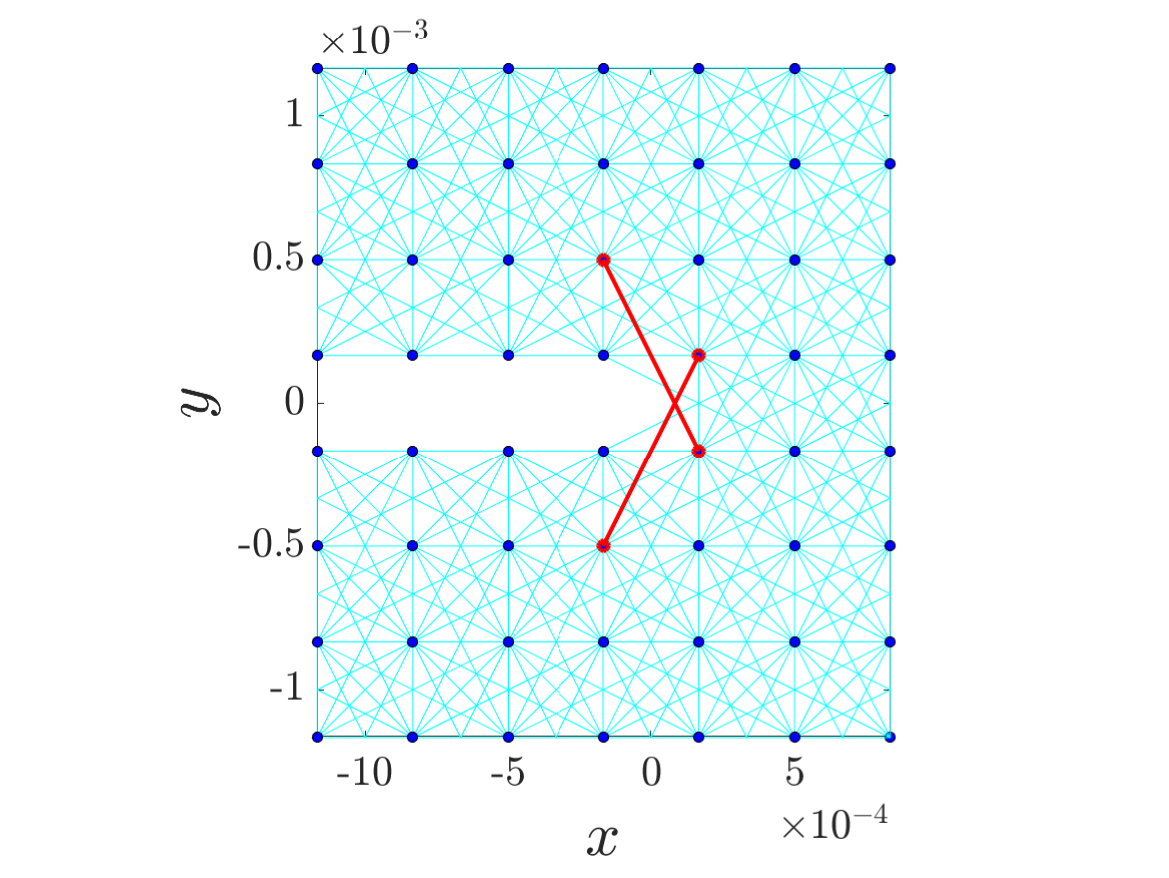}
        \caption{step $= 930$}
         \label{Fig: Example 2a results (b)}
    \end{subfigure}
        \begin{subfigure}[t]{0.49\textwidth}
        \centering
        \includegraphics[width=\width\textwidth, trim=0cm 0cm 0cm 0cm,clip]{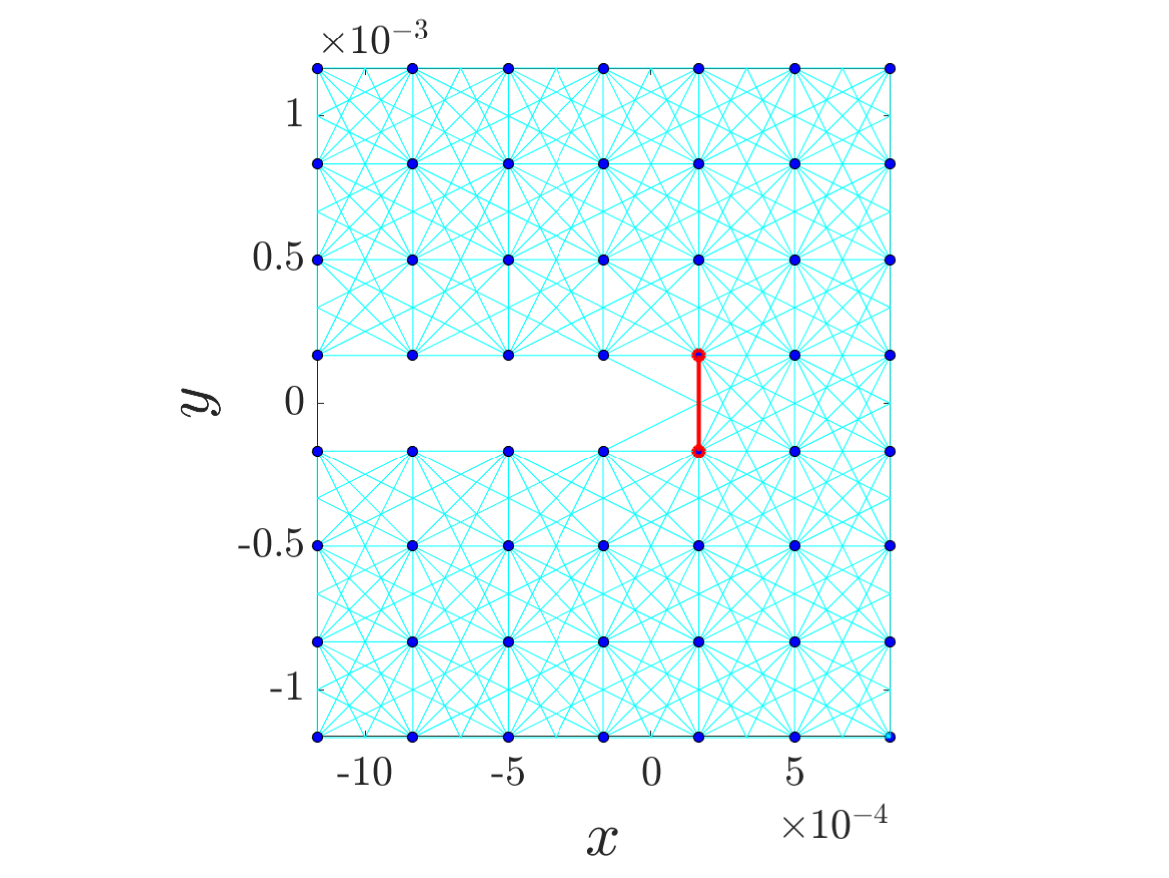}
        \caption{step $= 929$}
         \label{Fig: Example 2a results (c)}
    \end{subfigure}%
     \vspace{0.5pc}
    \begin{subfigure}[t]{0.49\textwidth}
        \centering
        \includegraphics[width=\width\textwidth, trim=0cm 0cm 0cm 0cm,clip]{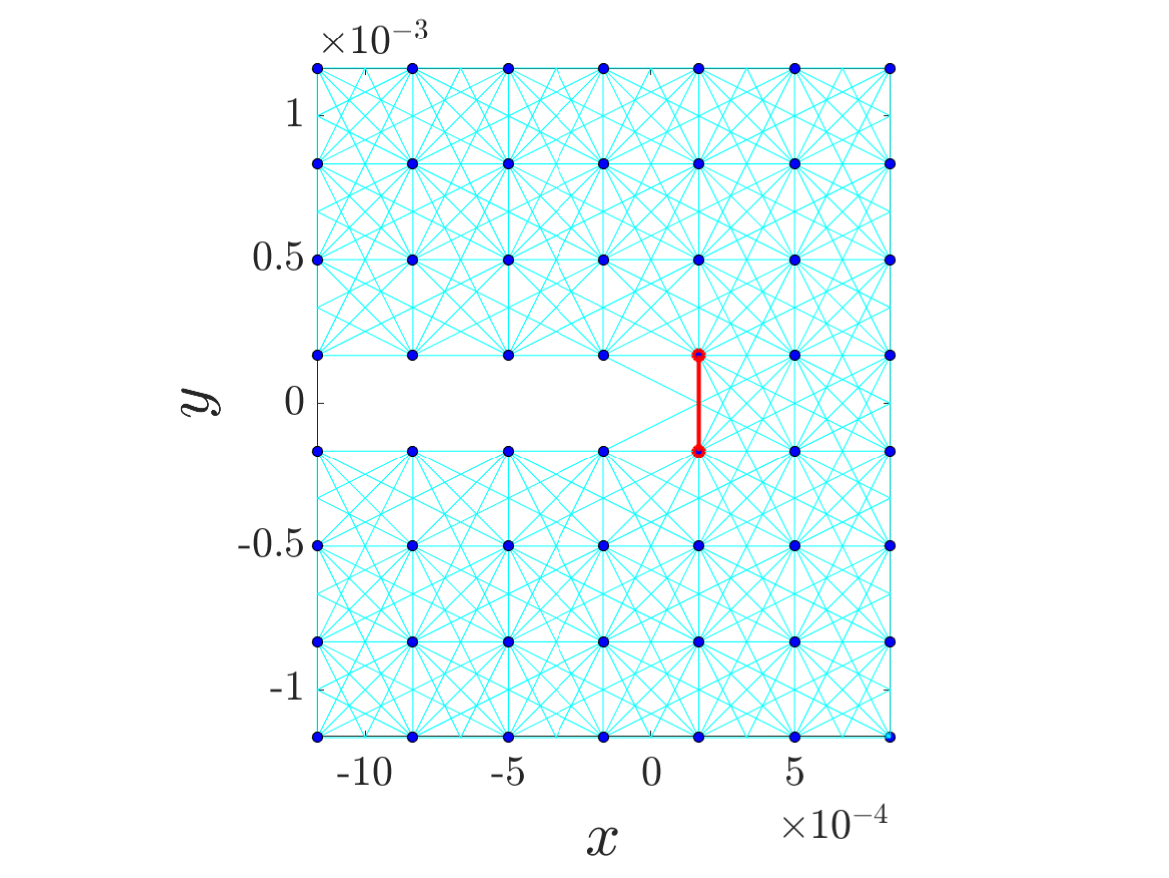}
        \caption{step $= 933$}
         \label{Fig: Example 2a results (d)}
    \end{subfigure}
        \begin{subfigure}[t]{0.49\textwidth}
        \centering
        \includegraphics[width=\width\textwidth, trim=0cm 0cm 0cm 0cm,clip]{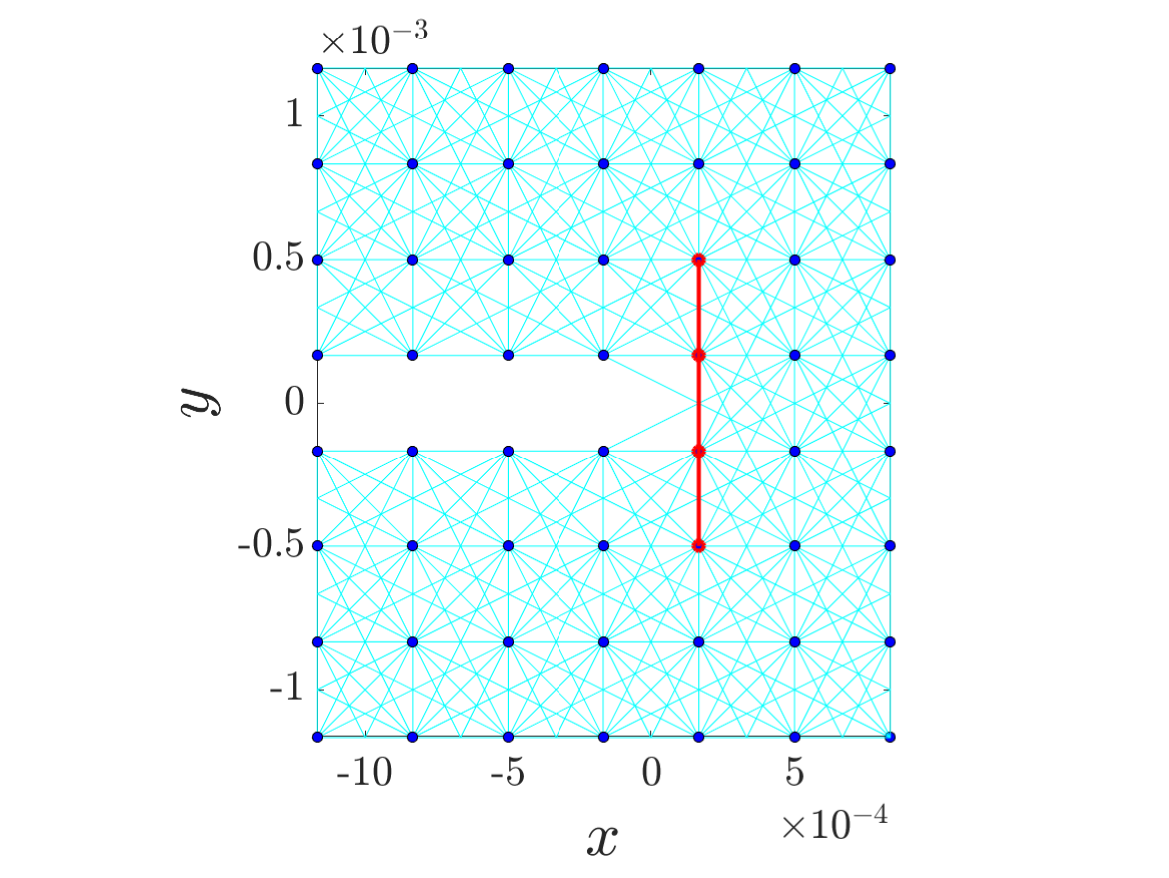}
        \caption{step $= 930$}
         \label{Fig: Example 2a results (e)}
    \end{subfigure}%
     \vspace{0.5pc}
    \begin{subfigure}[t]{0.49\textwidth}
        \centering
        \includegraphics[width=\width\textwidth, trim=0cm 0cm 0cm 0cm,clip]{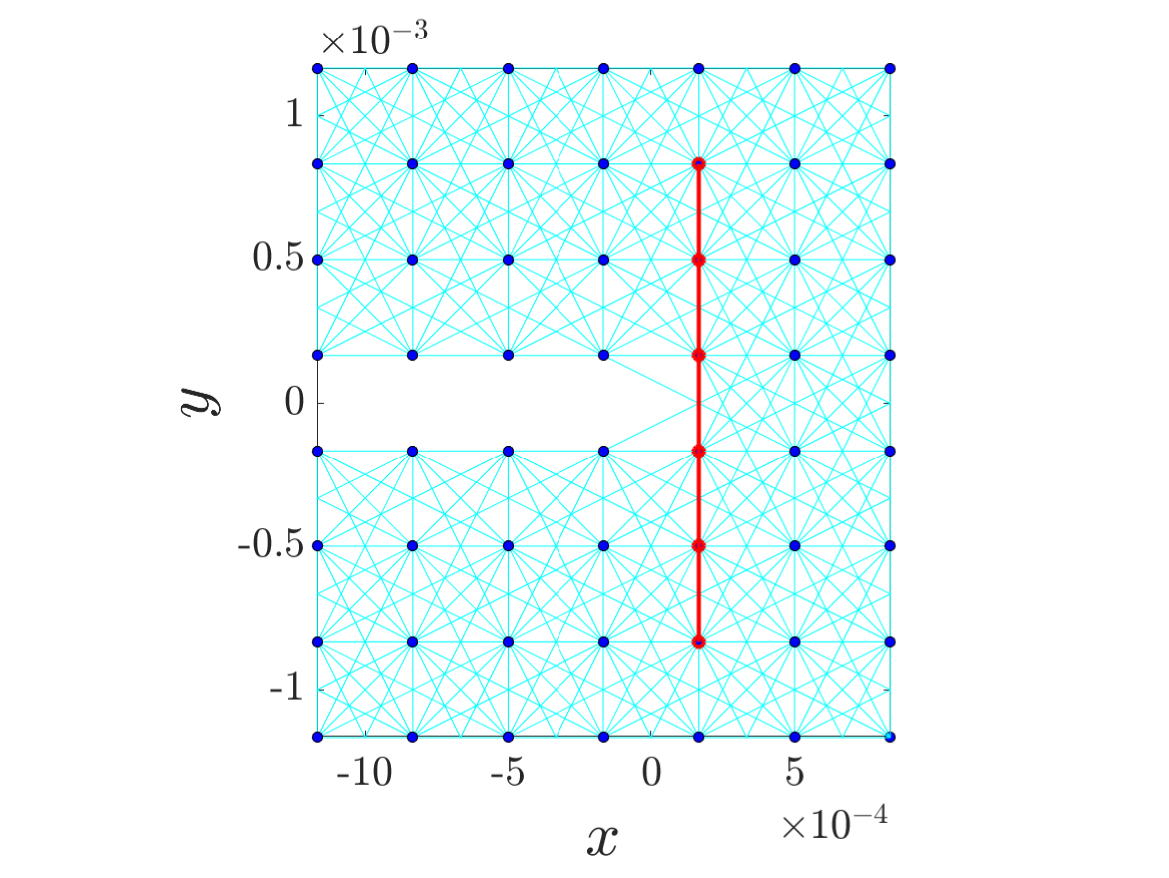}
        \caption{step $= 935$}
         \label{Fig: Example 2a results (f)}
    \end{subfigure}
            \begin{subfigure}[t]{0.49\textwidth}
        \centering
        \includegraphics[width=\width\textwidth, trim=0cm 0cm 0cm 0cm,clip]{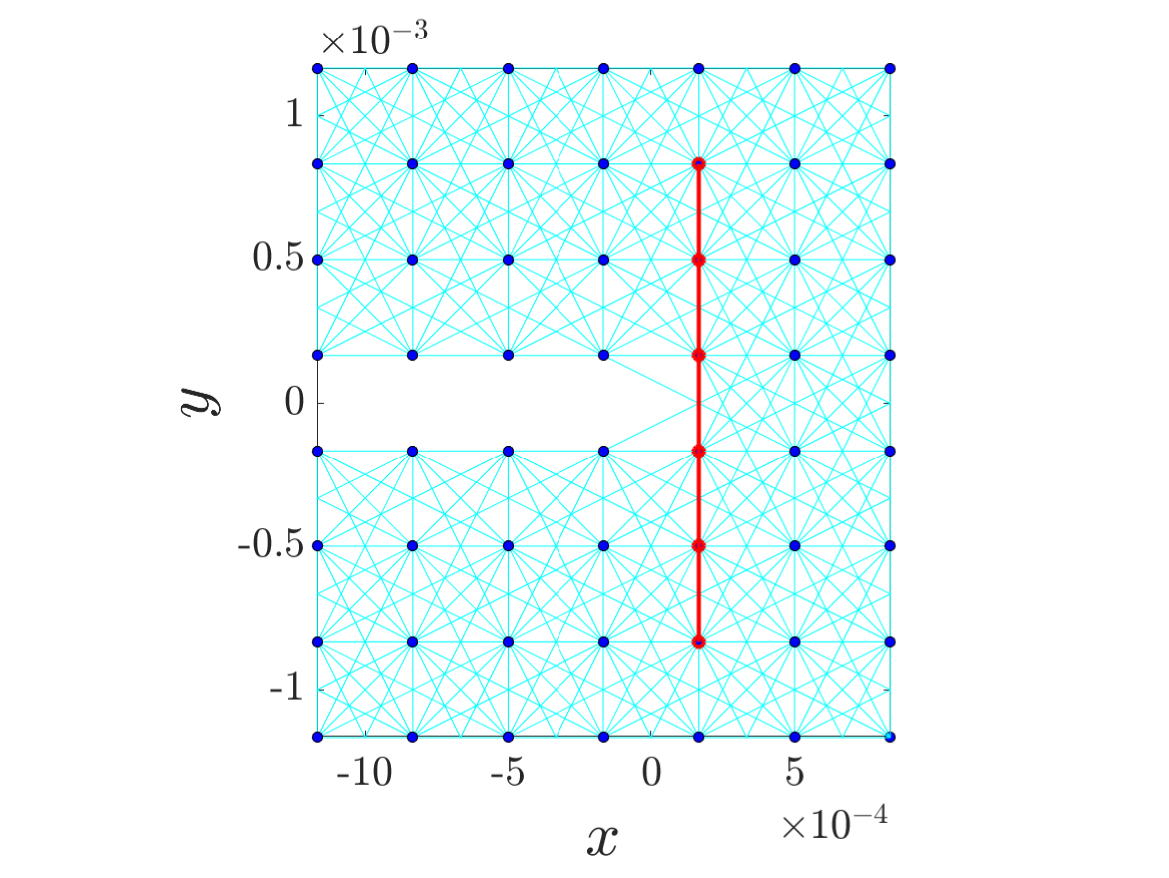}
        \caption{step $= 932$}
         \label{Fig: Example 2a results (g)}
    \end{subfigure}%
     \vspace{0.5pc}
    \begin{subfigure}[t]{0.49\textwidth}
        \centering
      \includegraphics[width=\width\textwidth, trim=0cm 0cm 0cm 0cm,clip]{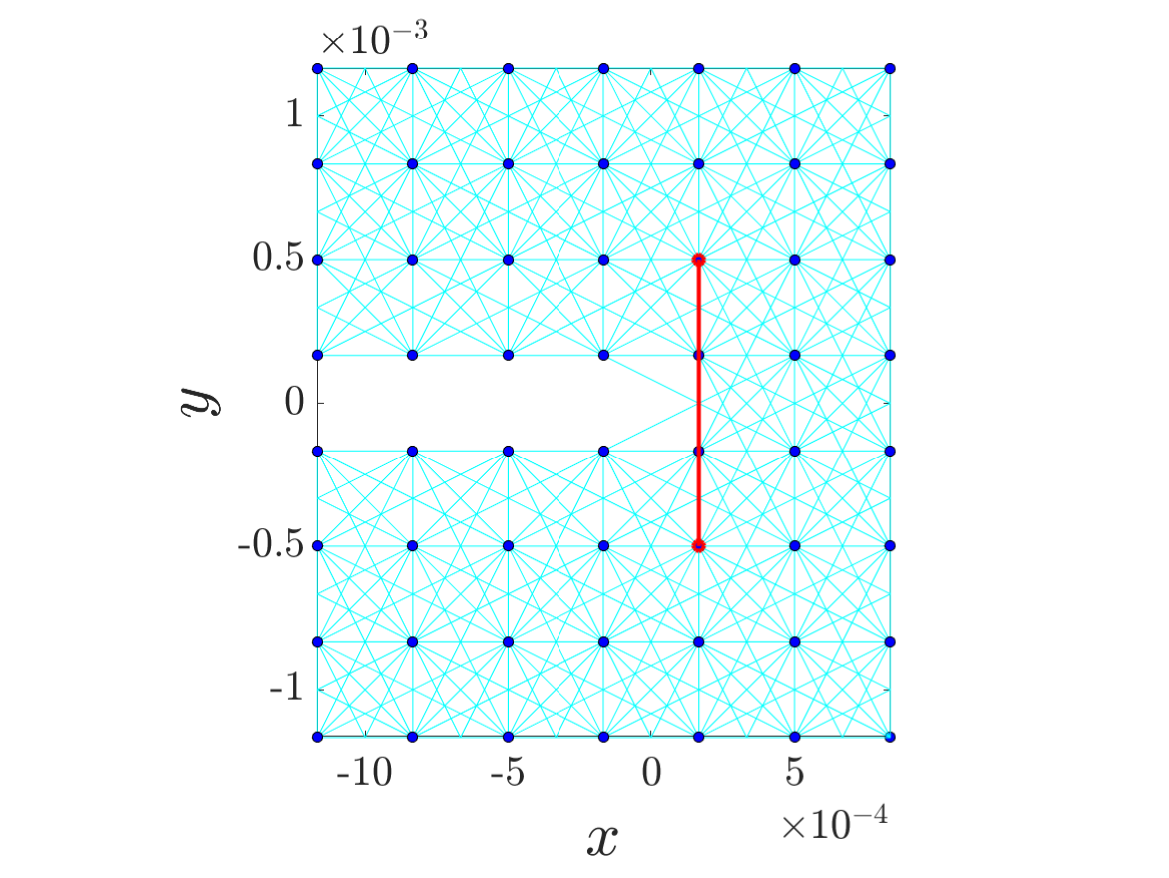}
        \caption{step $= 936$}
         \label{Fig: Example 2a results (h)}
    \end{subfigure}
                \begin{subfigure}[t]{0.49\textwidth}
        \centering
        \includegraphics[width=\width\textwidth, trim=0cm 0cm 0cm 0cm,clip]{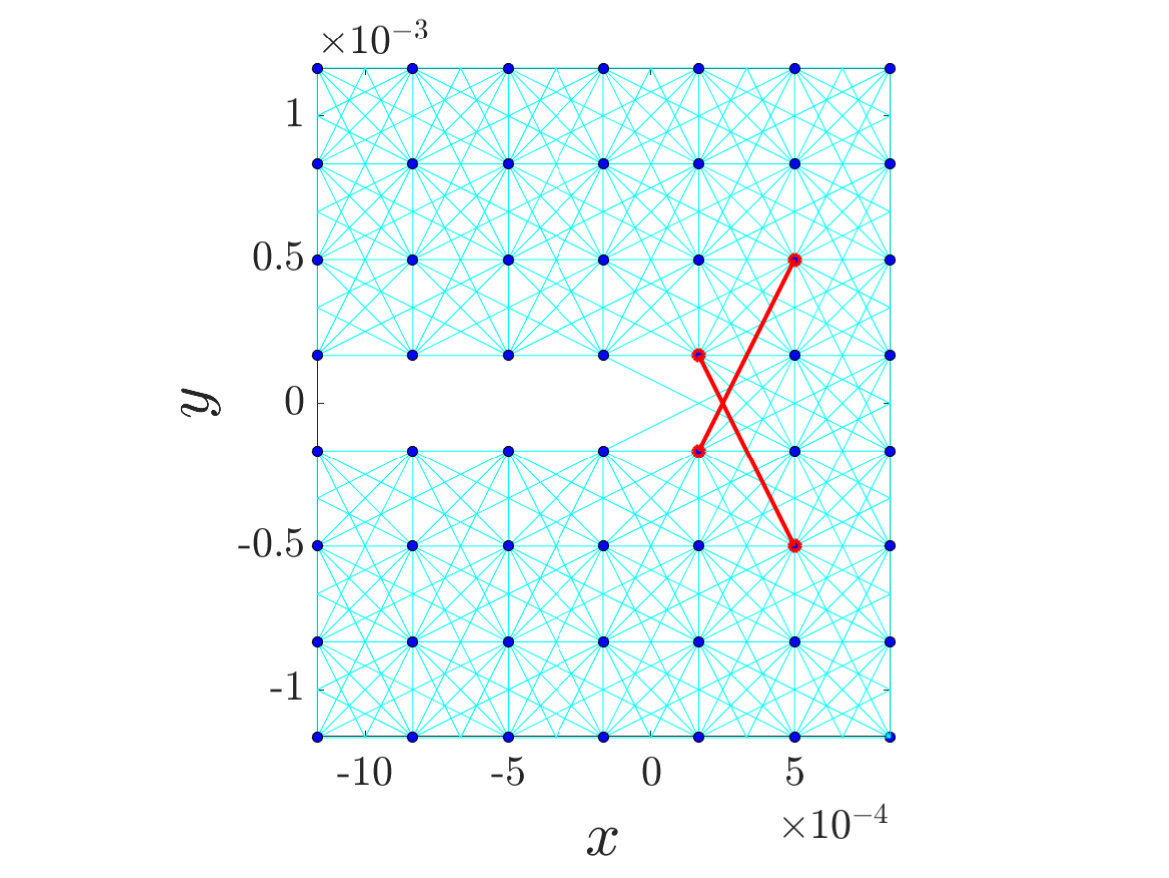}
        \caption{step $= 933$}
         \label{Fig: Example 2a results (i)}
    \end{subfigure}%
     \vspace{0.5pc}
    \begin{subfigure}[t]{0.49\textwidth}
        \centering
        \includegraphics[width=\width\textwidth, trim=0cm 0cm 0cm 0cm,clip]{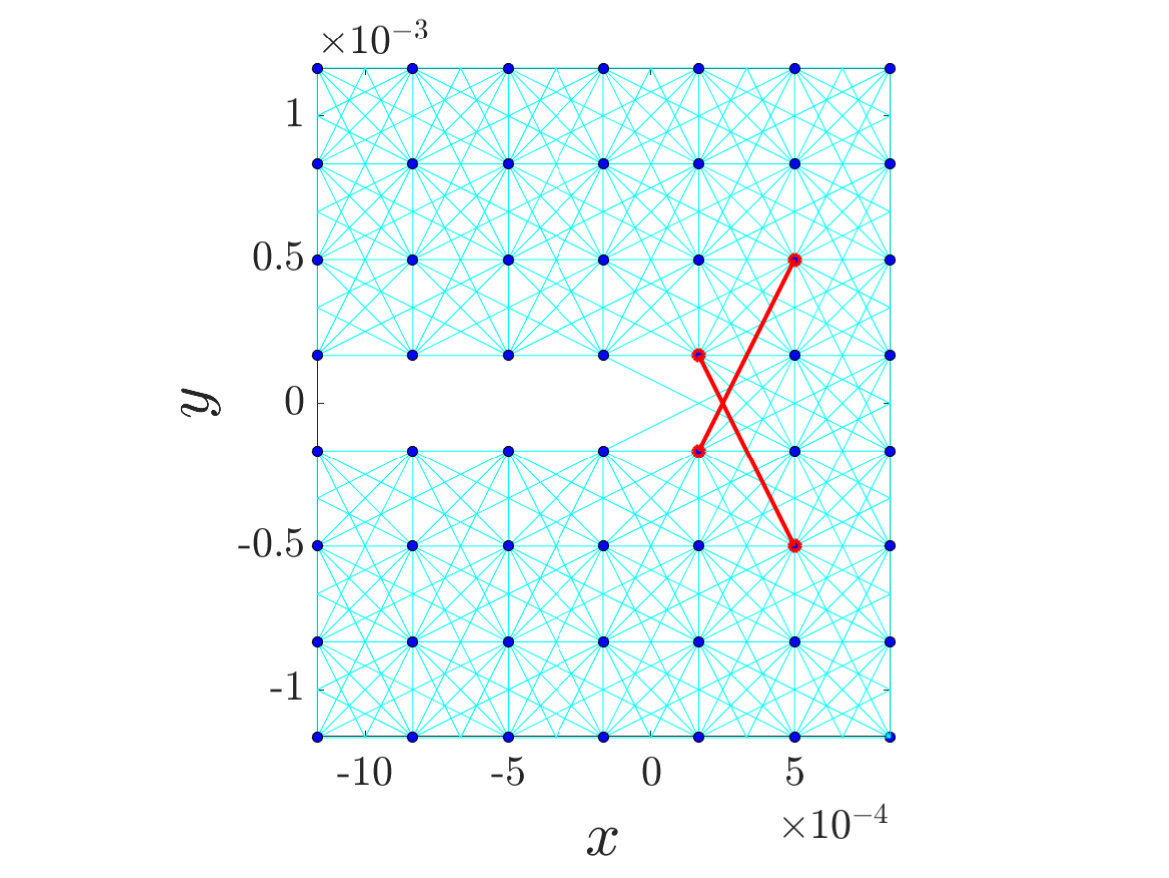}
        \caption{step $= 937$}
         \label{Fig: Example 2a results (j)}
    \end{subfigure}
        \caption{Comparison of the crack tip evolution between the two bond-failure criteria for $\alpha = 0$ in Example 2.} 
    \label{Fig: Example 2a results}
\end{figure*} 


\begin{figure*}[htbp!]
    \centering
    {\bf The case}~\boldsymbol{$\alpha = 0$} {\bf (cont.)}\\[0.1in]
    \begin{subfigure}[t]{0.49\textwidth}
        \centering
      {\bf Critical stretch}\\[0.1in]
        \includegraphics[width=\width\textwidth, trim=0cm 0cm 0cm 0cm,clip]{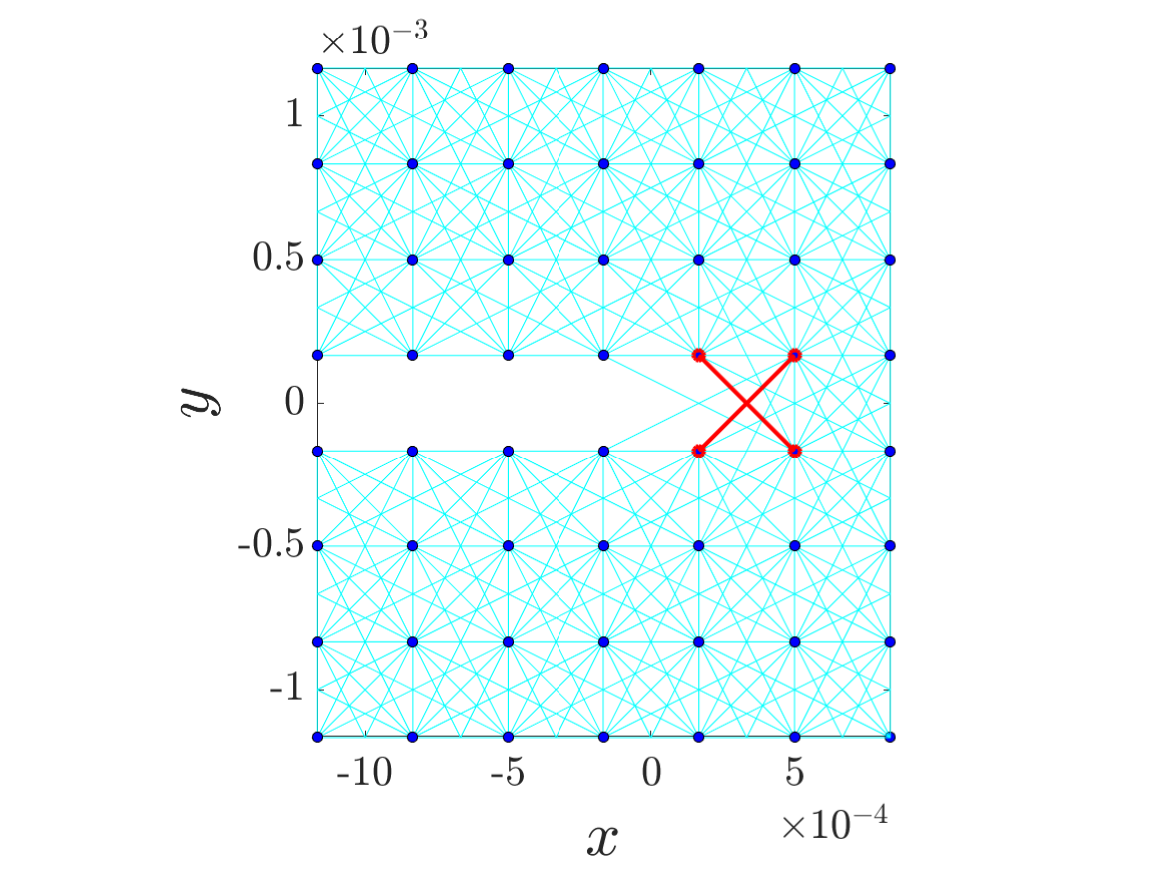}
        \caption{step $= 935$}
    \label{Fig: Example 2b results (a)}
    \end{subfigure}%
     \vspace{0.5pc}
    \begin{subfigure}[t]{0.49\textwidth}
        \centering
    {\bf Critical energy density}\\[0.1in]
        \includegraphics[width=\width\textwidth, trim=0cm 0cm 0cm 0cm,clip]{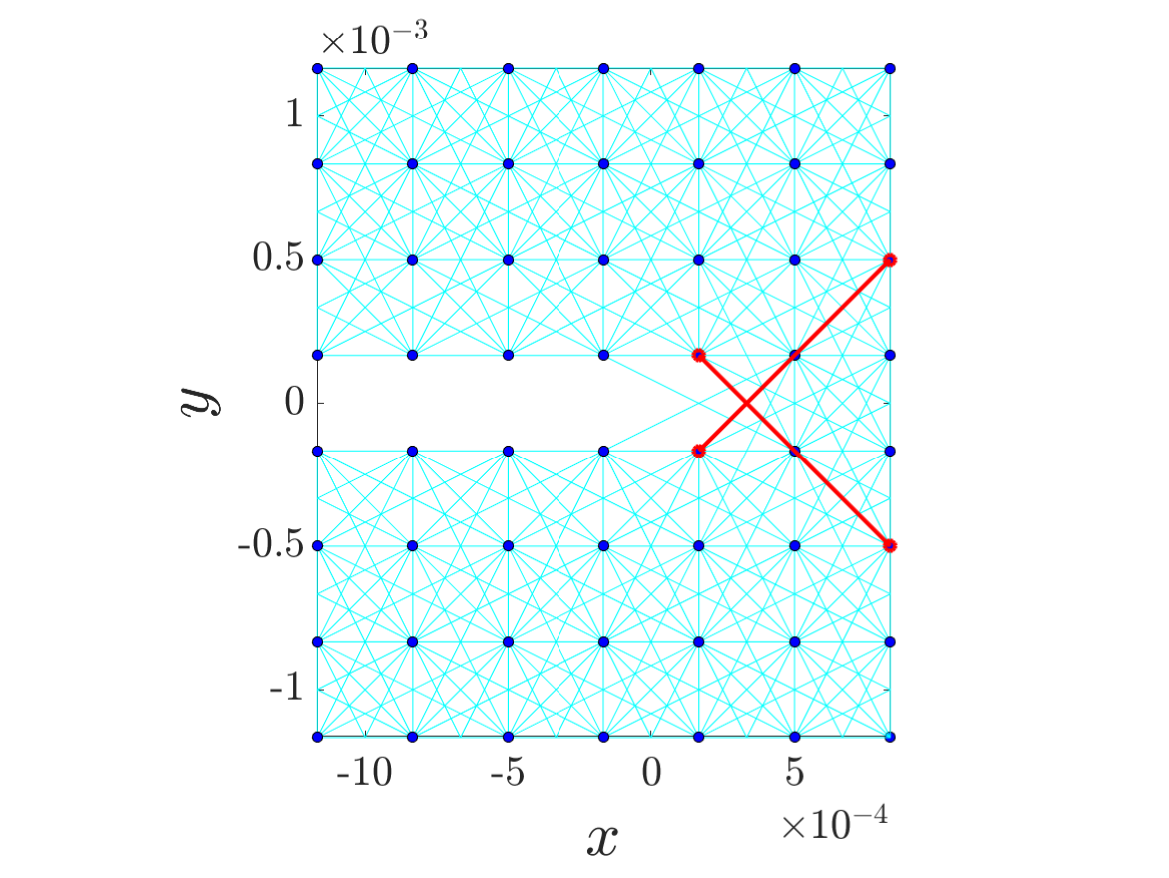}
        \caption{step $= 1001$}
    \label{Fig: Example 2b results (b)}
    \end{subfigure}
        \begin{subfigure}[t]{0.49\textwidth}
        \centering
       \includegraphics[width=\width\textwidth, trim=0cm 0cm 0cm 0cm,clip]{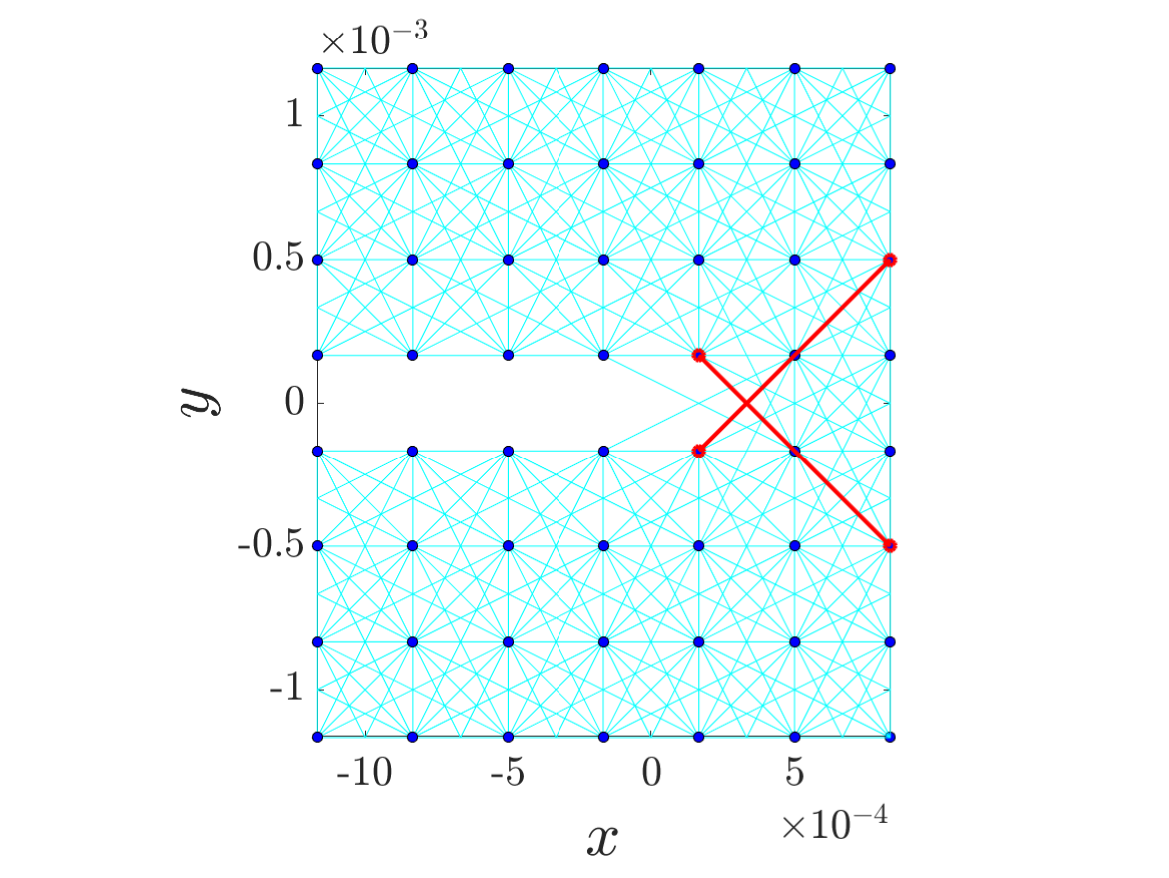}
        \caption{step $= 1036$}
    \label{Fig: Example 2b results (c)}
    \end{subfigure}%
     \vspace{0.5pc}
    \begin{subfigure}[t]{0.49\textwidth}
        \centering
       \includegraphics[width=\width\textwidth, trim=0cm 0cm 0cm 0cm,clip]{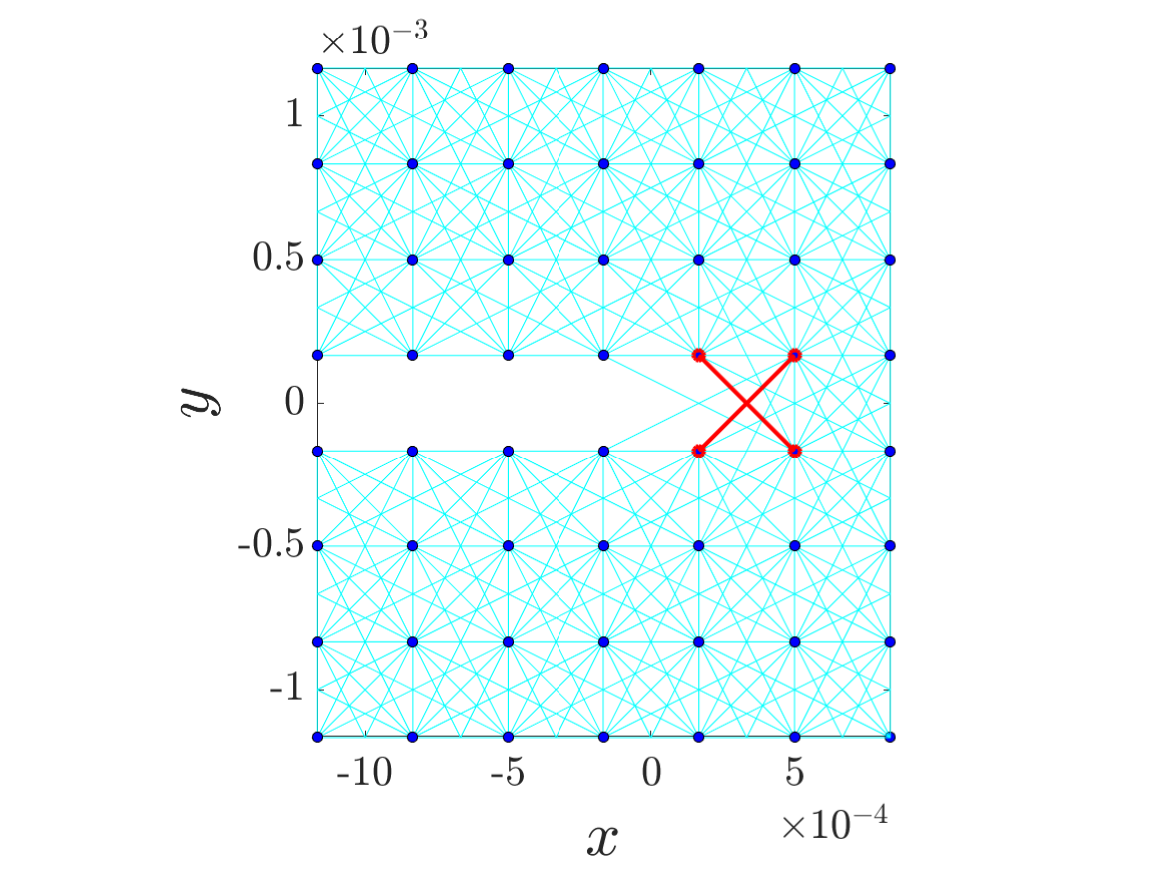}
        \caption{step $= 1003$}
    \label{Fig: Example 2b results (d)}
    \end{subfigure}
            \begin{subfigure}[t]{0.49\textwidth}
        \centering
       \includegraphics[width=\width\textwidth, trim=0cm 0cm 0cm 0cm,clip]{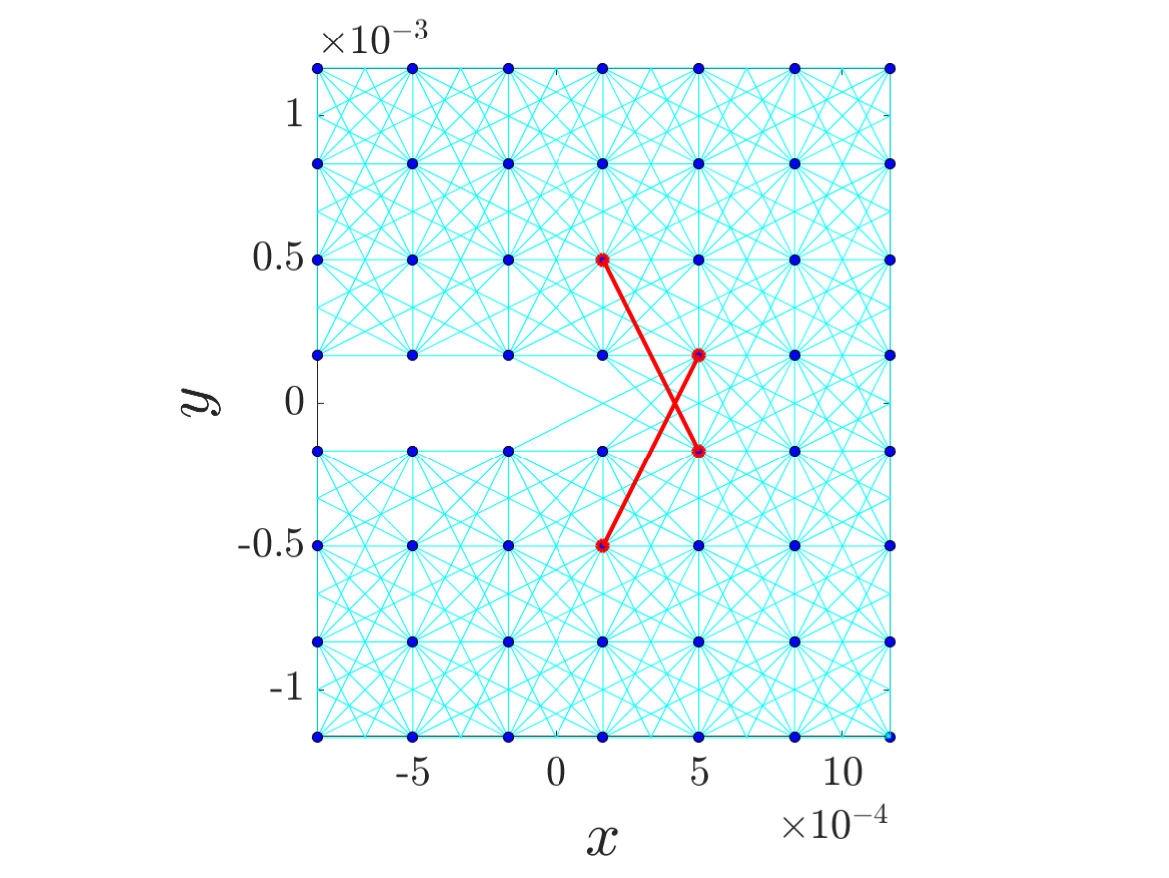}
        \caption{step $= 1041$}
    \label{Fig: Example 2b results (e)}
    \end{subfigure}%
     \vspace{0.5pc}
    \begin{subfigure}[t]{0.49\textwidth}
        \centering
        \includegraphics[width=\width\textwidth, trim=0cm 0cm 0cm 0cm,clip]{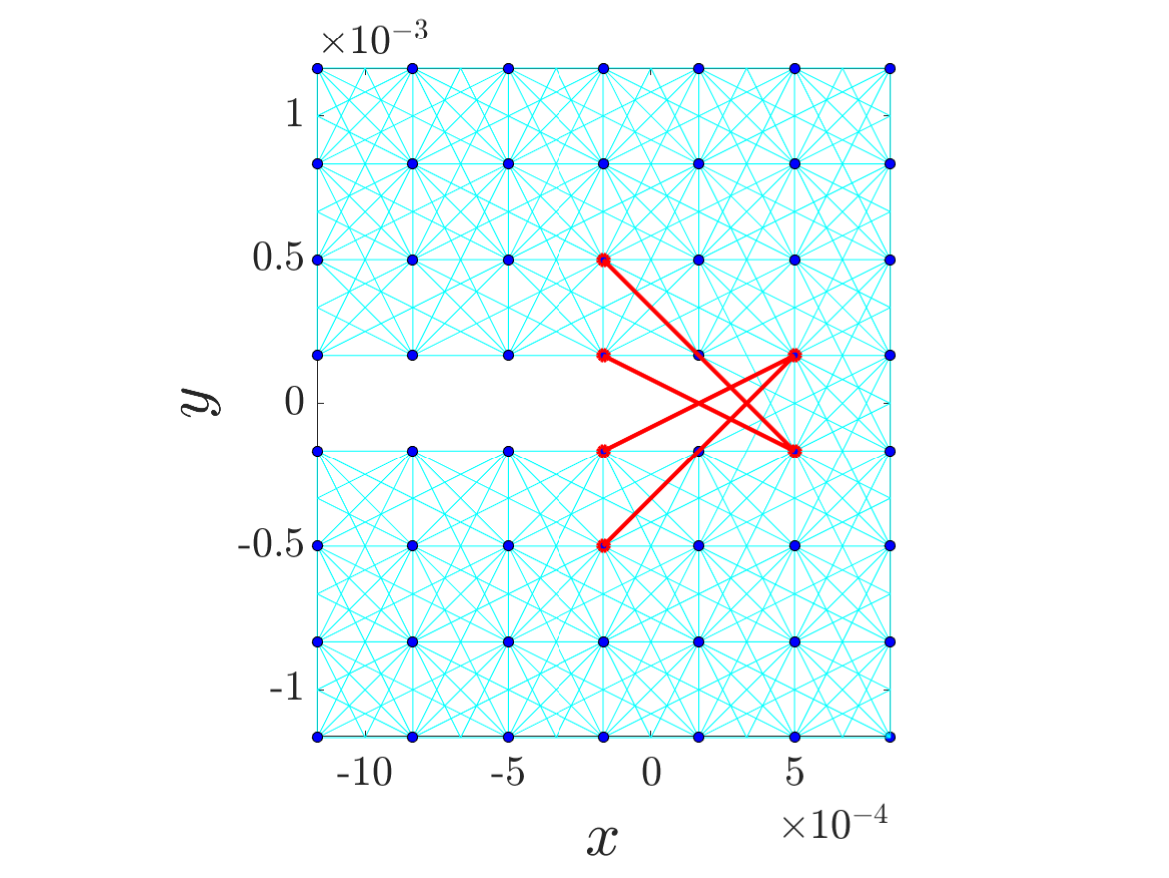}
        \caption{step $= 1005$}
    \label{Fig: Example 2b results (f)}
    \end{subfigure}
        \begin{subfigure}[t]{0.49\textwidth}
        \centering
       \includegraphics[width=\width\textwidth, trim=0cm 0cm 0cm 0cm,clip]{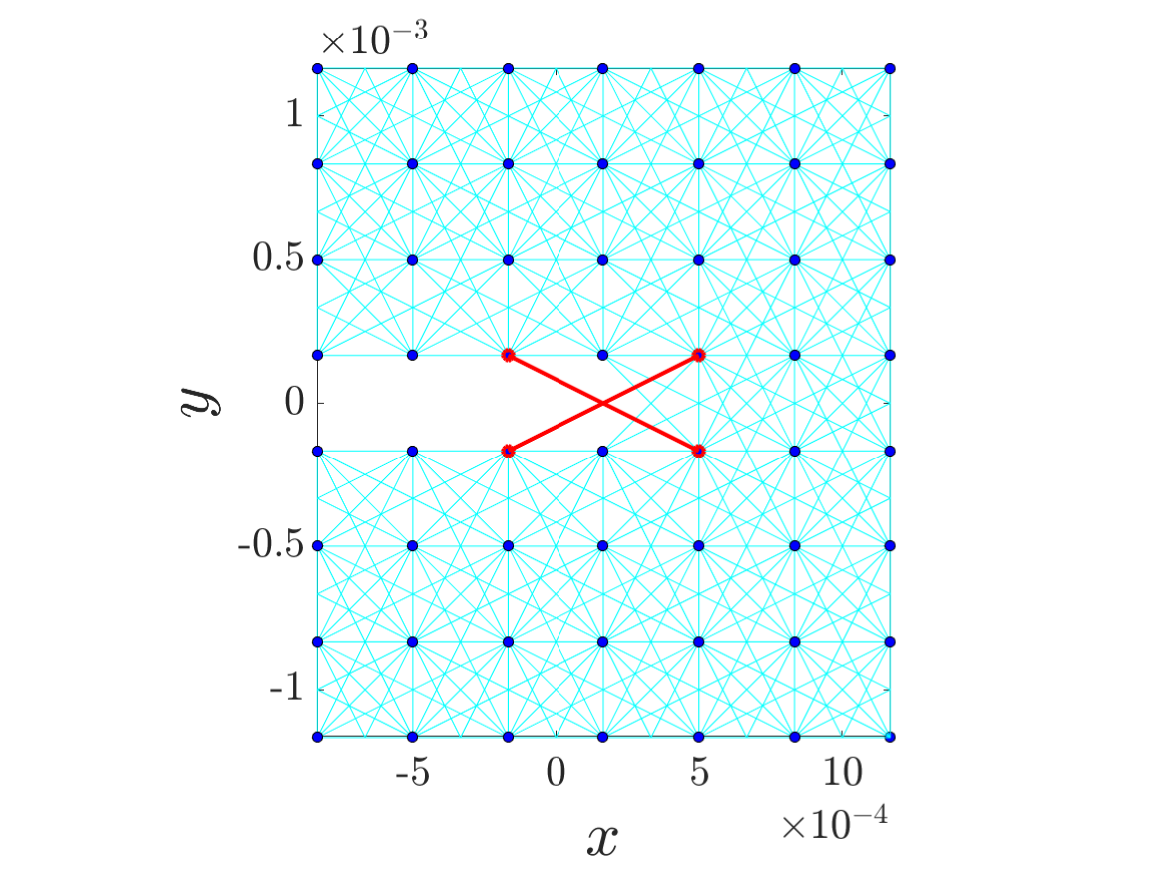}
        \caption{step $= 1042$}
    \label{Fig: Example 2b results (g)}
    \end{subfigure}%
     \vspace{0.5pc}
    \begin{subfigure}[t]{0.49\textwidth}
        \centering
        \includegraphics[width=\width\textwidth, trim=0cm 0cm 0cm 0cm,clip]{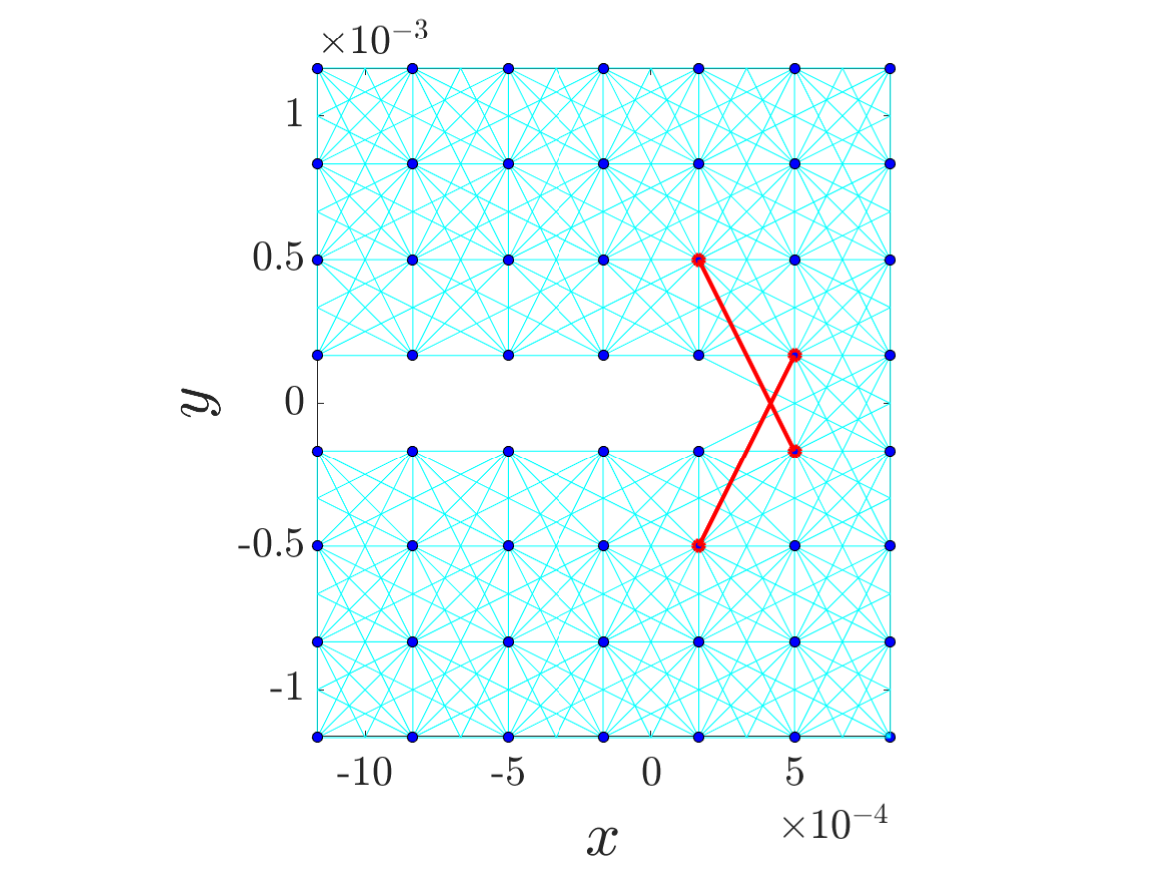}
        \caption{step $= 1006$}
    \label{Fig: Example 2b results (h)}
    \end{subfigure}
            \begin{subfigure}[t]{0.49\textwidth}
        \centering
        \includegraphics[width=\width\textwidth, trim=0cm 0cm 0cm 0cm,clip]{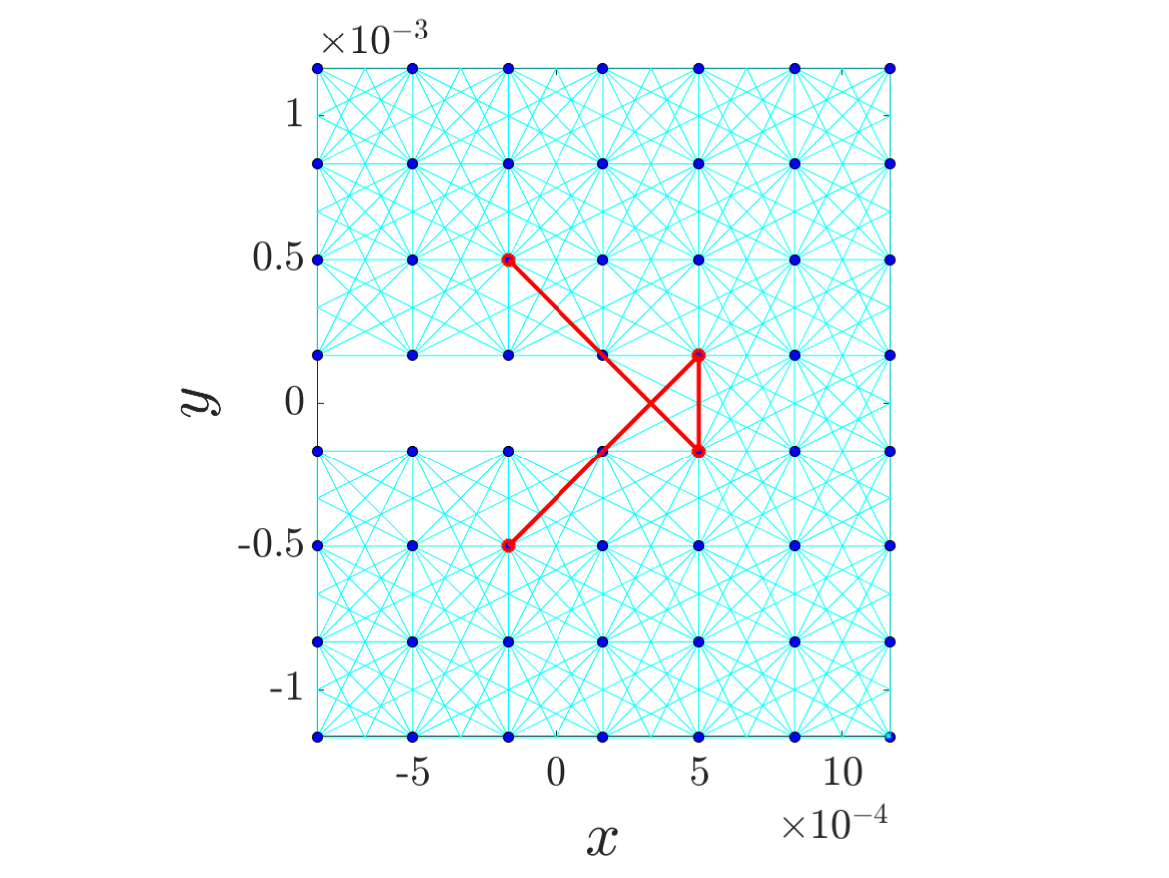}
        \caption{step $= 1043$}
    \label{Fig: Example 2b results (i)}
    \end{subfigure}%
     \vspace{0.5pc}
    \begin{subfigure}[t]{0.49\textwidth}
        \centering
        \includegraphics[width=\width\textwidth, trim=0cm 0cm 0cm 0cm,clip]{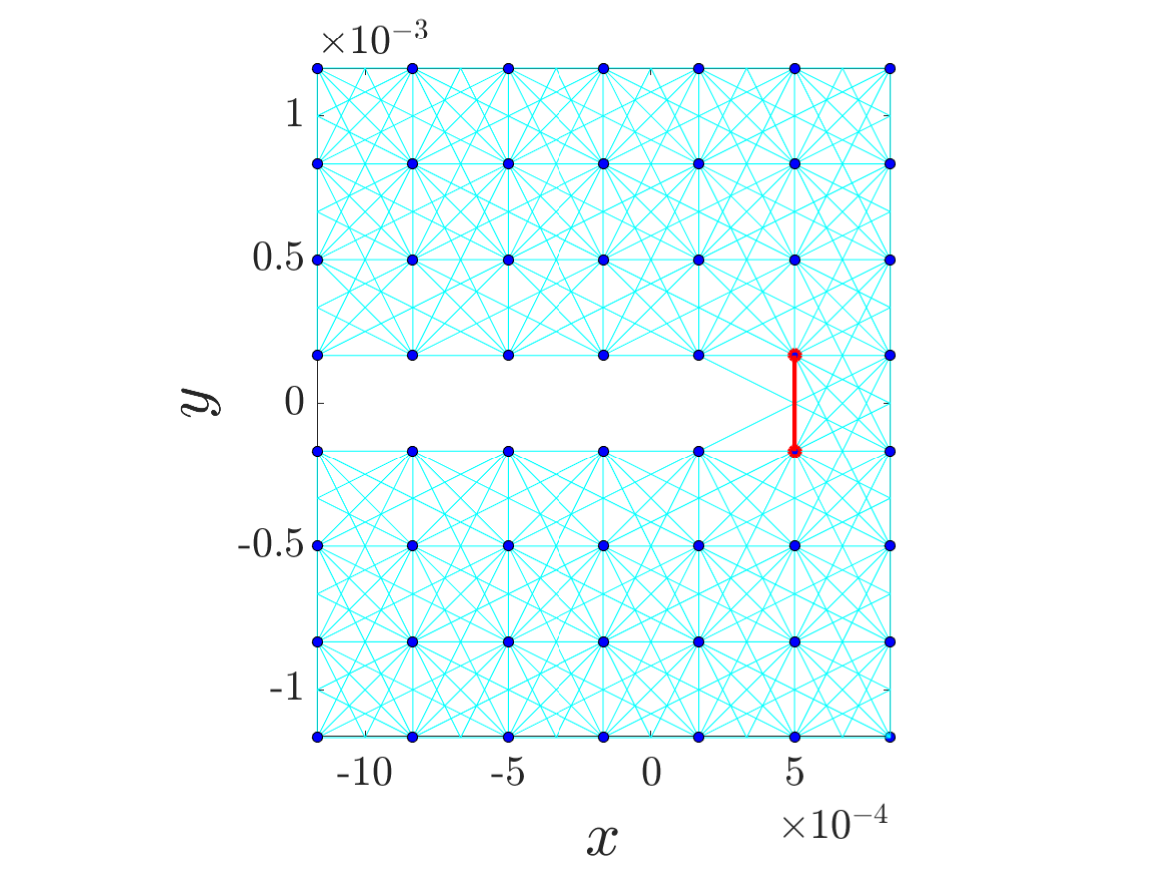}
        \caption{step $= 1008$}
    \label{Fig: Example 2b results (j)}
    \end{subfigure}
        \caption{Comparison of the crack tip evolution between the two bond-failure criteria for $\alpha = 0$ in Example 2 (cont.).}
    \label{Fig: Example 2b results}
\end{figure*} 

\begin{table}[htbp!]
\centering
{\tabulinesep=1.5mm
\begin{tabu}{||c | c | c c |  c| c c |} 
 \hline
    & Step   & \multicolumn{2}{c|}{\bf Critical stretch criterion} & Step    & \multicolumn{2}{c|}{\bf Critical energy density criterion}  \\ [0.5ex] 
    \cline{3-4} \cline{6-7} 
   &                   &          Bond & Bond length               &                    & Bond & Bond length  \\ [0.5ex] 
 \hline\hline
$1$ & $928$ & $1$:~$[17{,}550 \; ; \; 18{,}151]$ & $\sqrt{5}\Delta x$ & $930$ & $1$:~$[17{,}550 \; ; \; 18{,}151]$ & $\sqrt{5}\Delta x$ \\
       &            & $2$:~$[17{,}851 \; ; \; 18{,}450]$ & $\sqrt{5}\Delta x$&            & $2$:~$[17{,}851 \; ; \; 18{,}450]$ & $\sqrt{5}\Delta x$ \\
\hline
$2$ & $929$ & $1$:~$[17{,}851 \; ; \; 18{,}151]$ & $\Delta x$ & $933$ & $1$:~$[17{,}851 \; ; \; 18{,}151]$ & $\Delta x$\\ 
\hline
$3$ & $930$ & $1$:~$[17{,}551 \; ; \; 18{,}151]$ & $2\Delta x$ & $935$ & $1$:~$[17{,}251 \; ; \; 18{,}151]$ & $3\Delta x$ \\ 
       &            & $2$:~$[17{,}851 \; ; \; 18{,}451]$ & $2\Delta x$ &            & $2$:~$[17{,}551 \; ; \; 18{,}151]$ & $2\Delta x$ \\ 
       &            &                                                     &                   &            & $3$:~$[17{,}851 \; ; \; 18{,}451]$ & $2\Delta x$ \\
       &            &                                                     &                   &            & $4$:~$[17{,}851 \; ; \; 18{,}751]$ & $3\Delta x$ \\
\hline
$4$ & $932$ & $1$:~$[17{,}251 \; ; \; 18{,}151]$ & $3\Delta x$ & $936$ & $1$:~$[17{,}551 \; ; \; 18{,}451]$ & $3\Delta x$ \\ 
       &            & $2$:~$[17{,}551 \; ; \; 18{,}451]$ & $3\Delta x$ & &   & \\ 
       &            & $3$:~$[17{,}851 \; ; \; 18{,}751]$ & $3\Delta x$ & &   & \\ 
\hline
$5$ & $933$ & $1$:~$[17{,}552 \; ; \; 18{,}151]$ & $\sqrt{5}\Delta x$ & $937$ & $1$:~$[17{,}552 \; ; \; 18{,}151]$ & $\sqrt{5}\Delta x$ \\ 
       &            & $2$:~$[17{,}851 \; ; \; 18{,}452]$ & $\sqrt{5}\Delta x$ &            & $2$:~$[17{,}851 \; ; \; 18{,}452]$ & $\sqrt{5}\Delta x$ \\ 
\hline
$6$ & $935$ & $1$:~$[17{,}851 \; ; \; 18{,}152]$ & $\sqrt{2}\Delta x$ & $1001$ & $1$:~$[17{,}553 \; ; \; 18{,}151]$ & $2\sqrt{2}\Delta x$ \\ 
       &            & $2$:~$[17{,}852 \; ; \; 18{,}151]$ & $\sqrt{2}\Delta x$ &            & $2$:~$[17{,}851 \; ; \; 18{,}453]$ & $2\sqrt{2}\Delta x$ \\ 
\hline
$7$ & $1036$ & $1$:~$[17{,}553 \; ; \; 18{,}151]$ & $2\sqrt{2}\Delta x$ & $1003$ & $1$:~$[17{,}851 \; ; \; 18{,}152]$ & $\sqrt{2}\Delta x$ \\ 
       &            & $2$:~$[17{,}851 \; ; \; 18{,}453]$ & $2\sqrt{2}\Delta x$ &            & $2$:~$[17{,}852 \; ; \; 18{,}151]$ & $\sqrt{2}\Delta x$ \\ 
\hline
$8$ & $1041$ & $1$:~$[17{,}551 \; ; \; 18{,}152]$ & $\sqrt{5}\Delta x$ & $1005$ & $1$:~$[17{,}550 \; ; \; 18{,}152]$ & $2\sqrt{2}\Delta x$ \\ 
       &            & $2$:~$[17{,}852 \; ; \; 18{,}451]$ & $\sqrt{5}\Delta x$ &            & $2$:~$[17{,}850 \; ; \; 18{,}152]$ & $\sqrt{5}\Delta x$ \\ 
       &            &                                                     &                   &            & $3$:~$[17{,}852 \; ; \; 18{,}150]$ & $\sqrt{5}\Delta x$ \\
       &            &                                                     &                   &            & $4$:~$[17{,}852 \; ; \; 18{,}450]$ & $2\sqrt{2}\Delta x$ \\
\hline
$9$ & $1042$ & $1$:~$[17{,}850 \; ; \; 18{,}152]$ & $\sqrt{5}\Delta x$ & $1006$ & $1$:~$[17{,}551 \; ; \; 18{,}152]$ & $\sqrt{5}\Delta x$ \\ 
       &            & $2$:~$[17{,}852 \; ; \; 18{,}150]$ & $\sqrt{5}\Delta x$ &            & $2$:~$[17{,}852 \; ; \; 18{,}451]$ & $\sqrt{5}\Delta x$\\ 
\hline
$10$ & $1043$ & $1$:~$[17{,}550 \; ; \; 18{,}152]$ & $2\sqrt{2}\Delta x$ & $1008$ & $1$:~$[17{,}852 \; ; \; 18{,}152]$ & $\Delta x$ \\ 
       &            & $2$:~$[17{,}852 \; ; \; 18{,}152]$ &  $\Delta x$ & &   & \\ 
       &            & $3$:~$[17{,}852 \; ; \; 18{,}450]$ &  $2\sqrt{2}\Delta x$ & &   & \\ 
 \hline
\end{tabu}
}
\caption{List of broken bonds and corresponding bond lengths for the first $10$ stages of bond breaking, for the case of $\alpha = 0$, for the two bond-failure criteria in Example 2 ({\em cf.}~Figures~\ref{Fig: Example 2a results} and~\ref{Fig: Example 2b results}). Bonds are listed with the numbers of the pair of computational nodes they connect.}
\label{table: Example 2 broken bonds alpha 0.}
\end{table}


\begin{figure*}[htbp!]
    \centering
    {\bf The case}~\boldsymbol{$\alpha = 2$}\\[0.1in]
    \begin{subfigure}[t]{0.49\textwidth}
        \centering
      {\bf Critical stretch}\\[0.1in]
       \includegraphics[width=\width\textwidth, trim=0cm 0cm 0cm 0cm,clip]{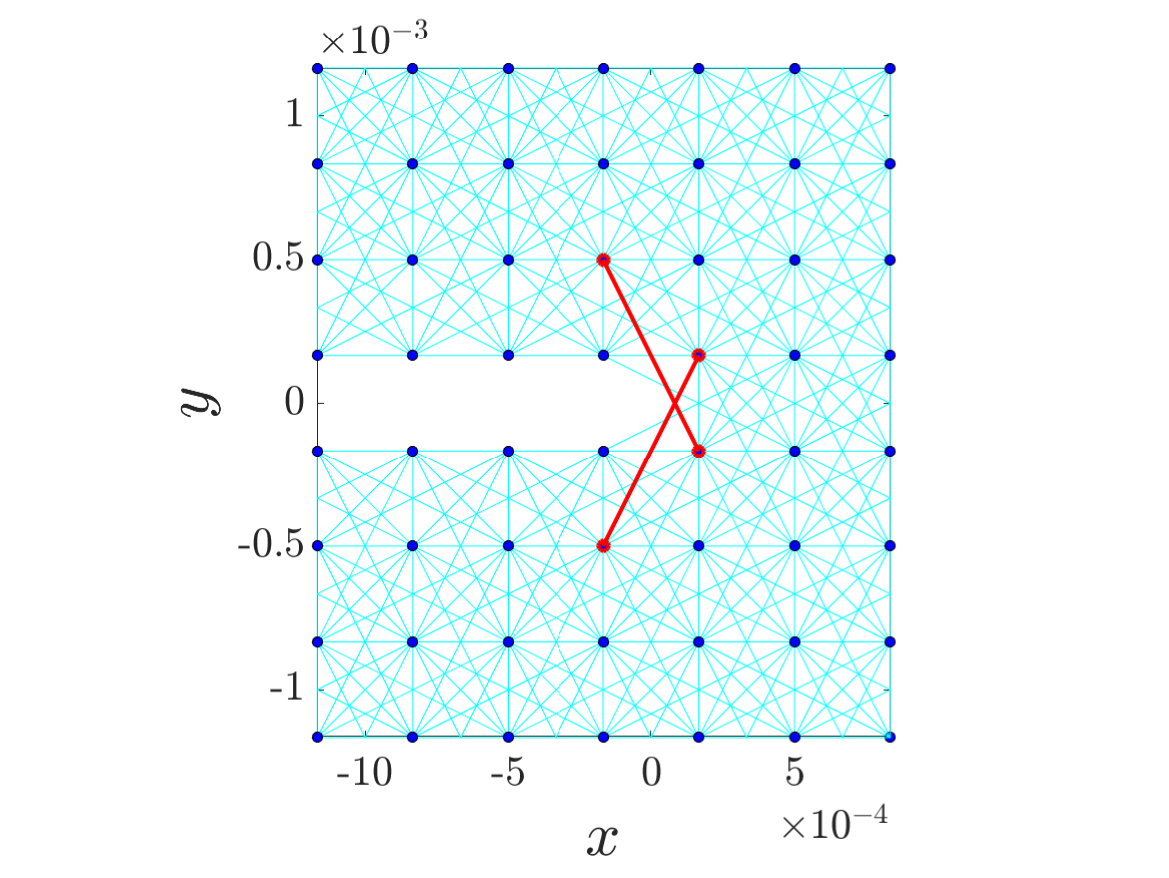}
        \caption{step $= 1022$}
    \label{Fig: Example 2c results (a)}
    \end{subfigure}%
     \vspace{0.5pc}
    \begin{subfigure}[t]{0.49\textwidth}
        \centering
    {\bf Critical energy density  }\\[0.1in]
        \includegraphics[width=\width\textwidth, trim=0cm 0cm 0cm 0cm,clip]{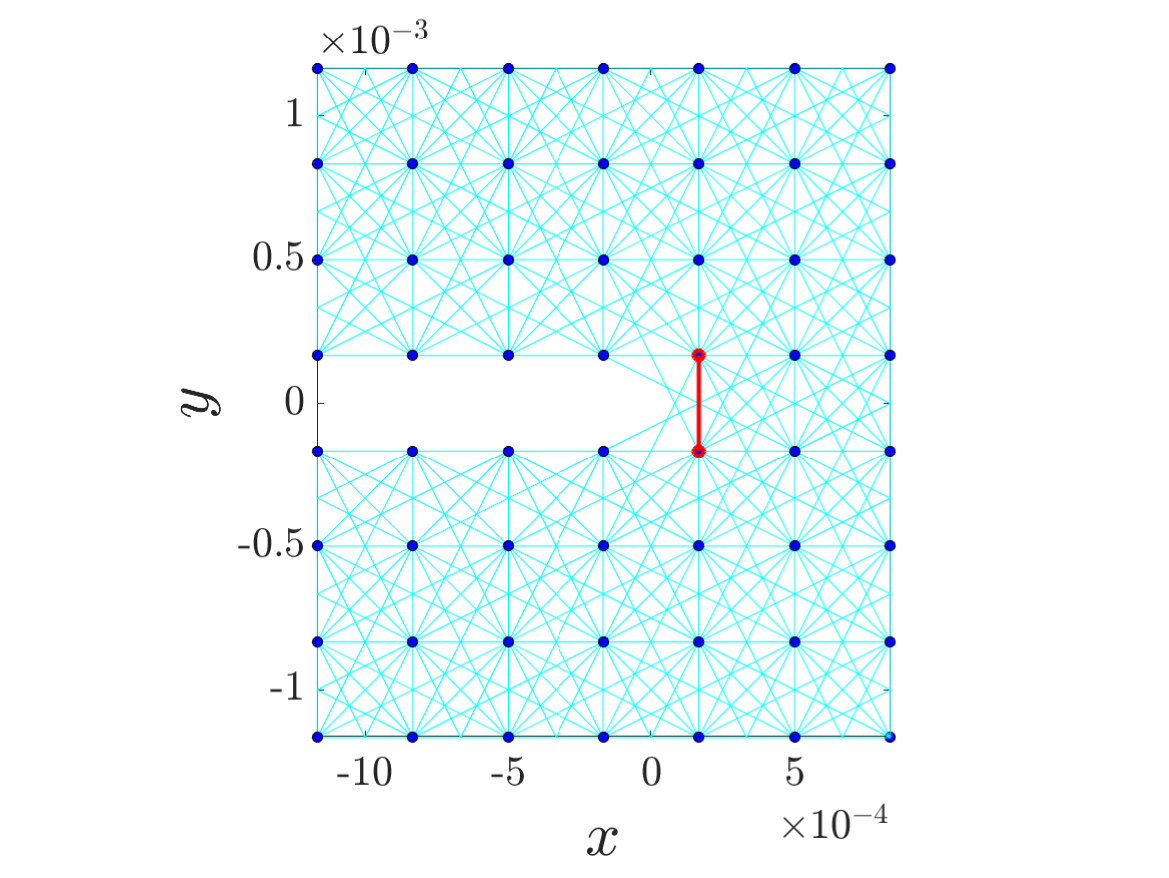}
        \caption{step $= 932$}
    \label{Fig: Example 2c results (b)}
    \end{subfigure}
        \begin{subfigure}[t]{0.49\textwidth}
        \centering
        \includegraphics[width=\width\textwidth, trim=0cm 0cm 0cm 0cm,clip]{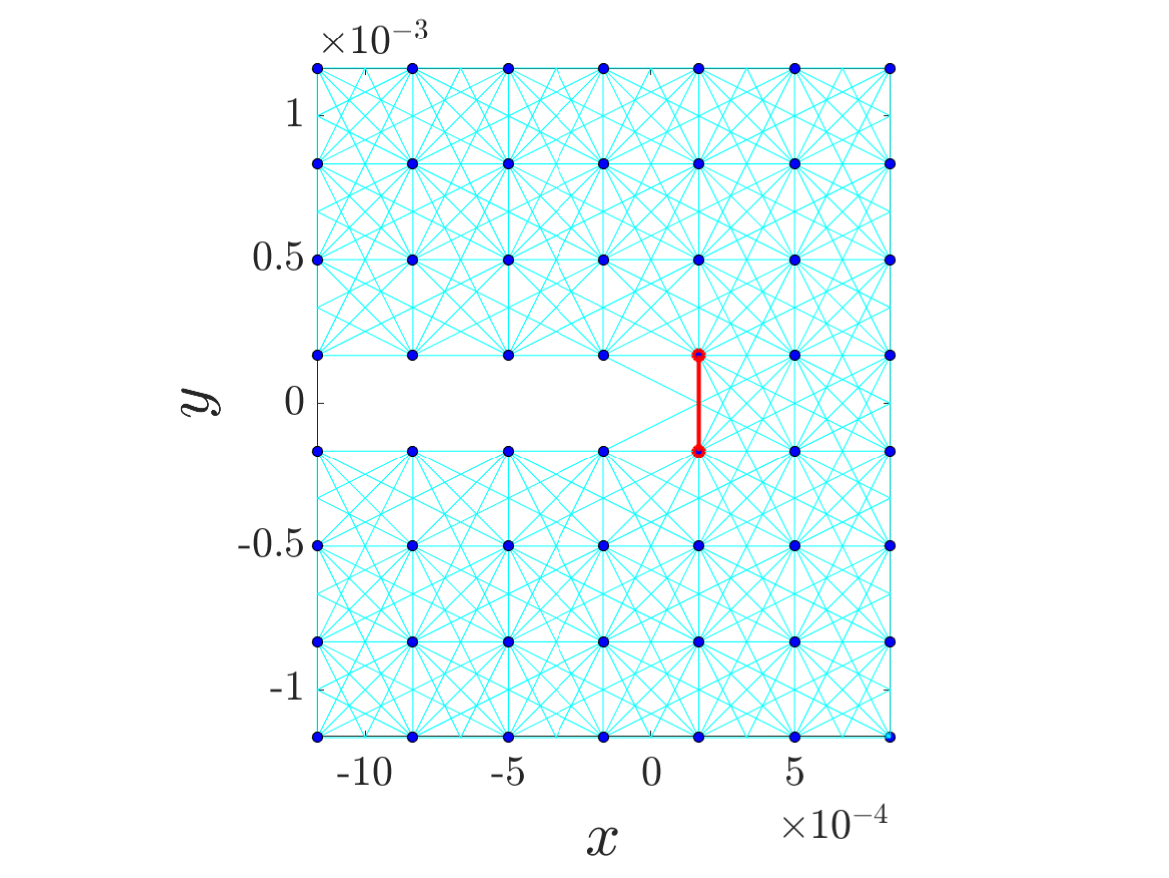}
        \caption{step $= 1024$}
    \label{Fig: Example 2c results (c)}
    \end{subfigure}%
     \vspace{0.5pc}
    \begin{subfigure}[t]{0.49\textwidth}
        \centering
        \includegraphics[width=\width\textwidth, trim=0cm 0cm 0cm 0cm,clip]{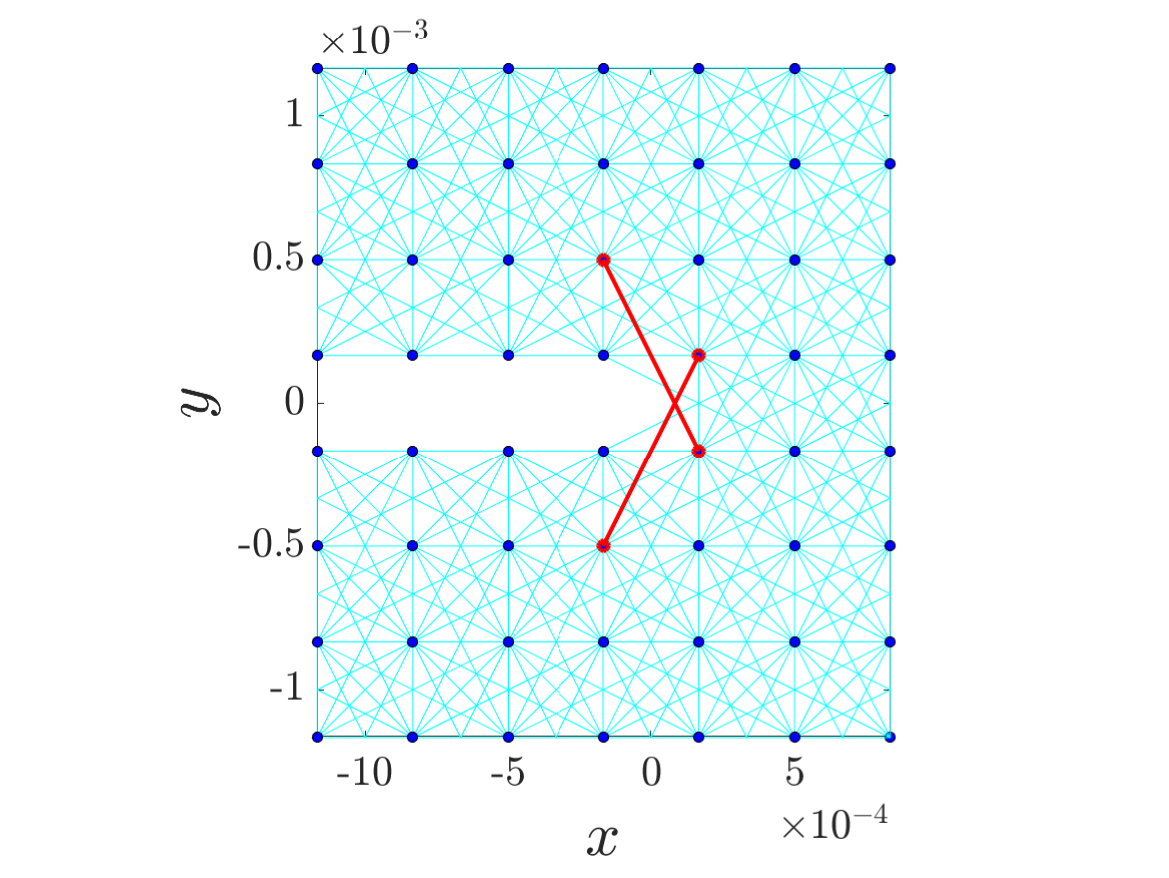}
        \caption{step $= 934$}
    \label{Fig: Example 2c results (d)}
    \end{subfigure}
            \begin{subfigure}[t]{0.49\textwidth}
        \centering
      \includegraphics[width=\width\textwidth, trim=0cm 0cm 0cm 0cm,clip]{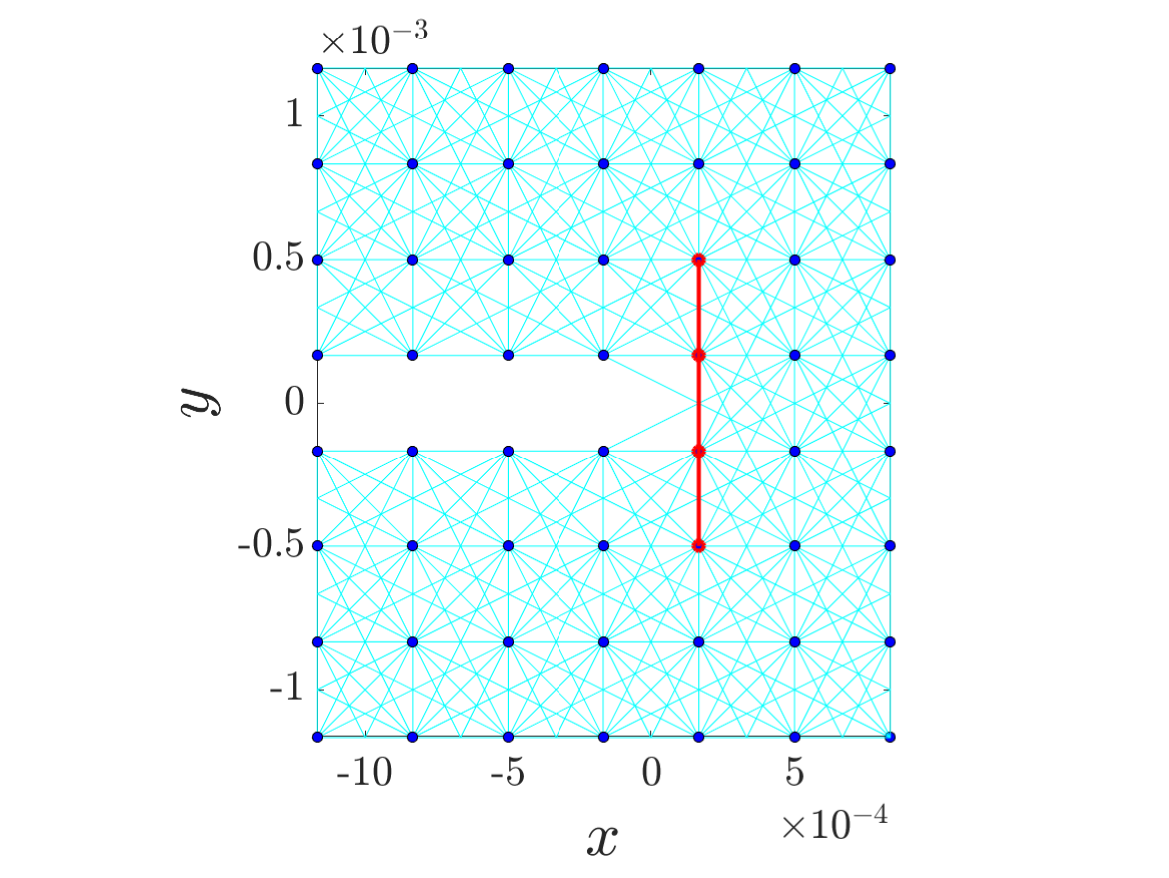}
        \caption{step $= 1025$}
      \label{Fig: Example 2c results (e)}
    \end{subfigure}%
     \vspace{0.5pc}
    \begin{subfigure}[t]{0.49\textwidth}
        \centering
        \includegraphics[width=\width\textwidth, trim=0cm 0cm 0cm 0cm,clip]{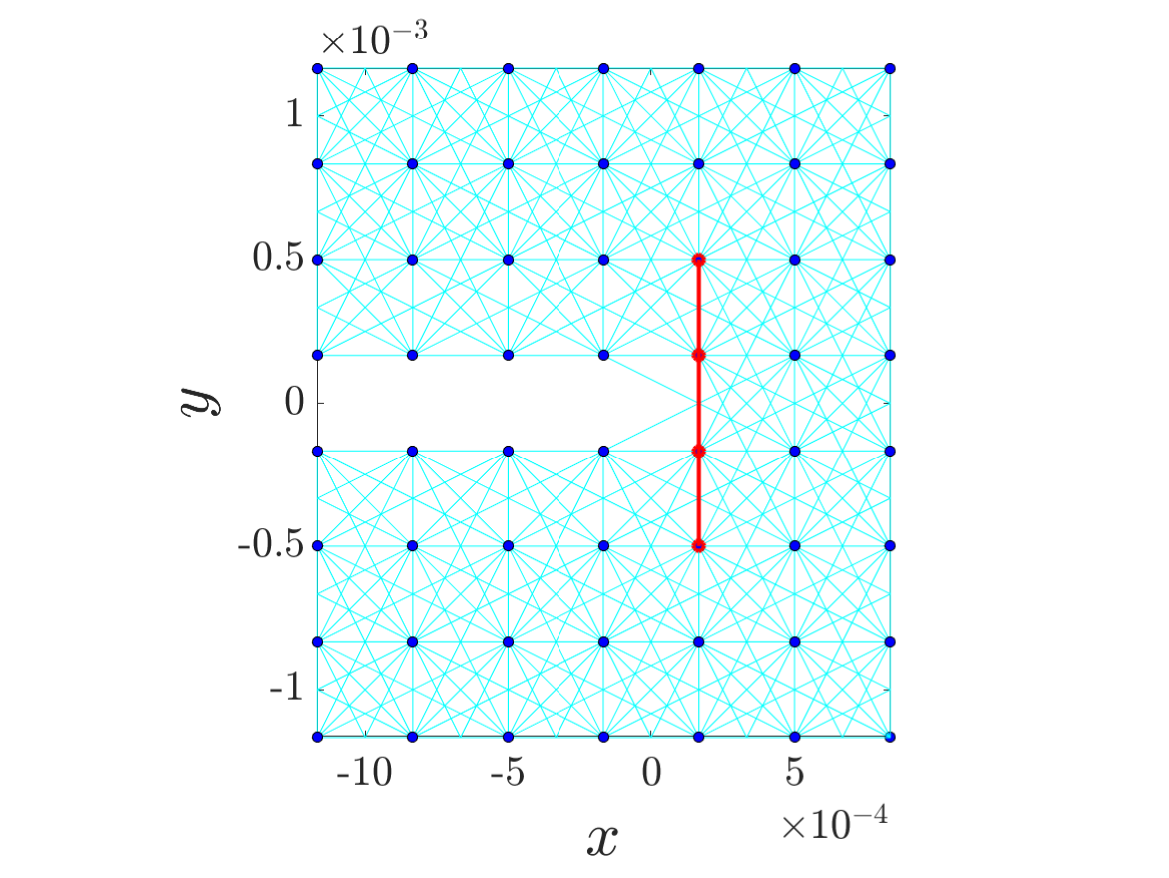}
        \caption{step $= 935$}
    \label{Fig: Example 2c results (f)}
    \end{subfigure}
                \begin{subfigure}[t]{0.49\textwidth}
        \centering
       \includegraphics[width=\width\textwidth, trim=0cm 0cm 0cm 0cm,clip]{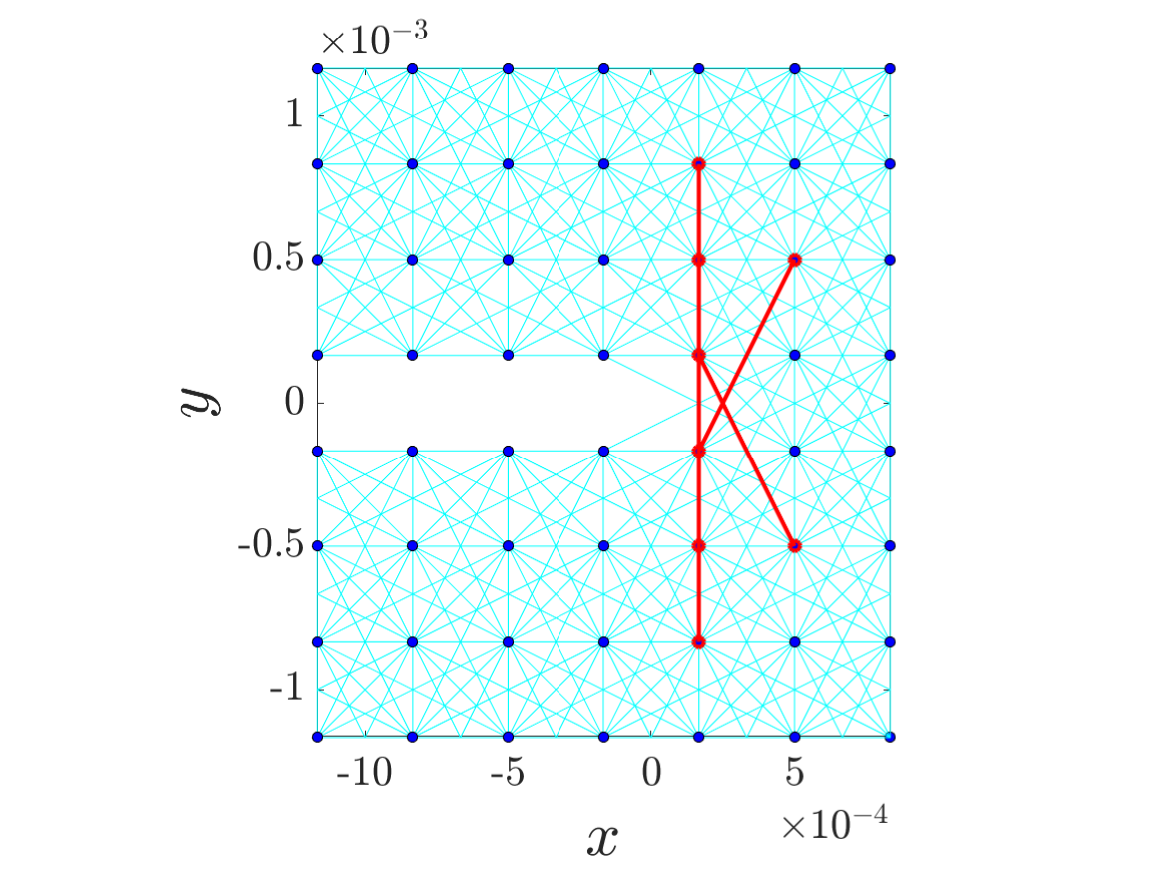}
        \caption{step $= 1027$}
    \label{Fig: Example 2c results (g)}
    \end{subfigure}%
     \vspace{0.5pc}
    \begin{subfigure}[t]{0.49\textwidth}
        \centering
       \includegraphics[width=\width\textwidth, trim=0cm 0cm 0cm 0cm,clip]{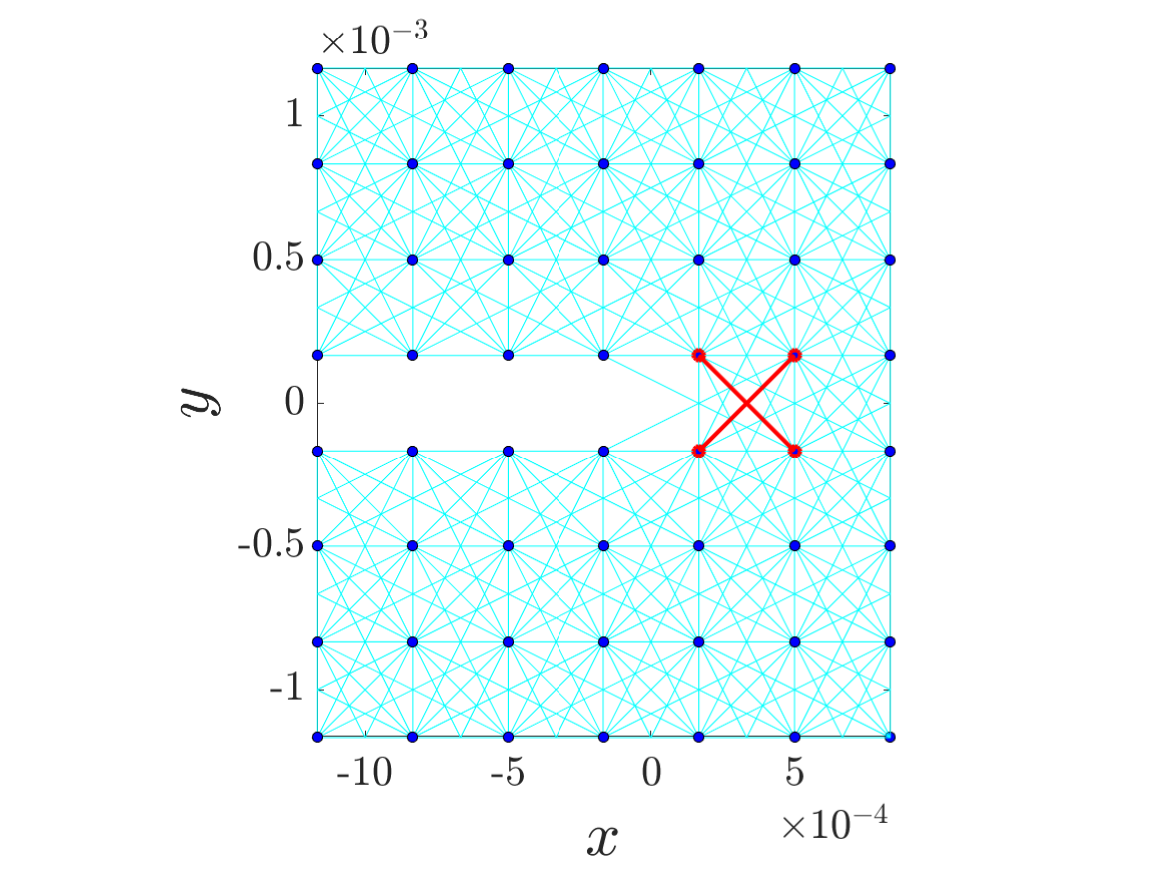}
        \caption{step $= 1010$}
    \label{Fig: Example 2c results (h)}
    \end{subfigure}
                    \begin{subfigure}[t]{0.49\textwidth}
        \centering
        \includegraphics[width=\width\textwidth, trim=0cm 0cm 0cm 0cm,clip]{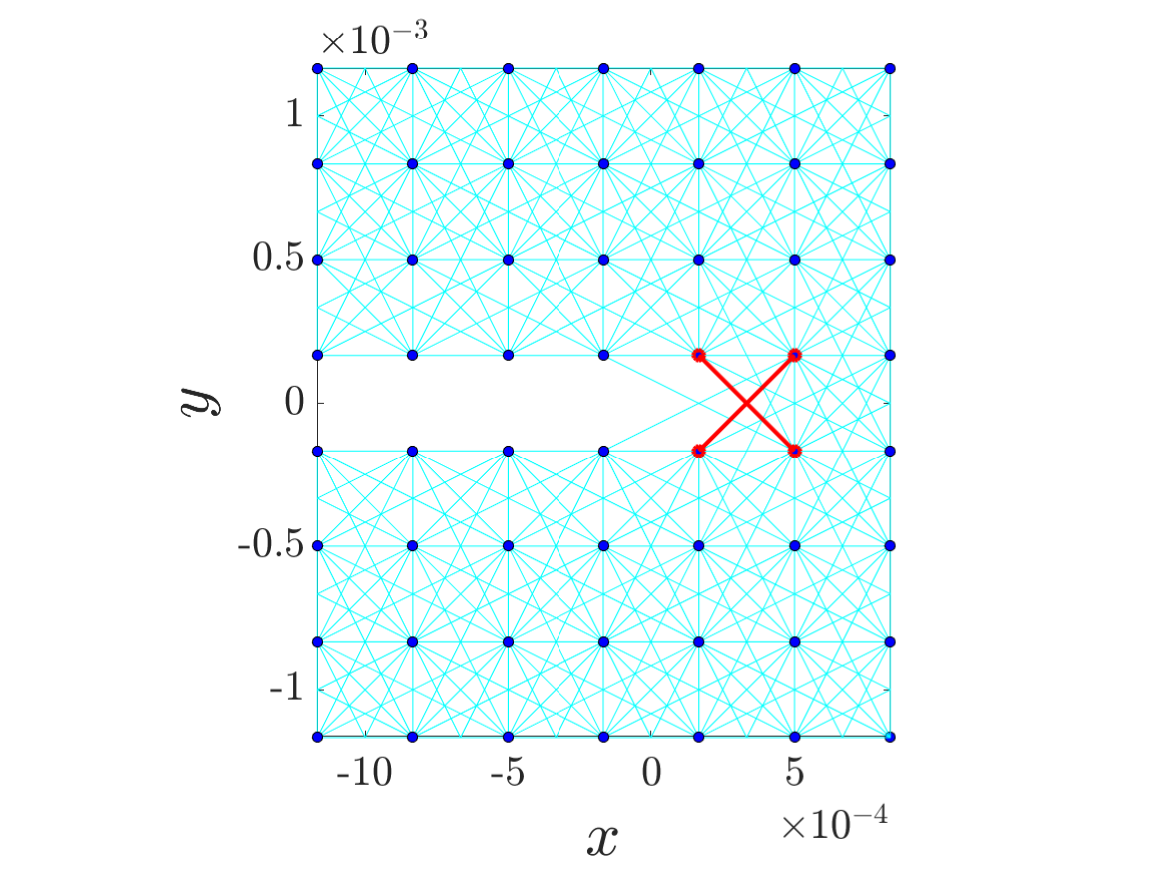}
        \caption{step $= 1028$}
    \label{Fig: Example 2c results (i)}
    \end{subfigure}%
     \vspace{0.5pc}
    \begin{subfigure}[t]{0.49\textwidth}
        \centering
        \includegraphics[width=\width\textwidth, trim=0cm 0cm 0cm 0cm,clip]{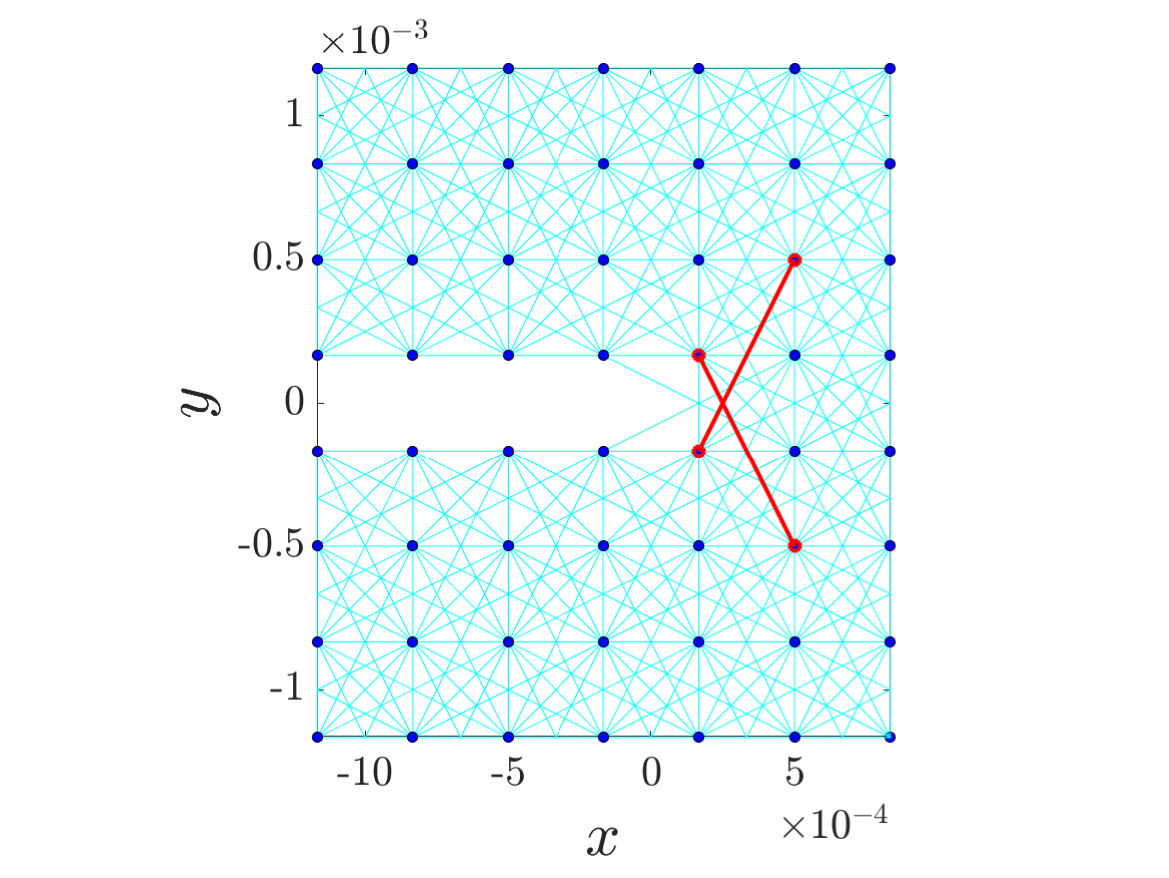}
        \caption{step $= 1011$}
       \label{Fig: Example 2c results (j)}     
    \end{subfigure}
        \caption{Comparison of the crack tip evolution between the two bond-failure criteria for $\alpha = 2$ in Example 2.} 
    \label{Fig: Example 2c results}
\end{figure*}

\begin{figure*}[htbp!]
    \centering
     {\bf The case}~\boldsymbol{$\alpha = 2$} {\bf (cont.)}\\[0.1in]
    \begin{subfigure}[t]{0.49\textwidth}
        \centering
      {\bf Critical stretch}\\[0.1in]
       \includegraphics[width=\width\textwidth, trim=0cm 0cm 0cm 0cm,clip]{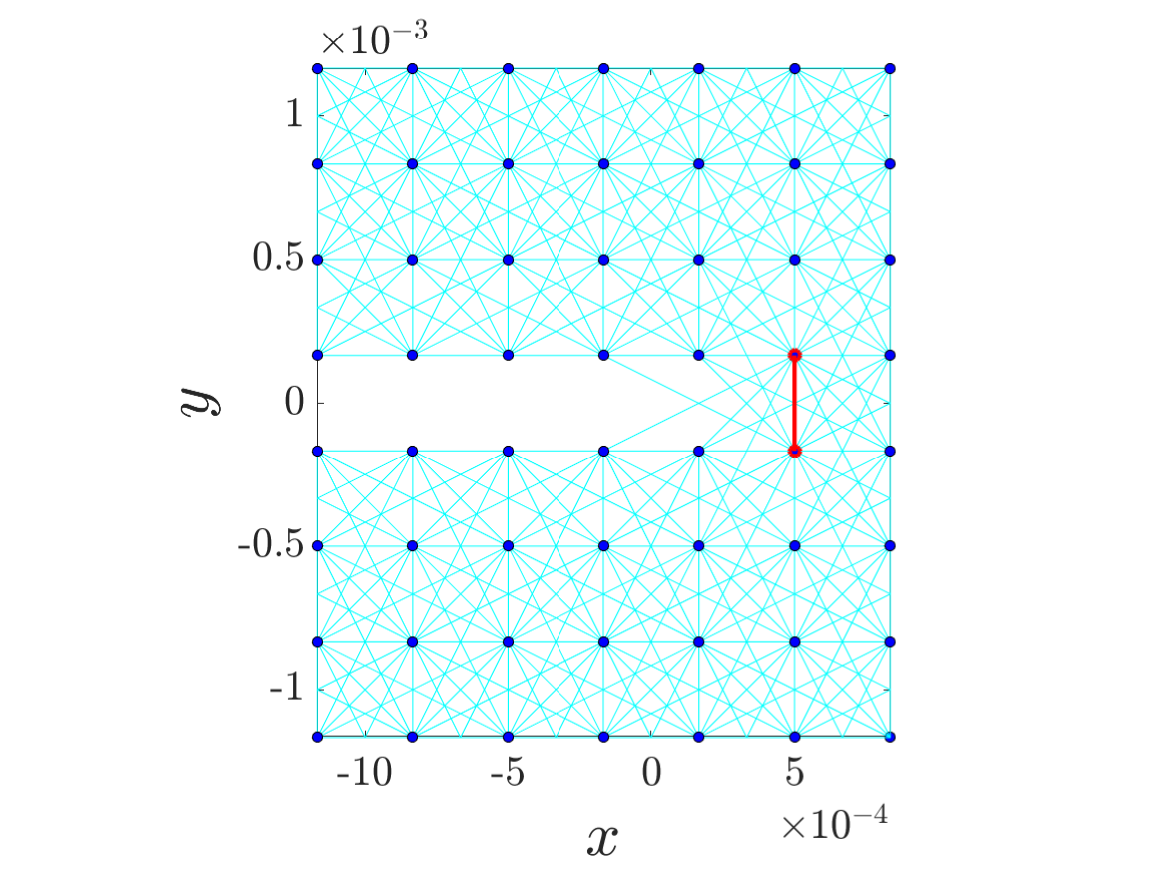}
        \caption{step $= 1030$}
    \label{Fig: Example 2d results (a)}
    \end{subfigure}%
     \vspace{0.5pc}
    \begin{subfigure}[t]{0.49\textwidth}
        \centering
    {\bf Critical energy density}\\[0.1in]
        \includegraphics[width=\width\textwidth, trim=0cm 0cm 0cm 0cm,clip]{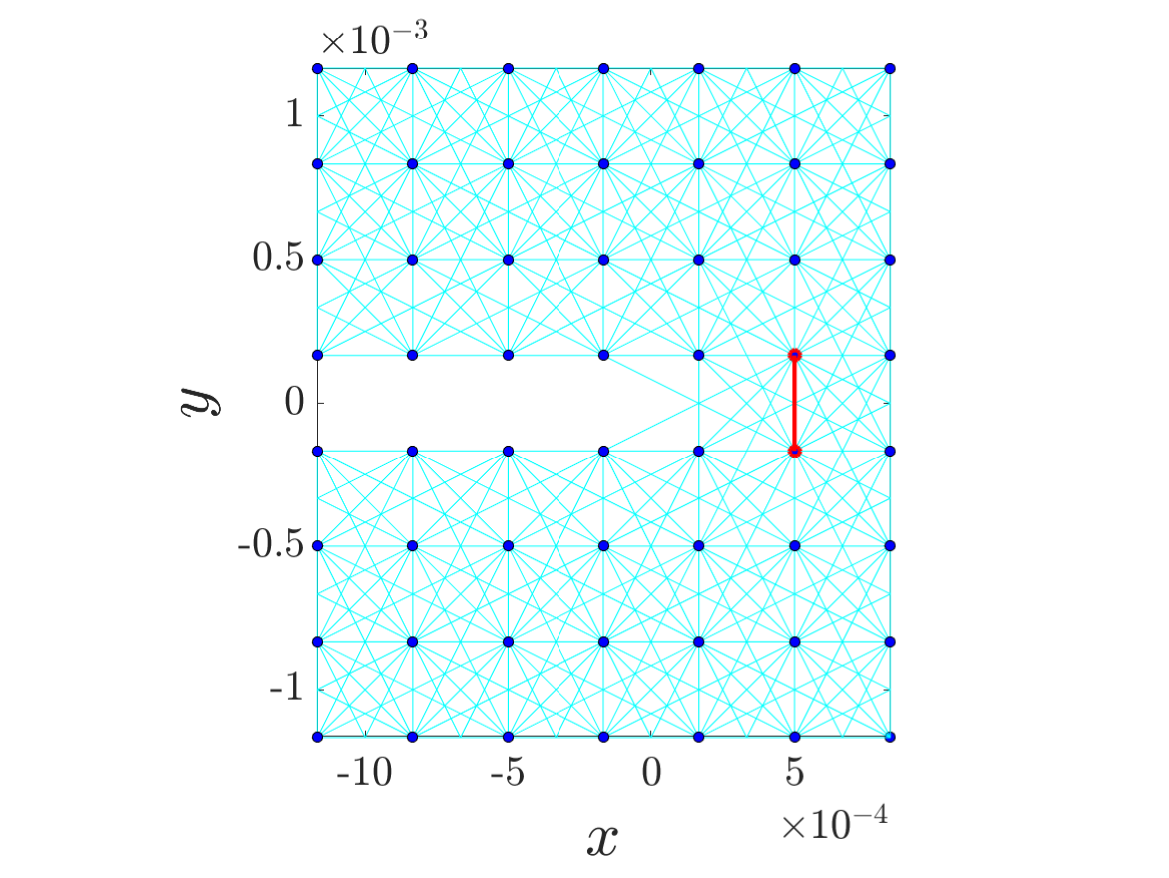}
        \caption{step $= 1012$}
    \label{Fig: Example 2d results (b)}
    \end{subfigure}
        \begin{subfigure}[t]{0.49\textwidth}
        \centering
        \includegraphics[width=\width\textwidth, trim=0cm 0cm 0cm 0cm,clip]{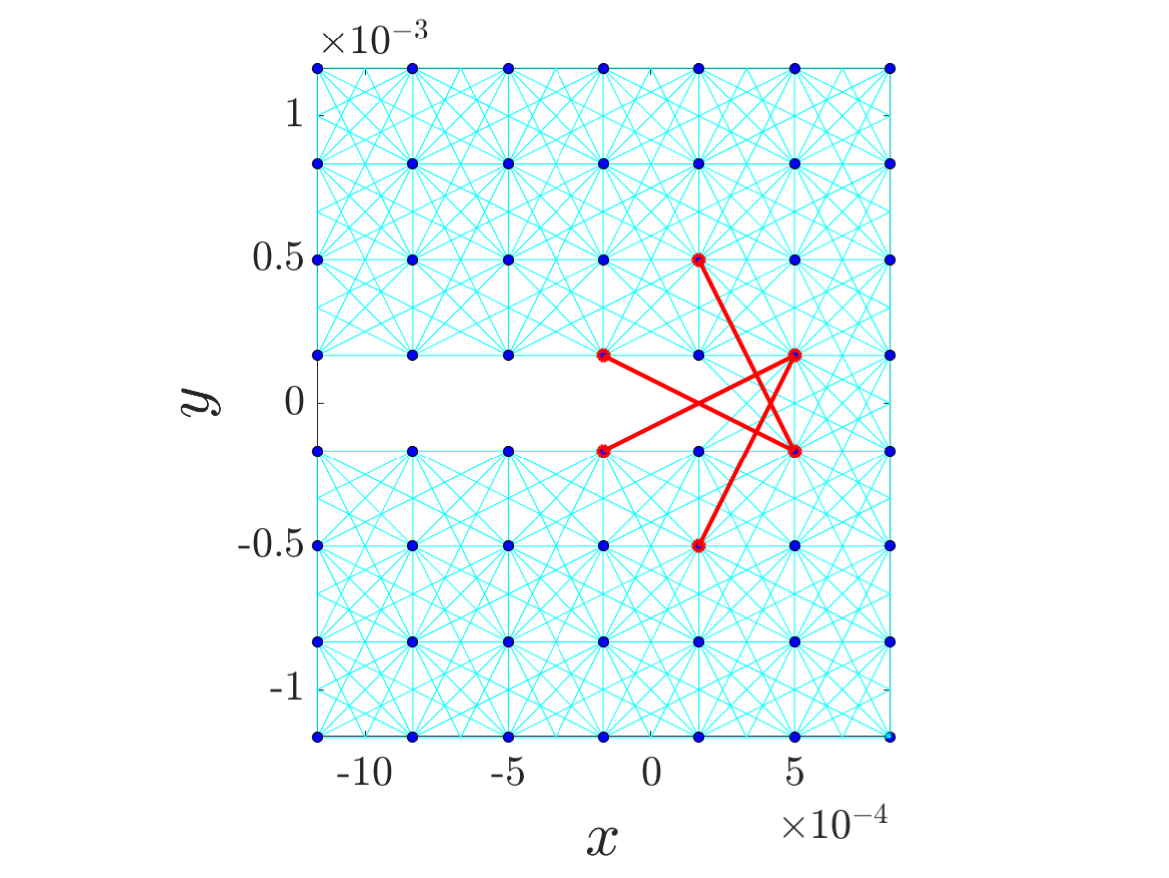}
        \caption{step $= 1031$}
    \label{Fig: Example 2d results (c)}
    \end{subfigure}%
     \vspace{0.5pc}
    \begin{subfigure}[t]{0.49\textwidth}
        \centering
      \includegraphics[width=\width\textwidth, trim=0cm 0cm 0cm 0cm,clip]{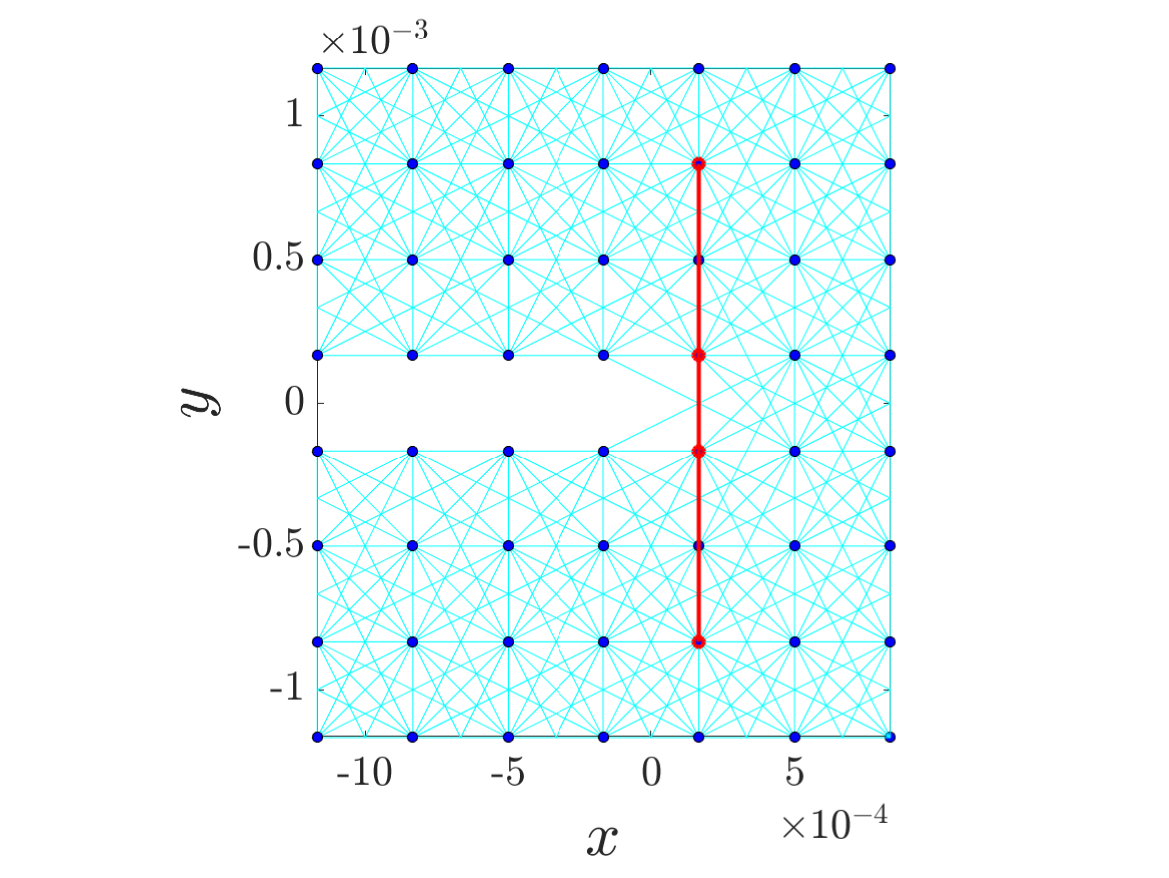}
        \caption{step $= 1013$}
    \label{Fig: Example 2d results (d)}
    \end{subfigure}
            \begin{subfigure}[t]{0.49\textwidth}
        \centering
        \includegraphics[width=\width\textwidth, trim=0cm 0cm 0cm 0cm,clip]{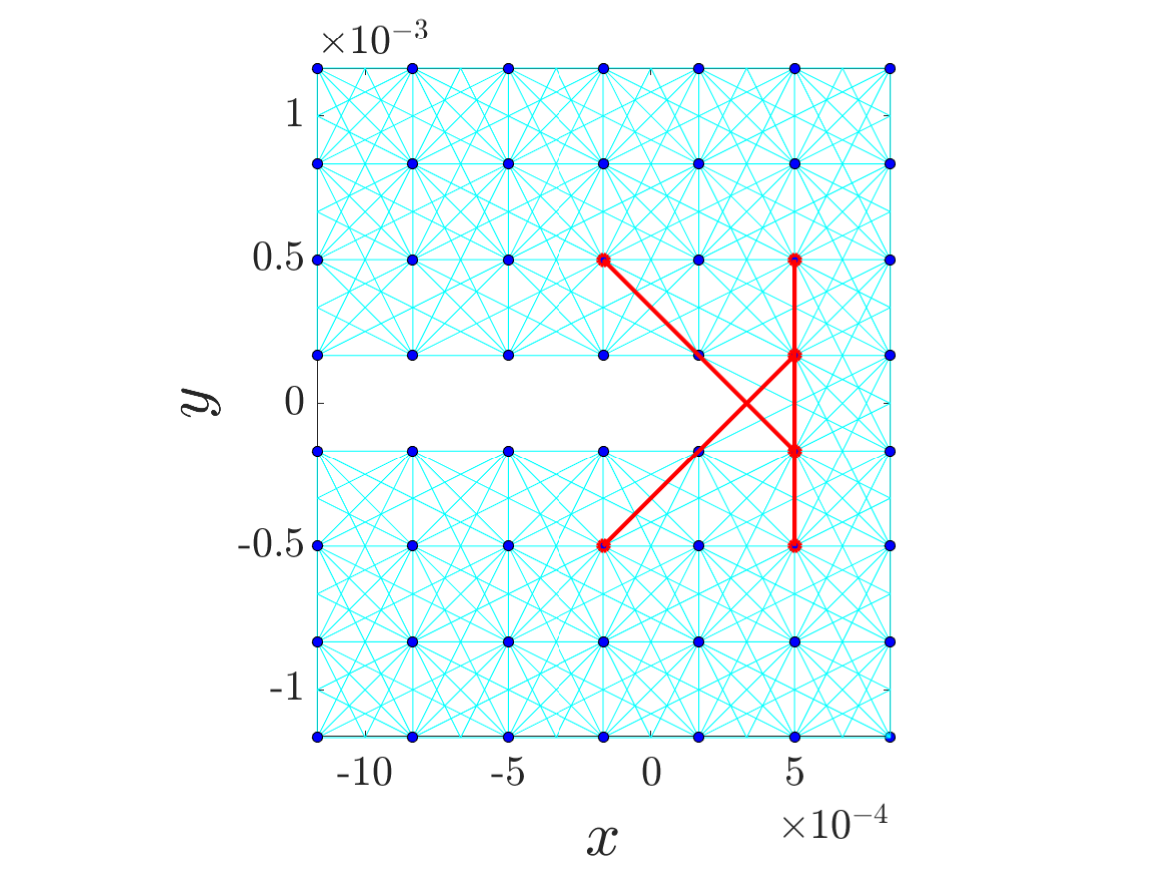}
        \caption{step $= 1032$}
    \label{Fig: Example 2d results (e)}
    \end{subfigure}%
     \vspace{0.5pc}
    \begin{subfigure}[t]{0.49\textwidth}
        \centering
        \includegraphics[width=\width\textwidth, trim=0cm 0cm 0cm 0cm,clip]{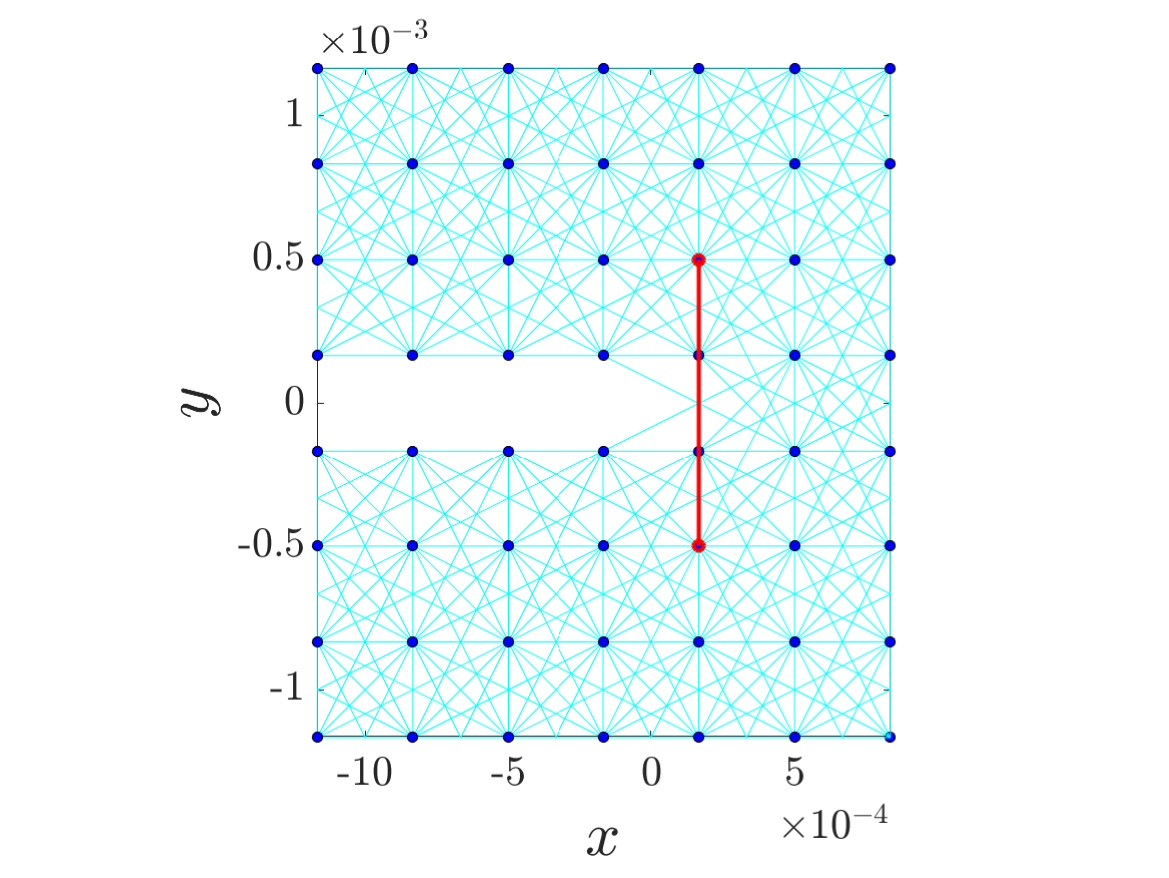}
        \caption{step $= 1014$}
    \label{Fig: Example 2d results (f)}
    \end{subfigure}
                \begin{subfigure}[t]{0.49\textwidth}
        \centering
      \includegraphics[width=\width\textwidth, trim=0cm 0cm 0cm 0cm,clip]{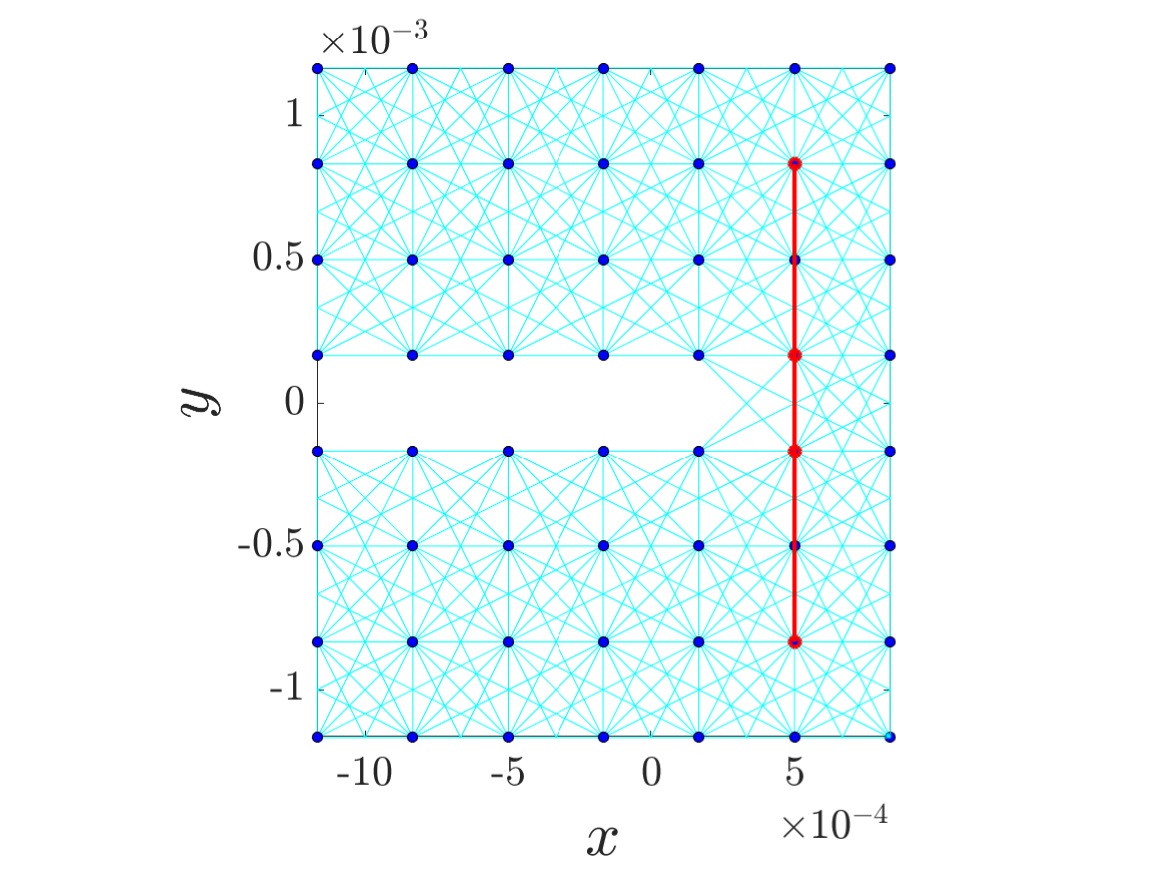}
        \caption{step $= 1033$}
    \label{Fig: Example 2d results (g)}
    \end{subfigure}%
     \vspace{0.5pc}
    \begin{subfigure}[t]{0.49\textwidth}
        \centering
        \includegraphics[width=\width\textwidth, trim=0cm 0cm 0cm 0cm,clip]{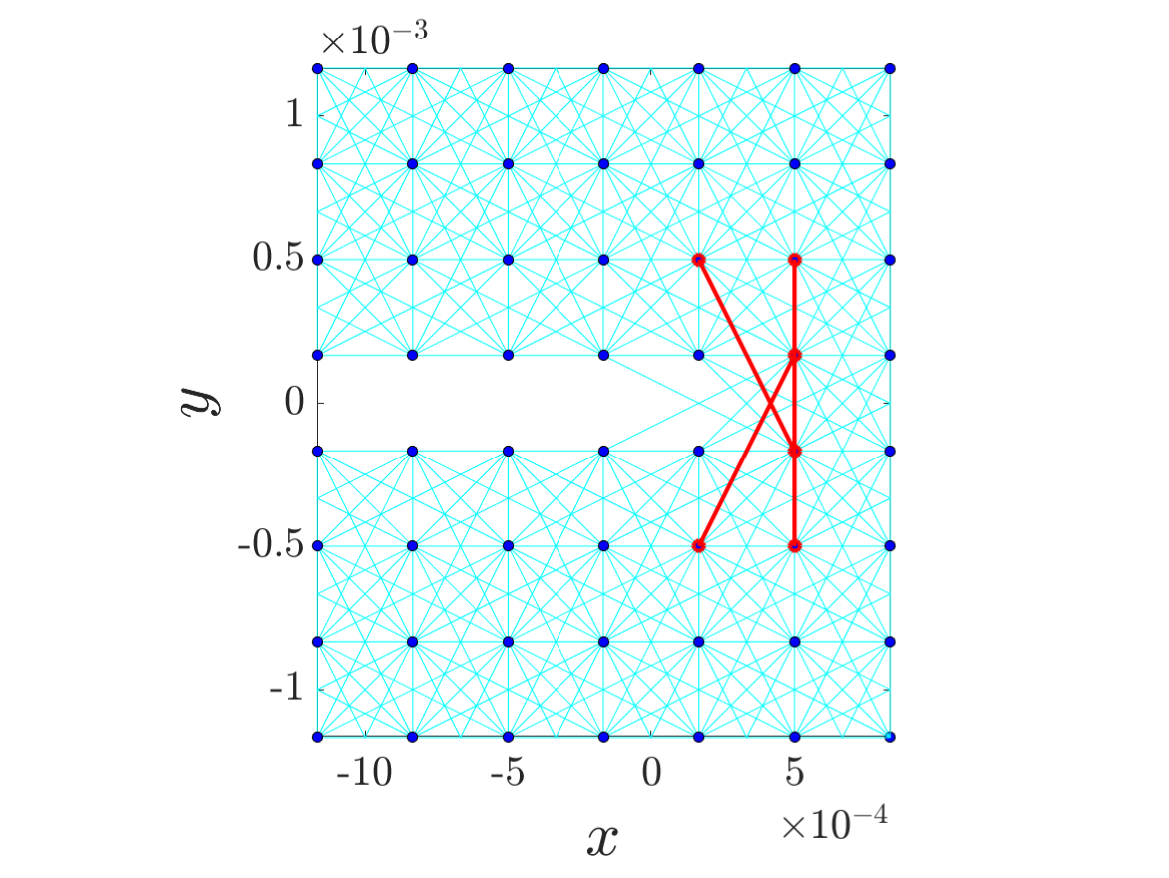}
        \caption{step $= 1015$}
    \label{Fig: Example 2d results (h)}
    \end{subfigure}
                    \begin{subfigure}[t]{0.49\textwidth}
        \centering
        \includegraphics[width=\width\textwidth, trim=0cm 0cm 0cm 0cm,clip]{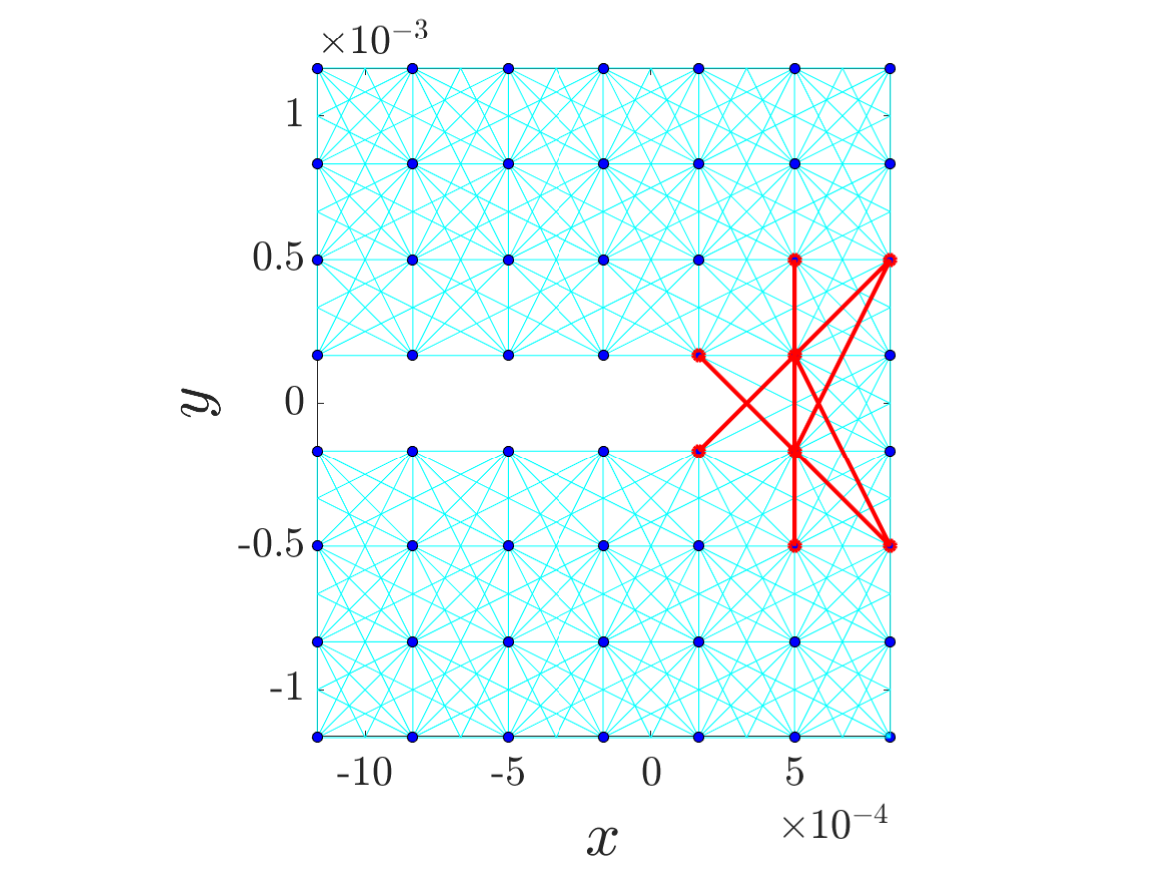}
        \caption{step $= 1034$}
    \label{Fig: Example 2d results (i)}
    \end{subfigure}%
     \vspace{0.5pc}
    \begin{subfigure}[t]{0.49\textwidth}
        \centering
        \includegraphics[width=\width\textwidth, trim=0cm 0cm 0cm 0cm,clip]{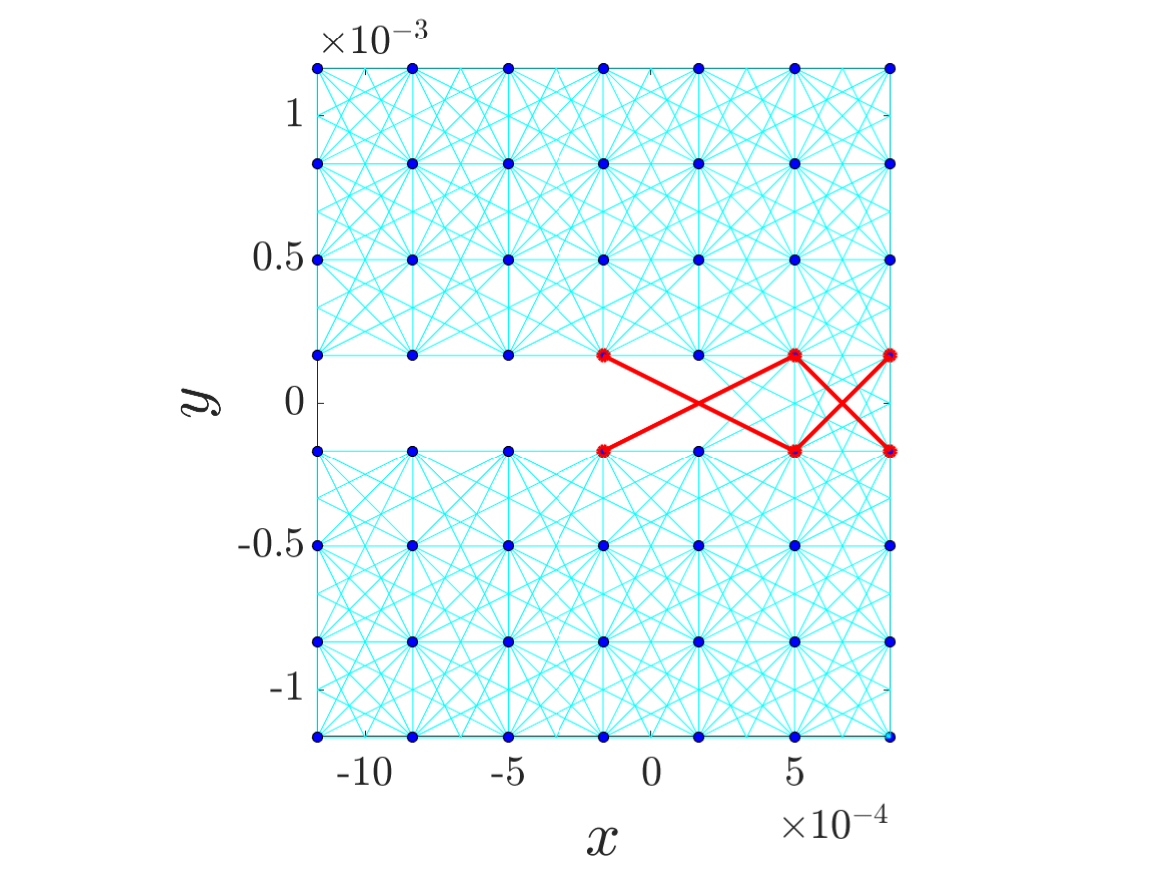}
        \caption{step $= 1018$}
    \label{Fig: Example 2d results (j)}
    \end{subfigure}
        \caption{Comparison of the crack tip evolution between the two bond-failure criteria for $\alpha = 2$ in Example 2 (cont.).} 
    \label{Fig: Example 2d results}
\end{figure*} 

\begin{table}[htbp!]
\centering
{\tabulinesep=1.2mm
\begin{tabu}{||c | c | c c |  c| c c |} 
 \hline
    & Step   & \multicolumn{2}{c|}{\bf Critical stretch criterion} & Step    & \multicolumn{2}{c|}{\bf Critical energy density criterion}  \\ [0.5ex] 
    \cline{3-4} \cline{6-7} 
   &                   &          Bond & Bond length               &                    & Bond & Bond length  \\ [0.5ex] 
 \hline\hline
 $1$ & $1022$ & $1$:~$[17{,}550 \; ; \; 18{,}151]$ & $\sqrt{5}\Delta x$ & $932$ & $1$:~$[17{,}851 \; ; \; 18{,}151]$ & $\Delta x$ \\
        &              & $2$:~$[17{,}851 \; ; \; 18{,}450]$ & $\sqrt{5}\Delta x$ & &  & \\
\hline
$2$ & $1024$ & $1$:~$[17{,}851 \; ; \; 18{,}151]$ & $\Delta x$ & $934$ & $1$:~$[17{,}550 \; ; \; 18{,}151]$ & $\sqrt{5}\Delta x$ \\
       &              &                                                      &                  &              & $2$:~$[17{,}851 \; ; \; 18{,}450]$ & $\sqrt{5}\Delta x$ \\
\hline
$3$ & $1025$ & $1$:~$[17{,}551 \; ; \; 18{,}151]$ & $2\Delta x$ & $935$ & $1$:~$[17{,}551 \; ; \; 18{,}151]$ & $2\Delta x$ \\
        &              & $2$:~$[17{,}851 \; ; \; 18{,}451]$ & $2\Delta x$ &              & $2$:~$[17{,}851 \; ; \; 18{,}451]$ & $2\Delta x$ \\
\hline
$4$ & $1027$ & $1$:~$[17{,}251 \; ; \; 18{,}151]$ & $3\Delta x$ & $1010$ & $1$:~$[17{,}851 \; ; \; 18{,}152]$ & $\sqrt{2}\Delta x$ \\
        &              & $2$:~$[17{,}551 \; ; \; 18{,}451]$ & $3\Delta x$ &              & $2$:~$[17{,}852 \; ; \; 18{,}151]$ & $\sqrt{2}\Delta x$ \\
        &              & $3$:~$[17{,}552 \; ; \; 18{,}151]$ & $\sqrt{5}\Delta x$ & &  & \\
        &              & $4$:~$[17{,}851 \; ; \; 18{,}452]$ & $\sqrt{5}\Delta x$ & &  & \\
        &              & $5$:~$[17{,}851 \; ; \; 18{,}751]$ & $3\Delta x$ & &  & \\
\hline
$5$ & $1028$ & $1$:~$[17{,}851 \; ; \; 18{,}152]$ & $\sqrt{2}\Delta x$ & $1011$ & $1$:~$[17{,}552 \; ; \; 18{,}151]$ & $\sqrt{5}\Delta x$ \\
        &              & $2$:~$[17{,}852 \; ; \; 18{,}151]$ & $\sqrt{2}\Delta x$ &              & $2$:~$[17{,}851 \; ; \; 18{,}452]$ & $\sqrt{5}\Delta x$ \\
\hline
$6$ & $1030$ & $1$:~$[17{,}852 \; ; \; 18{,}152]$ & $\Delta x$ & $1012$ & $1$:~$[17{,}852 \; ; \; 18{,}152]$ & $\Delta x$\\
\hline
$7$ & $1031$ & $1$:~$[17{,}551 \; ; \; 18{,}152]$ & $\sqrt{5}\Delta x$ & $1013$ & $1$:~$[17{,}251 \; ; \; 18{,}151]$ & $3\Delta x$ \\
        &              & $2$:~$[17{,}850 \; ; \; 18{,}152]$ & $\sqrt{5}\Delta x$ &              & $2$:~$[17{,}851 \; ; \; 18{,}751]$ & $3\Delta x$\\
        &              & $3$:~$[17{,}852 \; ; \; 18{,}150]$ & $\sqrt{5}\Delta x$ & &  & \\
        &              & $4$:~$[17{,}852 \; ; \; 18{,}451]$ & $\sqrt{5}\Delta x$ & &  & \\
\hline
$8$ & $1032$ & $1$:~$[17{,}550 \; ; \; 18{,}152]$ & $2\sqrt{2}\Delta x$ & $1014$ & $1$:~$[17{,}551 \; ; \; 18{,}451]$ & $3\Delta x$\\
        &              & $2$:~$[17{,}552 \; ; \; 18{,}152]$ & $2\Delta x$ & &  & \\
        &              & $3$:~$[17{,}852 \; ; \; 18{,}450]$ & $2\sqrt{2}\Delta x$ & &  & \\
        &              & $4$:~$[17{,}852 \; ; \; 18{,}452]$ & $2\Delta x$ & &  & \\
\hline
$9$ & $1033$ & $1$:~$[17{,}252 \; ; \; 18{,}152]$ & $3\Delta x$ & $1015$ & $1$:~$[17{,}551 \; ; \; 18{,}152]$ & $\sqrt{5}\Delta x$\\
        & $1033$ & $2$:~$[17{,}852 \; ; \; 18{,}752]$ & $3\Delta x$ &              & $2$:~$[17{,}552 \; ; \; 18{,}152]$ & $2\Delta x$\\
       &              &                                                      &                  &              & $3$:~$[17{,}852 \; ; \; 18{,}451]$ & $\sqrt{5}\Delta x$ \\
       &              &                                                      &                  &              & $4$:~$[17{,}852 \; ; \; 18{,}452]$ & $2\Delta x$ \\      
\hline
$10$ & $1034$ & $1$:~$[17{,}552 \; ; \; 18{,}452]$ & $3\Delta x$ & $1018$ & $1$:~$[17{,}850 \; ; \; 18{,}152]$ & $\sqrt{5}\Delta x$ \\
         &              & $2$:~$[17{,}553 \; ; \; 18{,}151]$ & $2\sqrt{2}\Delta x$ &              & $2$:~$[17{,}852 \; ; \; 18{,}150]$ & $\sqrt{5}\Delta x$ \\
        &              & $3$:~$[17{,}553 \; ; \; 18{,}152]$ & $\sqrt{5}\Delta x$ &              & $3$:~$[17{,}852 \; ; \; 18{,}153]$ & $\sqrt{2}\Delta x$ \\
        &              & $4$:~$[17{,}851 \; ; \; 18{,}453]$ & $2\sqrt{2}\Delta x$ &              & $4$:~$[17{,}853 \; ; \; 18{,}152]$ & $\sqrt{2}\Delta x$\\
        &              & $5$:~$[17{,}852 \; ; \; 18{,}453]$ & $\sqrt{5}\Delta x$ & &  & \\
 \hline
\end{tabu}
}
\caption{List of broken bonds and corresponding bond lengths for the first $10$ stages of bond breaking, for the case of $\alpha = 2$, for the two bond-failure criteria in Example 2 ({\em cf.}~Figures~\ref{Fig: Example 2c results} and ~\ref{Fig: Example 2d results}). Bonds are listed with the numbers of the pair of computational nodes they connect.}
\label{table: Example 2 broken bonds alpha 2.}
\end{table}

\subsubsection{The case $\mathit{\alpha=0}$}
In the case of $\alpha = 0$, 
the bond energy density required for bond breaking is larger for the critical energy density criterion than for the critical stretch criterion (i.e., $s_c^2(\|\bm{\xi}\|) > s_0^2$) for bonds of length $\sqrt{5}\Delta x$ or shorter ({\em cf.}~Figure~\ref{table: Example 1 critical stretch from energy}). 
This is reflected in Figure~\ref{Fig: Example 2a results}, where although the first bonds to break (bonds of length $\sqrt{5}\Delta x$; see Table~\ref{table: Example 2 broken bonds alpha 0.}) are the same for both criteria, those bonds break 
earlier for the critical stretch criterion (step $928$; see Figure~\ref{Fig: Example 2a results (a)}) than for the critical energy density criterion (step $930$; see Figure~\ref{Fig: Example 2a results (b)}). 
A similar situation occurs during the next stage of bond breaking: although the same bond breaks next (bond of length $\Delta x$; see Table~\ref{table: Example 2 broken bonds alpha 0.}), it breaks earlier for the critical stretch criterion (step $929$; see Figure~\ref{Fig: Example 2a results (c)}) than for the critical energy density criterion (step~$933$; see Figure~\ref{Fig: Example 2a results (d)}). 
Note that the energy density required to break bonds of length $\sqrt{5}\Delta x$ or $\Delta x$ for the critical energy density criterion is larger than that required to break either of these bonds for the critical stretch criterion ({\em cf.}~Figure~\ref{table: Example 1 critical stretch from energy}). This is reflected in that no bonds break until step~930 for the critical energy density criterion, whereas bonds of length  $\sqrt{5}\Delta x$ and $\Delta x$ already break by step~929. 

More precisely, Theorem~\ref{thm:2Dbehavior} implies that, for $\alpha = 0$ and $\delta = 3 \Delta x$, shorter bonds satisfying $ \|\bm{\xi}\| < 9\Delta x/4$ break earlier for the critical stretch criterion than for the critical energy density criterion, whereas the opposite holds for  longer bonds satisfying $9\Delta x/4 < \|\bm{\xi}\| \leq \delta$. 
%
%
This is consistent with the results observed in Figures~\ref{Fig: Example 2a results} and \ref{Fig: Example 2b results}. For instance, in Figure~\ref{Fig: Example 2b results}, we observe that, for the critical stretch criterion, shorter bonds (of length $\sqrt{2}\Delta x$; see Table~\ref{table: Example 2 broken bonds alpha 0.}) break at step $935$ (see Figure~\ref{Fig: Example 2b results (a)}), followed by longer bonds (of length $2\sqrt{2}\Delta x$; see Table~\ref{table: Example 2 broken bonds alpha 0.}) that break at step $1036$ (see Figure~\ref{Fig: Example 2b results (c)}); in contrast, the reversed scenario occurs for the critical energy density criterion (see Figures~\ref{Fig: Example 2b results (b)} and~\ref{Fig: Example 2b results (d)}).

A comparison of the strain energy density field at the time of initial bond breaking is provided in Figure~\ref{Fig: Example 2a2 results}. Since  there is only a two time-step difference between the result for the critical stretch criterion (step $928$) and the result for the critical energy density criterion (step $930$), the two figures look alike.  

\begin{figure*}[htbp!]
    \centering
   {\bf The case}~\boldsymbol{$\alpha = 0$}\\[0.1in]
    \begin{subfigure}[t]{0.49\textwidth}
        \centering
      {\bf Critical stretch}\\[0.1in]
       \includegraphics[width=\textwidth, trim=0cm 0cm 0cm 0cm,clip]{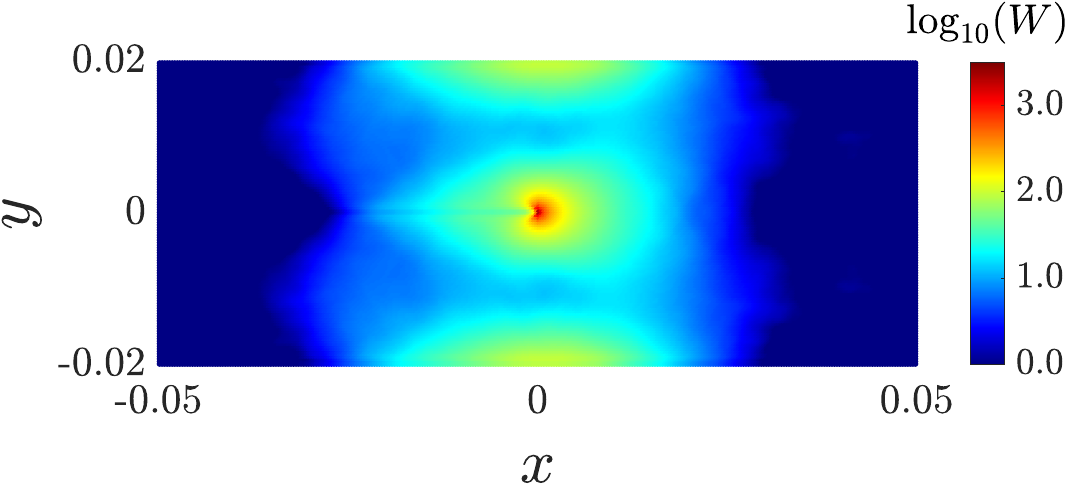}
        \caption{step $= 928$}
    \label{Fig: Example 2a2 results (a)}
    \end{subfigure}%
      ~ 
    \begin{subfigure}[t]{0.49\textwidth}
        \centering
    {\bf Critical energy density}\\[0.1in]
        \includegraphics[width=\textwidth, trim=0cm 0cm 0cm 0cm,clip]{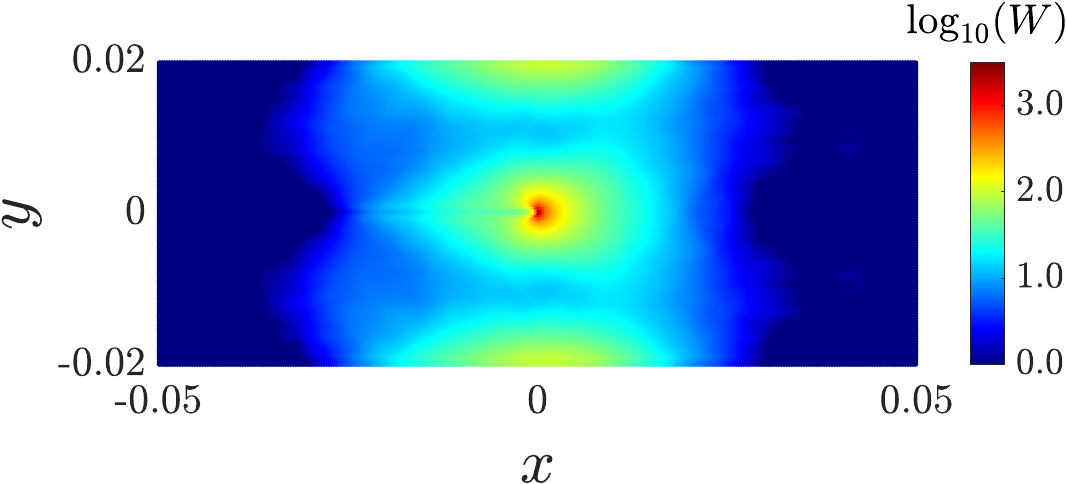}
        \caption{step $= 930$}
    \label{Fig: Example 2a2 results (b)}
    \end{subfigure}
        \caption{Comparison of the strain energy density field between the two bond-failure criteria at the time of initial bond breaking for $\alpha = 0$ in Example~2.} 
    \label{Fig: Example 2a2 results}
\end{figure*}

\subsubsection{The case $\mathit{\alpha=2}$}
In the case of $\alpha = 2$, the bond energy density required for bond breaking is smaller for the critical energy density criterion than for the critical stretch criterion (i.e., $s_c^2(\|\bm{\xi}\|) < s_0^2$)  for bonds of length $\Delta x$ or $\sqrt{2}\Delta x$ ({\em cf.}~Figure~\ref{table: Example 1 critical stretch from energy}). 
More precisely, Theorem~\ref{thm:2Dbehavior} implies that, for $\alpha = 2$ and $\delta = 3 \Delta x$, shorter bonds satisfying $\|\bm{\xi}\| < 2\Delta x$ break earlier for the critical energy density criterion than for the critical stretch criterion, whereas the opposite holds for longer bonds satisfying $2\Delta x < \|\bm{\xi}\| \leq \delta$.
This is reflected 
 in Figure~\ref{Fig: Example 2c results}, where the first bond to break for the critical energy density criterion (step $932$; see Figure~\ref{Fig: Example 2c results (b)}) is a shortest bond of length $\Delta x$ (see Table~\ref{table: Example 2 broken bonds alpha 2.}). This bond also breaks much earlier than the same one for the critical stretch criterion (step $1024$; see Figure~\ref{Fig: Example 2c results (c)}), which only breaks after longer bonds (of length $\sqrt{5}\Delta x$; see Table~\ref{table: Example 2 broken bonds alpha 2.}) break.  
These longer bonds require more energy density to break for the critical energy density criterion, compared to the bond of length $\Delta x$, so they break second. 
We note that for both criteria, the same bonds break third, which are of length $2\Delta x$ (see Table~\ref{table: Example 2 broken bonds alpha 2.} and Figures~\ref{Fig: Example 2c results (e)} and~\ref{Fig: Example 2c results (f)}). The fact that these bonds break in the same stage of bond breaking for both criteria is consistent with the fact that the same energy density is require for them to break in both criteria ({\em cf.}~Figure~\ref{table: Example 1 critical stretch from energy}), even though they break at different time steps. 
Another example of shorter bonds breaking first for the critical energy density criterion is observed in the fourth stage of bond breaking (step~$1010$; see Figure~\ref{Fig: Example 2c results (h)}), where bonds of length $\sqrt{2}\Delta x$ (see Table~\ref{table: Example 2 broken bonds alpha 2.}) break at an early stage compared to the critical stretch criterion (step $1028$; see Figure~\ref{Fig: Example 2c results (i)}), which only break after longer bonds (of length $\sqrt{5}\Delta x$ and $3\Delta x$; see Table~\ref{table: Example 2 broken bonds alpha 2.}) break.

A comparison of the strain energy density field at the time of initial bond breaking is provided in Figure~\ref{Fig: Example 2c2 results}. We observe that a larger amount of energy is required to initiate bond breaking under the critical stretch criterion (step $1022$) compared to the critical energy density criterion (step $932$).

\begin{figure*}[htbp!]
    \centering
     {\bf The case}~\boldsymbol{$\alpha = 2$}\\[0.1in]
    \begin{subfigure}[t]{0.49\textwidth}
        \centering
      {\bf Critical stretch}\\[0.1in]
       \includegraphics[width=\textwidth, trim=0cm 0cm 0cm 0cm,clip]{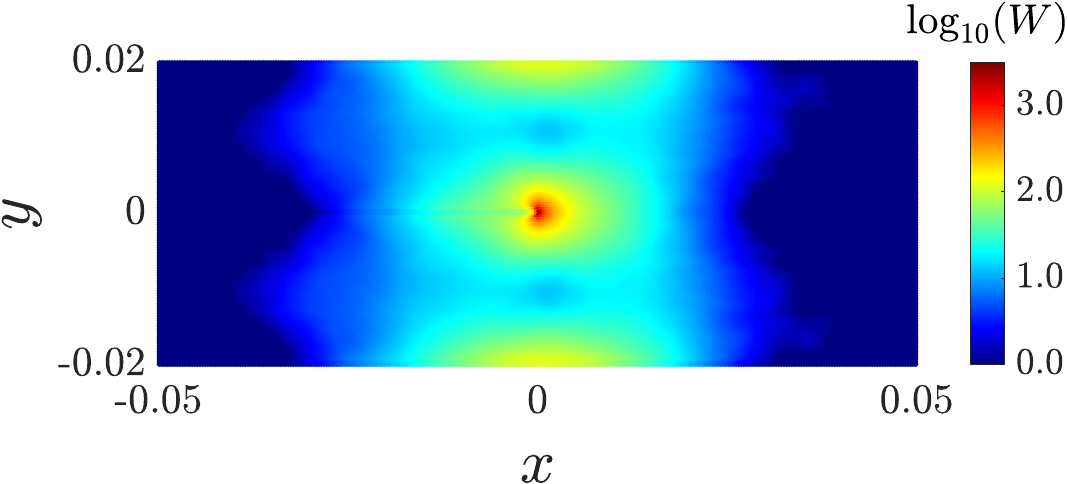}
        \caption{step $= 1022$}
    \label{Fig: Example 2c2 results (a)}
    \end{subfigure}%
      ~ 
    \begin{subfigure}[t]{0.49\textwidth}
        \centering
    {\bf Critical energy density}\\[0.1in]
        \includegraphics[width=\textwidth, trim=0cm 0cm 0cm 0cm,clip]{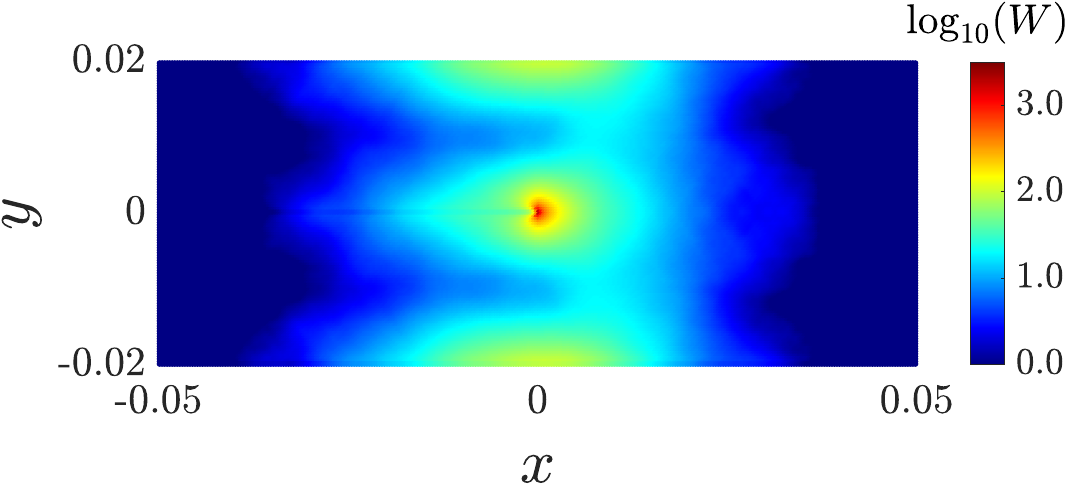}
        \caption{step $= 932$}
    \label{Fig: Example 2c2 results (b)}
    \end{subfigure}
        \caption{Comparison of the strain energy density field between the two bond-failure criteria at the time of initial bond breaking for $\alpha = 2$ in Example~2.} 
    \label{Fig: Example 2c2 results}
\end{figure*} 

\subsection{Example 3: Crack propagation and branching}\label{sec:example3}

As in Section~\ref{sec: Crack tip evolution}, we consider the problem of crack propagation in a pre-notched soda-lime glass thin plate subjected to traction loading~\cite{Bobaru-Zhang-2015}. In Section~\ref{sec: Crack tip evolution}, we focused on the crack tip evolution by examining the initial stages of bond breaking that lead to crack propagation. Here, in contrast, we perform a full simulation and study the actual propagation of the crack.

The simulation employs a two-dimensional pre-notched rectangular plate under a plane stress assumption, with dimensions $0.1$~m by $0.04$~m, as illustrated in Figure~\ref{fig: Crack propagation example}.  
 The material parameters are:  
mass density $\rho = 2{,}440$ kg/m$^3$, Young's modulus $E=72$~GPa, and fracture energy $G_0 = 3.8$~J/m$^2$; note that the thickness $h$ is not used in the simulations because it does not appear in the equations in practice (see Remarks~\ref{rem: h independence of model} and~\ref{rem: h independence of critical stretch}). The plate is initially at rest with zero initial displacements and velocities. We first consider in Section~\ref{sec: Example 3 lower traction} the case of a lower traction, with a magnitude of $\sigma = 0.2$~MPa (as in Section~\ref{sec: Crack tip evolution}), applied on the top and bottom of the plate. Then, in Section~\ref{sec: Example 3 higher traction}, we consider the case of a higher traction, with a  magnitude of $\sigma = 2$~MPa.

\begin{figure}[htbp!] 
\vspace*{1.5in}
\begin{center}
\setlength{\unitlength}{0.5cm}
\begin{picture}(0,1.2)
   \put(-5,7){\line(1,0){10}}    	
   \put(-5,3){\line(0,1){4}}     	
   \put(-5,3){\line(1,0){10}}   	
   \put(5,3){\line(0,1){4}}       
  {\color{red} \put(-5,5){\line(1,0){5}}}
  \put(-3.5,5.4){\vector(-1,0){1.4}} 
  \put(-1.5,5.4){\vector(1,0){1.4}}
  \put(-3.4,5.3){\tiny $0.05\,$m}
  \put(-3.6,4.4){\tiny pre-notch}
  \put(-6.2,6.4){\vector(0,1){0.6}} 
  \put(-6.2,5.6){\vector(0,-1){0.6}}
  \put(-7.0,5.8){\tiny $0.02\,$m}
   {\color{blue}\multiput(-5,7)(1,0){11}{\vector(0,1){1}}
   {\multiput(-5,3)(1,0){11}{\vector(0,-1){1}}}
   \put(-0.15,8.3){$\sigma$}
    \put(-0.15,1.4){$\sigma$} 
   \put(-5,8){\line(1,0){10}} 
   \put(-5,2){\line(1,0){10}}  
    }
  \put(-0.75,0.8){\vector(-1,0){4.25}} 
  \put(0.9,0.8){\vector(1,0){4.1}}
  \put(-0.7,0.7){\tiny $0.1\,$m}
  \put(5.7,5.4){\vector(0,1){1.6}} 
  \put(5.7,4.6){\vector(0,-1){1.6}}
  \put(5.2,4.8){\tiny $0.04\,$m}
  %
\end{picture}
\end{center}
\caption{Illustration of the pre-notched thin plate under traction loading in Example 3.}
\label{fig: Crack propagation example}
\end{figure}
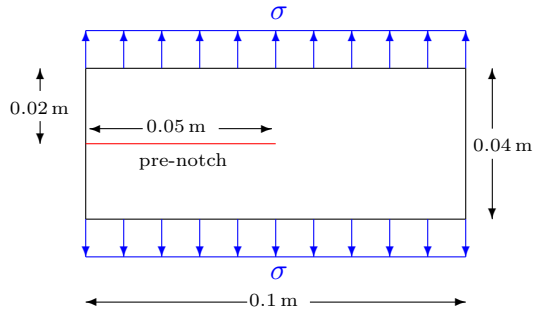

The discretization employs the meshfree approach from~\cite{SillingAskari} with a uniform grid of $300 \times 120$ computational nodes (total of $36{,}000$  nodes) and a horizon of $\delta = 1\times 10^{-3}$~m, which results in an $m$-ratio of $m=3$ ({\em cf.}~Figure~\ref{fig:Example1 neighborhood}). The total simulation time is $T = 150$~$\mu$s, and the time step is $\Delta t = 67\times 10^{-3}$~$\mu$s. Time integration is performed via the velocity Verlet algorithm.
In addition to the strain energy density field, the damage field introduced in~\cite{SillingAskari} is computed. 
The computations employ the PDMATLAB2D code; see~\cite{SelesonEtAl} for further details.

\subsubsection{Lower traction: $\mathit{\sigma = 0{.}2}$ MPa}\label{sec: Example 3 lower traction}

In Figure~\ref{Fig: Example 3a results}, we compare the  damage and strain energy density fields obtained using the critical stretch criterion (top plots) with those obtained using the critical energy density criterion (bottom plots) at the end of the simulation for the case of $\alpha = 0$. 
Although minor differences can be observed, especially in the strain energy density, overall the results are similar. To further assess the difference between the two criteria, we plot the  evolution of the horizontal position of the crack tip over time. Note that, 
in Figure~\ref{Fig: Example 3a results}, the crack  propagates horizontally for both criteria, so the vertical component of the crack tip position remains constant. To track the crack tip position over time, we select the right-most node with damage $\varphi >0.35$ as the crack tip, see~\cite{Bobaru-Ha-2010}. Figure~\ref{Fig: Example 3a2 results} shows the results for the  horizontal position of the crack tip over time for both criteria. The results suggest that crack propagation is initiated at the same time ($t \approx 74$ $\mu$s) for both criteria. 
However, after the initial propagation, the crack tip of the critical stretch criterion simulation lags behind the crack tip of the critical energy density criterion simulation.

\begin{figure*}[htbp!]
    \centering
      {\bf The case}~\boldsymbol{$\sigma = 0.2$} {\bf MPa and} \boldsymbol{$\alpha = 0$}\\[0.1in]
     {\bf Critical stretch}\\[0.1in]
    \begin{subfigure}[t]{0.49\textwidth}
        \centering
       \includegraphics[width=.95\textwidth, trim=0cm 0cm 0cm 0cm,clip]{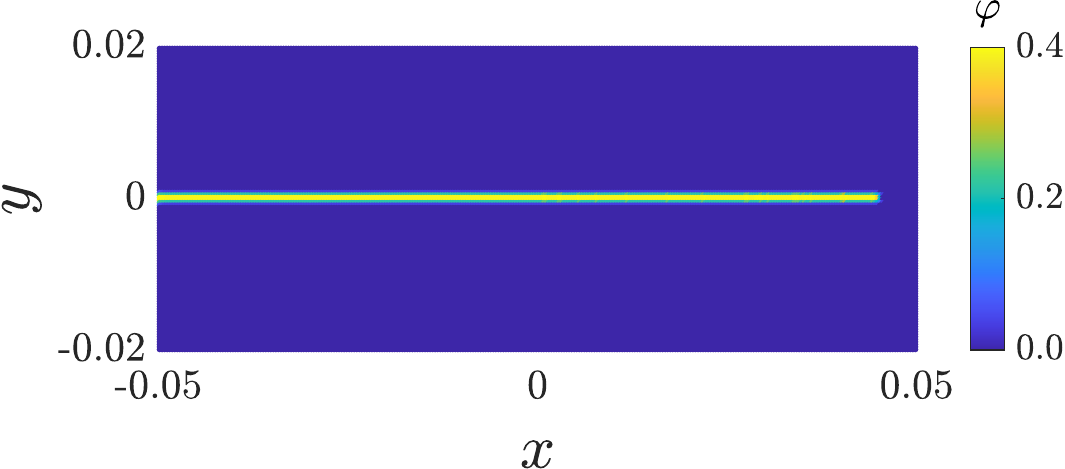}
        \caption{Damage}
    \label{Fig: Example 3a results (a)}
    \end{subfigure}%
      ~ 
    \begin{subfigure}[t]{0.49\textwidth}
        \centering
        \includegraphics[width=.95\textwidth, trim=0cm 0cm 0cm 0cm,clip]{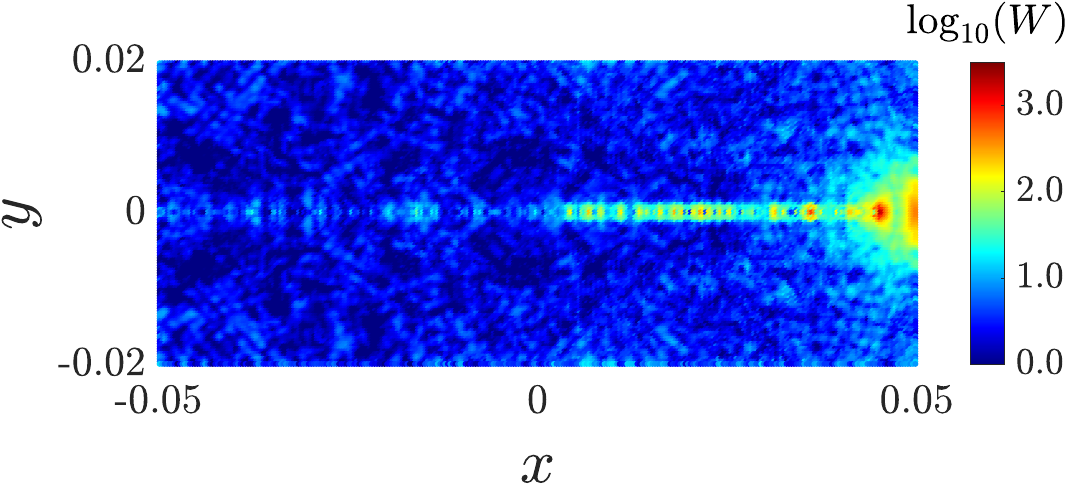}
        \caption{Strain energy density}
    \label{Fig: Example 3a results (b)}
    \end{subfigure}
        {\bf Critical energy density}\\[0.1in]
    \begin{subfigure}[t]{0.49\textwidth}
        \centering
       \includegraphics[width=.95\textwidth, trim=0cm 0cm 0cm 0cm,clip]{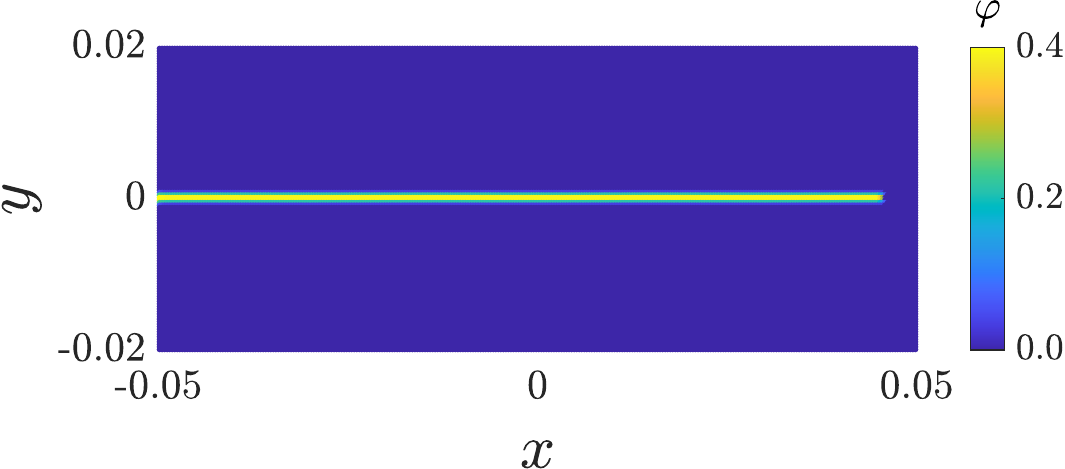}
        \caption{Damage}
    \label{Fig: Example 3a results (a)}
    \end{subfigure}%
      ~ 
    \begin{subfigure}[t]{0.49\textwidth}
        \centering
        \includegraphics[width=.95\textwidth, trim=0cm 0cm 0cm 0cm,clip]{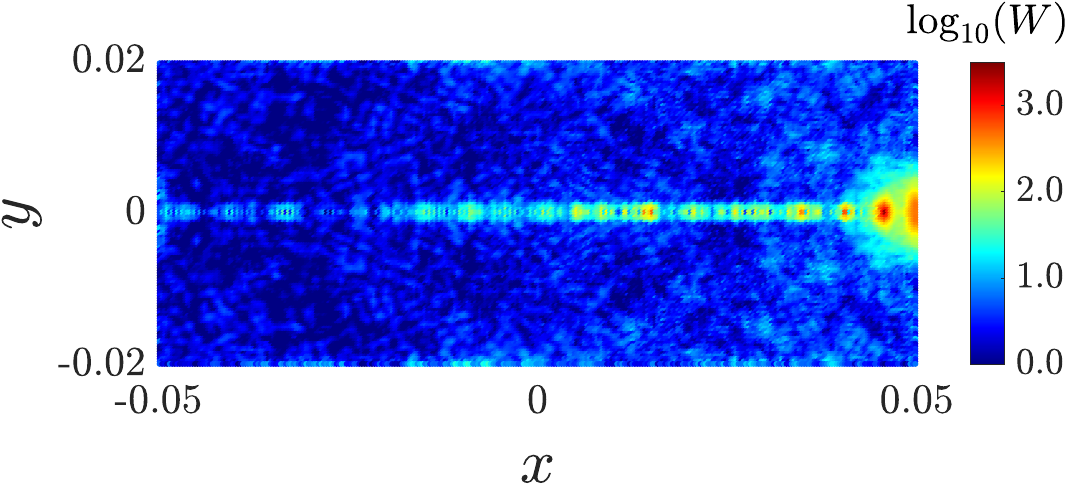}
        \caption{Strain energy density}
    \label{Fig: Example 3a results (b)}
    \end{subfigure}
        \caption{Comparison of the damage and strain energy density fields between the two bond-failure criteria for the case of $\sigma = 0.2$~MPa and $\alpha = 0$ at the final time in Example~3.} 
    \label{Fig: Example 3a results}
\end{figure*}

\begin{figure*}[htbp!]
     \begin{subfigure}[t]{0.49\textwidth}
        \centering
                \includegraphics[width=\textwidth, trim=0cm 0cm 0cm 0cm,clip]{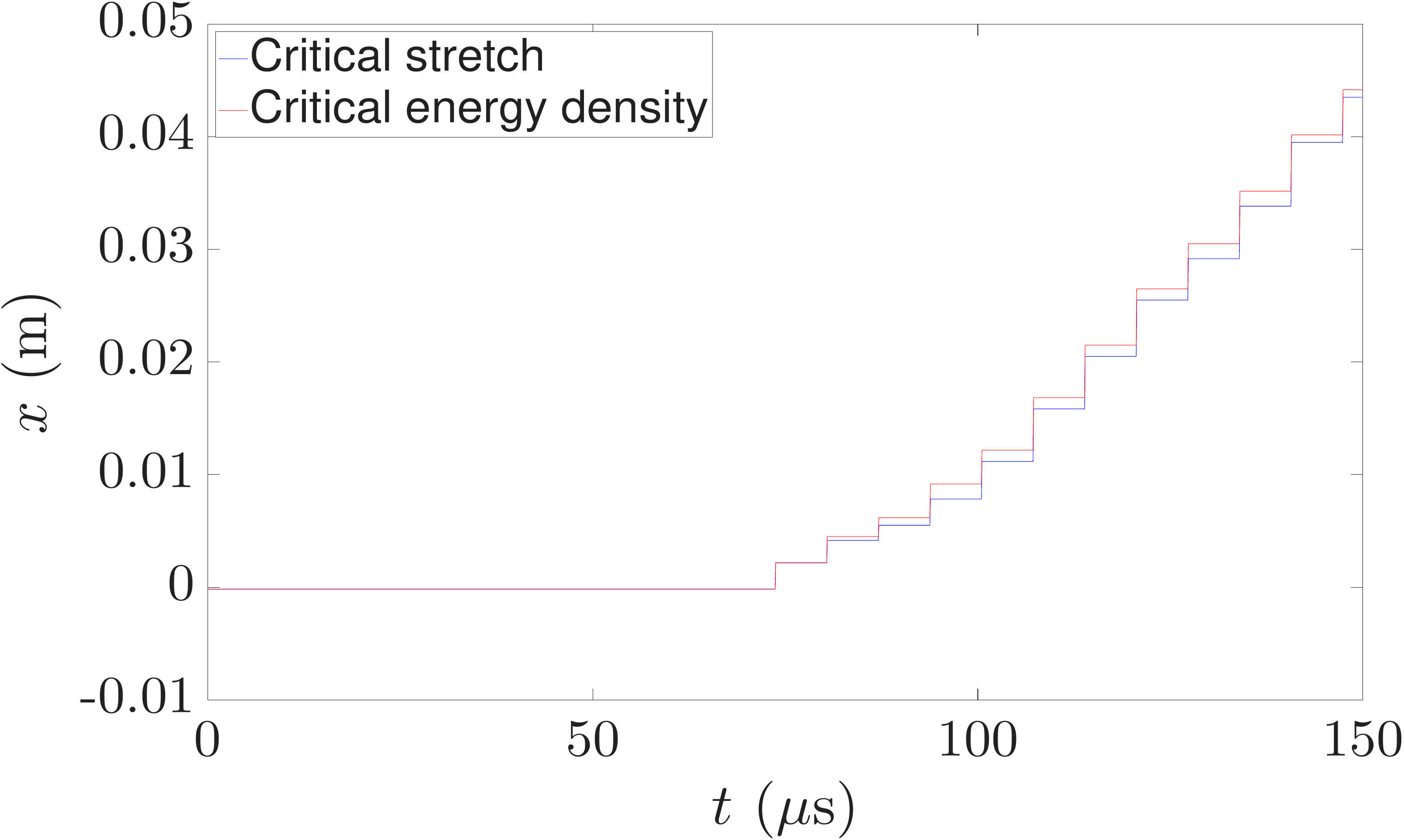}
        \caption{$\alpha = 0$} 
    \label{Fig: Example 3a2 results}
    \end{subfigure}
    \begin{subfigure}[t]{0.49\textwidth}
    \centering
          \includegraphics[width=\textwidth, trim=0cm 0cm 0cm 0cm,clip]{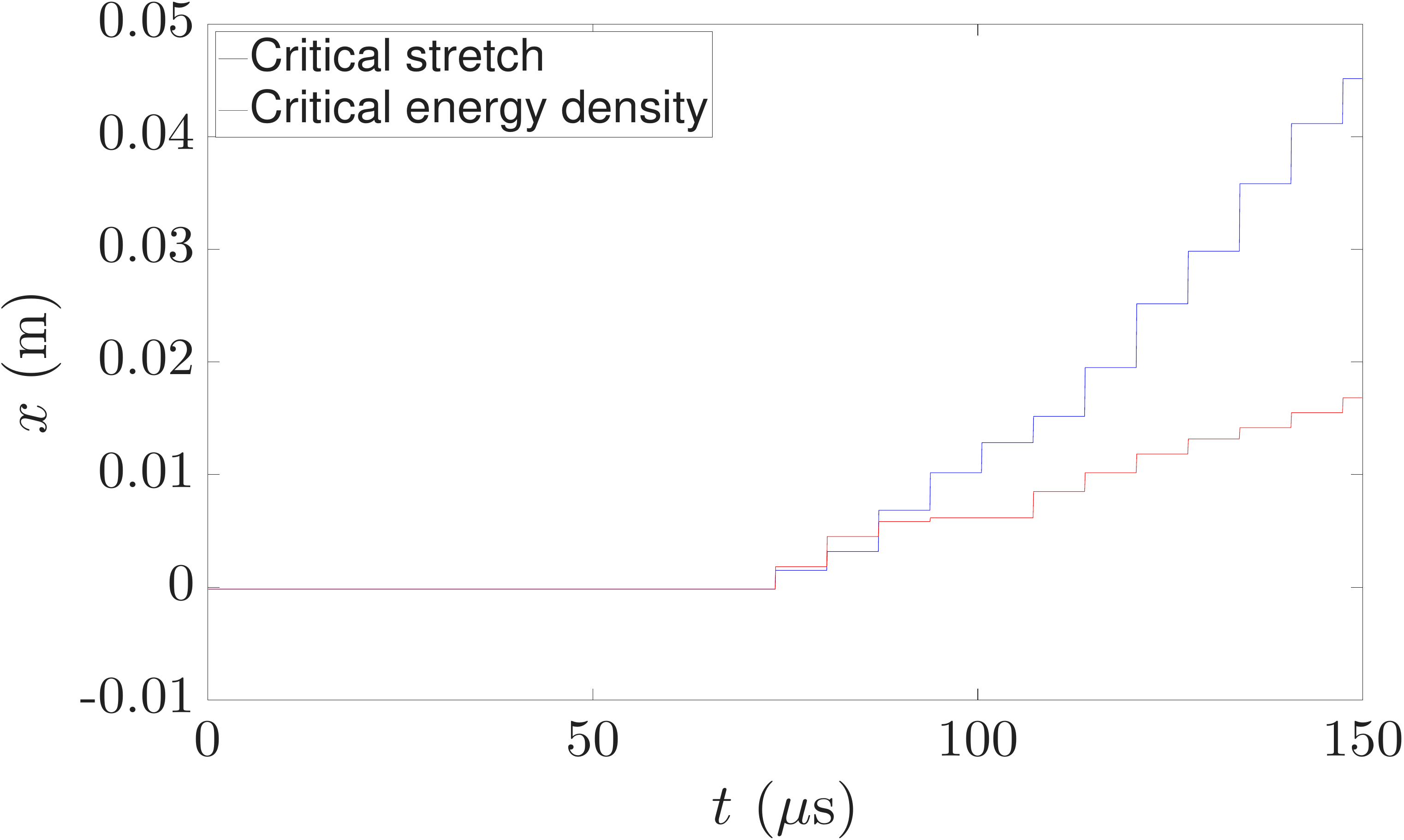}
        \caption{$\alpha = 2$} 
    \label{Fig: Example 3b2 results}
\end{subfigure} 
    \caption{Comparison of the horizontal position of the crack tip  over time between the two bond-failure criteria in Example~3. }
\end{figure*}

\begin{figure*}[htbp!]
    \centering
       {\bf The case}~\boldsymbol{$\sigma = 0.2$}  {\bf MPa and} \boldsymbol{$\alpha = 2$}\\[0.1in]
     {\bf Critical stretch}\\[0.1in]
    \begin{subfigure}[t]{0.49\textwidth}
        \centering
       \includegraphics[width=.95\textwidth, trim=0cm 0cm 0cm 0cm,clip]{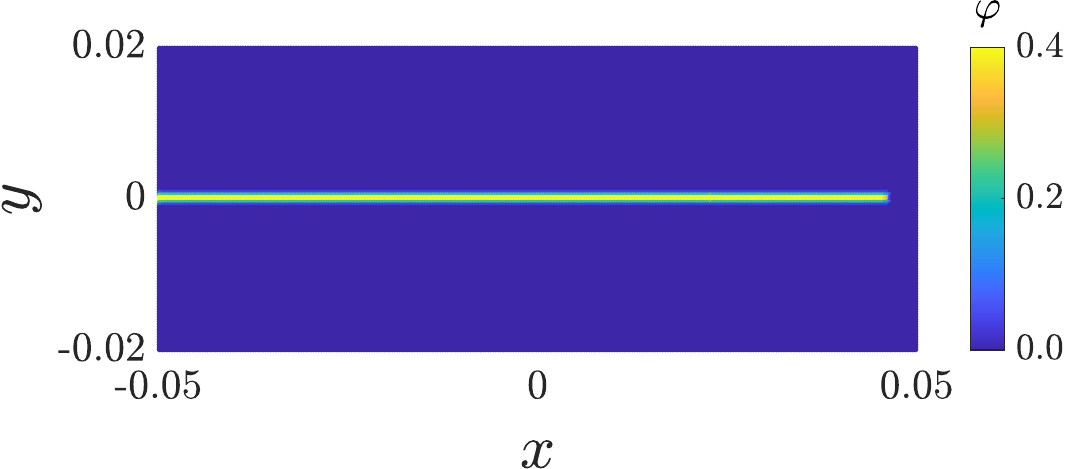}
        \caption{Damage}
    \label{Fig: Example 3b2 results (a)}
    \end{subfigure}%
      ~ 
    \begin{subfigure}[t]{0.49\textwidth}
        \centering
        \includegraphics[width=.95\textwidth, trim=0cm 0cm 0cm 0cm,clip]{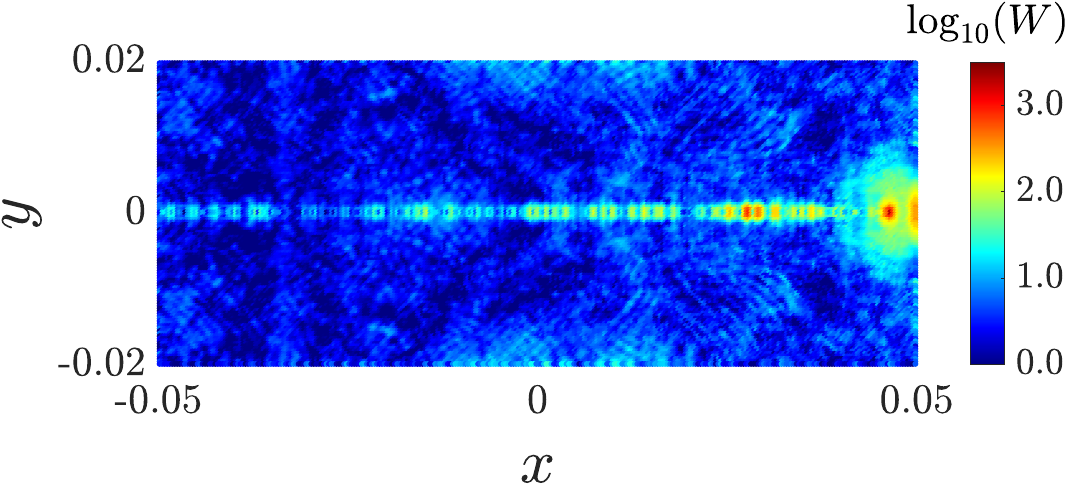}
        \caption{Strain energy density}
    \label{Fig: Example 3b2 results (b)}
    \end{subfigure}
        {\bf Critical energy density}\\[0.1in]
    \begin{subfigure}[t]{0.49\textwidth}
        \centering
       \includegraphics[width=.95\textwidth, trim=0cm 0cm 0cm 0cm,clip]{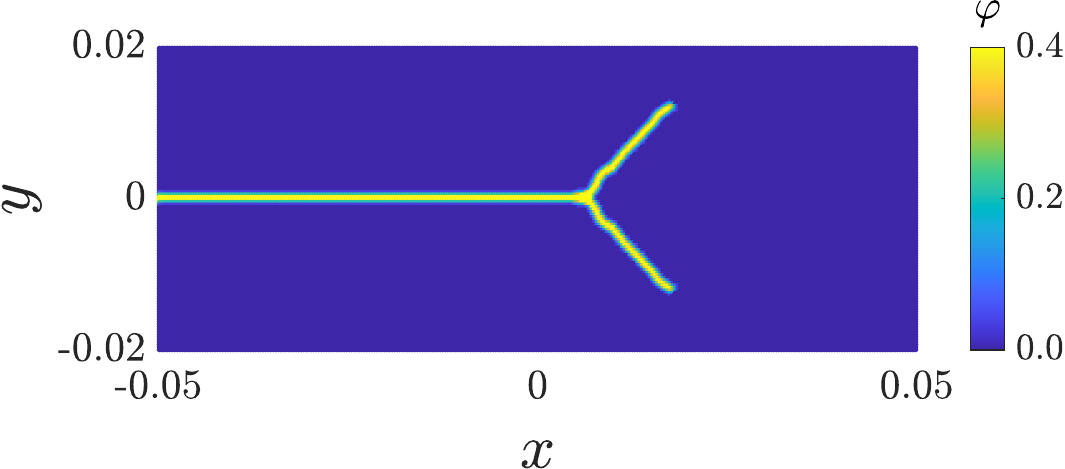}
        \caption{Damage}
    \label{Fig: Example 3b2 results (a)}
    \end{subfigure}%
      ~ 
    \begin{subfigure}[t]{0.49\textwidth}
        \centering
        \includegraphics[width=.95\textwidth, trim=0cm 0cm 0cm 0cm,clip]{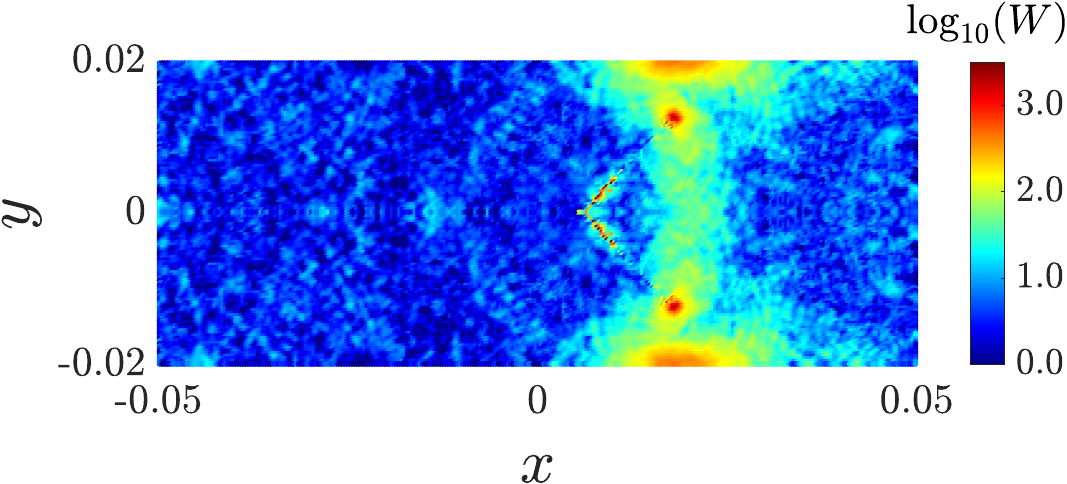}
        \caption{Strain energy density}
    \label{Fig: Example 3b2 results (b)}
    \end{subfigure}
        \caption{Comparison of the damage and strain energy density fields between the two bond-failure criteria for the case of $\sigma = 0.2$~MPa and $\alpha = 2$ at the final time in Example 3.} 
    \label{Fig: Example 3b results}
\end{figure*} 

In Figure~\ref{Fig: Example 3b results}, we compare the  damage and strain energy density fields obtained using the critical stretch criterion (top plots) with those obtained using the critical energy density criterion (bottom plots) at the end of the simulation for the case of $\alpha = 2$. In this case, the difference between the two criteria is significant: for the critical stretch criterion, the crack evolves as a single horizontal crack, whereas for the critical energy density criterion, the crack branches. 
We plot the  horizontal position of the crack tip over time for both criteria in Figure~\ref{Fig: Example 3b2 results}. Note that crack propagation is initiated at the same time for both criteria, which also
coincides with the case of $\alpha = 0$ ($t \approx 74$ $\mu$s).

\subsubsection{Higher traction: $\mathit{\sigma = 2}$ MPa}\label{sec: Example 3 higher traction}


We now consider the case of traction with a tenfold magnitude, i.e., $\sigma = 2$~MPa. This case has been shown to lead to crack branching for the critical stretch criterion ~\cite{Bobaru-Zhang-2015}.  
We run the same simulation as above, but with the new traction magnitude and a shorter simulation time of $T = 43$~$\mu$s. 
The results are shown in Figure~\ref{Fig: Example 4a results} for $\alpha = 0$ and in Figure~\ref{Fig: Example 4b results} for $\alpha = 2$. 
For $\alpha=0$, small differences between the two bond-failure criteria are observed, especially in the strain energy density. 
For $\alpha = 2$, a more significant difference is observed in the resulting crack paths.  

\begin{figure*}[htbp!]
    \centering
      {\bf The case}~\boldsymbol{$\sigma = 2$} {\bf MPa and} \boldsymbol{$\alpha = 0$}\\[0.1in]
     {\bf Critical stretch}\\[0.1in]
    \begin{subfigure}[t]{0.49\textwidth}
        \centering
       \includegraphics[width=.95\textwidth, trim=0cm 0cm 0cm 0cm,clip]{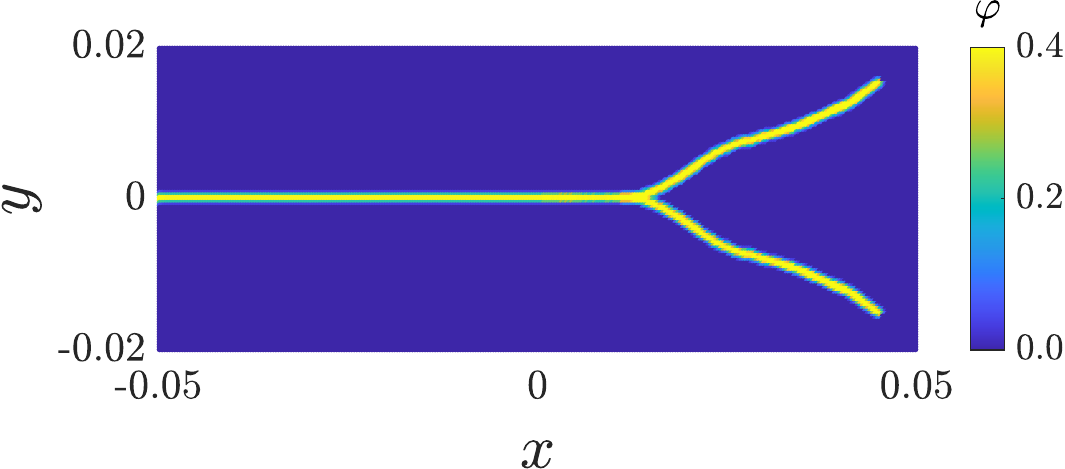}
        \caption{Damage}
    \label{Fig: Example 4a results (a)}
    \end{subfigure}%
      ~ 
    \begin{subfigure}[t]{0.49\textwidth}
        \centering
        \includegraphics[width=.95\textwidth, trim=0cm 0cm 0cm 0cm,clip]{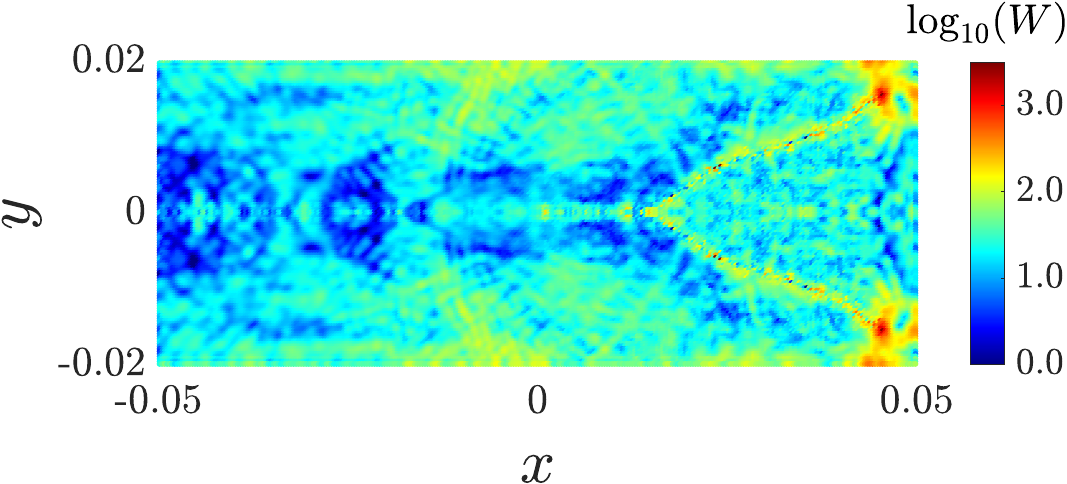}
        \caption{Strain energy density}
    \label{Fig: Example 4a results (b)}
    \end{subfigure}
        {\bf Critical energy density}\\[0.1in]
    \begin{subfigure}[t]{0.49\textwidth}
        \centering
       \includegraphics[width=.95\textwidth, trim=0cm 0cm 0cm 0cm,clip]{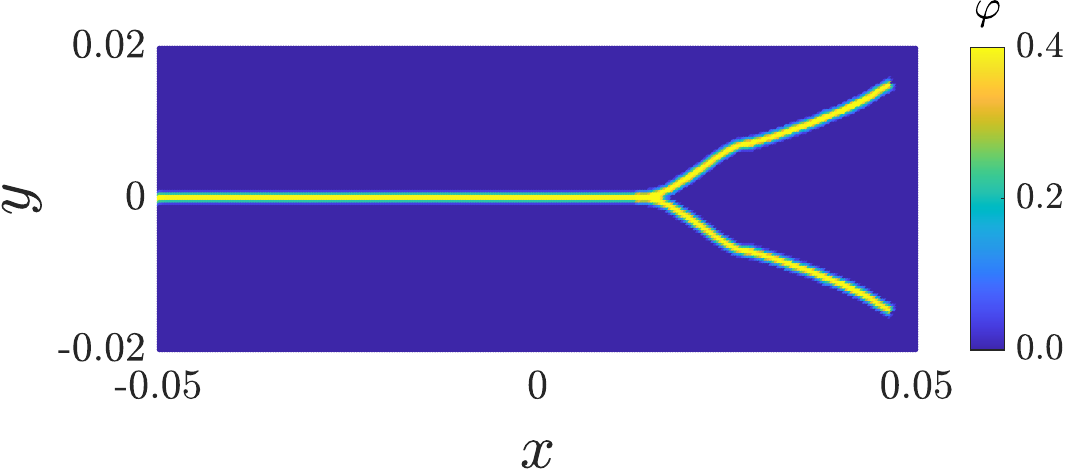}
        \caption{Damage}
    \label{Fig: Example 4a results (a)}
    \end{subfigure}%
      ~ 
    \begin{subfigure}[t]{0.49\textwidth}
        \centering
        \includegraphics[width=.95\textwidth, trim=0cm 0cm 0cm 0cm,clip]{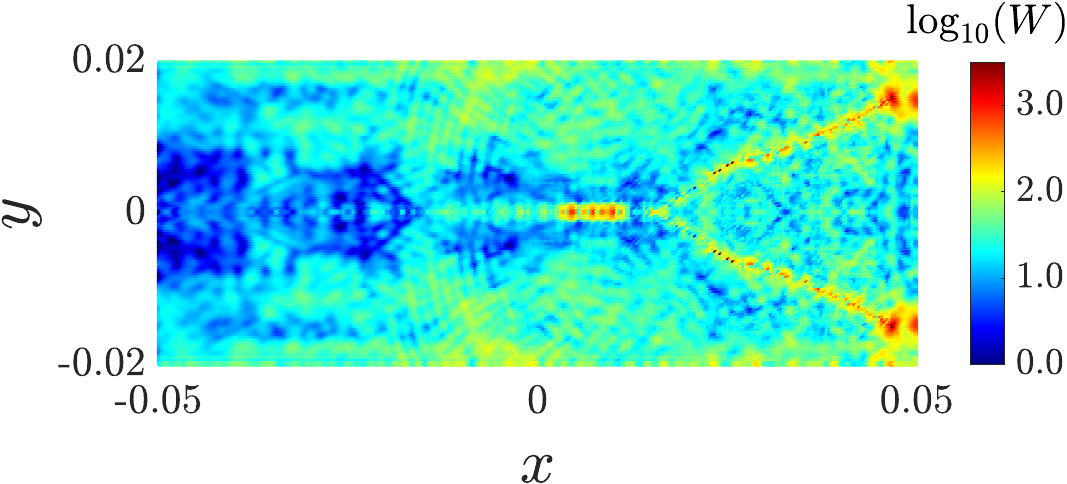}
        \caption{Strain energy density}
    \label{Fig: Example 4a results (b)}
    \end{subfigure}
        \caption{Comparison of the damage  and strain energy density fields between 
       the two bond-failure criteria for the case of $\sigma = 2$~MPa and $\alpha = 0$ at the final time in Example~3.} 
    \label{Fig: Example 4a results}
\end{figure*} 

\begin{figure*}[htbp!]
    \centering
         {\bf The case}~\boldsymbol{$\sigma = 2$} {\bf MPa and} \boldsymbol{$\alpha = 2$}\\[0.1in]
     {\bf Critical stretch}\\[0.1in]
    \begin{subfigure}[t]{0.49\textwidth}
        \centering
       \includegraphics[width=.95\textwidth, trim=0cm 0cm 0cm 0cm,clip]{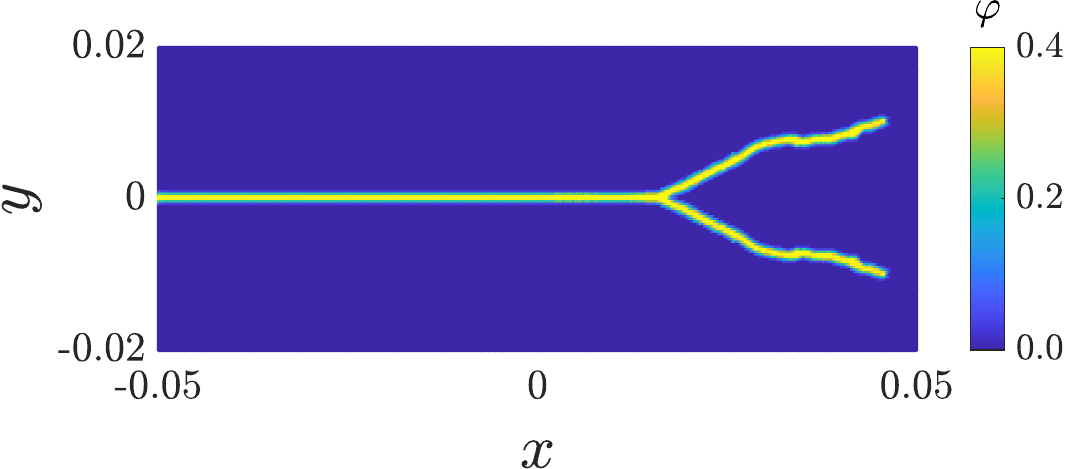}
        \caption{Damage}
    \label{Fig: Example 4b results (a)}
    \end{subfigure}%
      ~ 
    \begin{subfigure}[t]{0.49\textwidth}
        \centering
        \includegraphics[width=.95\textwidth, trim=0cm 0cm 0cm 0cm,clip]{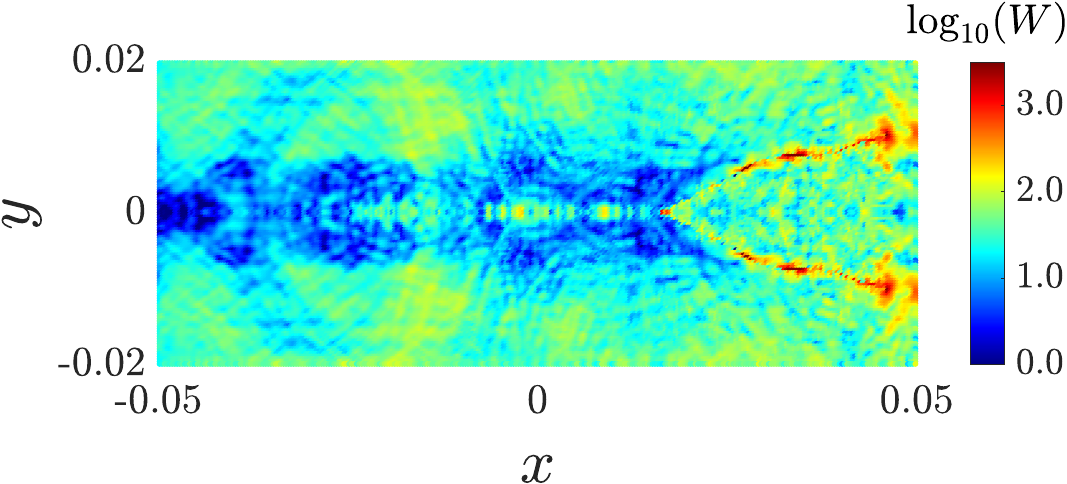}
        \caption{Strain energy density}
    \label{Fig: Example 4b results (b)}
    \end{subfigure}
        {\bf Critical energy density}\\[0.1in]
    \begin{subfigure}[t]{0.49\textwidth}
        \centering
       \includegraphics[width=.95\textwidth, trim=0cm 0cm 0cm 0cm,clip]{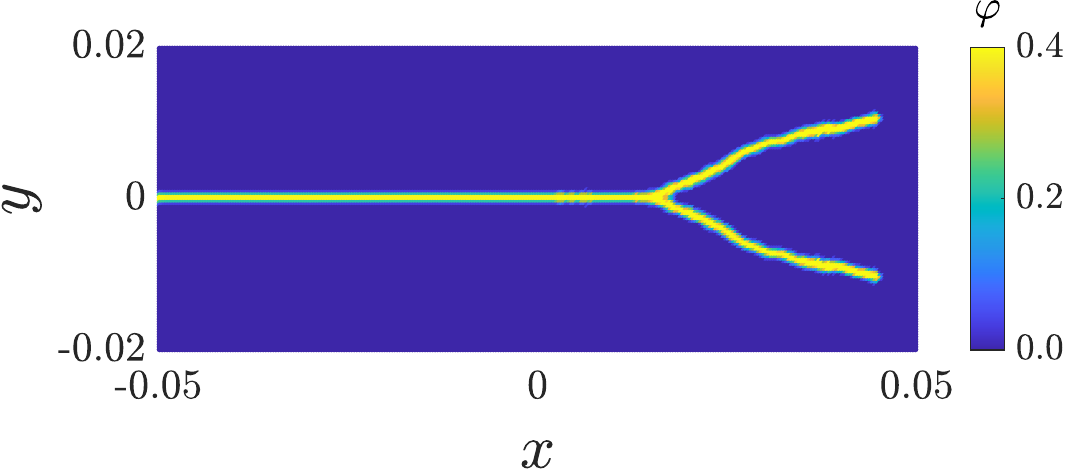}
        \caption{Damage}
    \label{Fig: Example 4b results (a)}
    \end{subfigure}%
      ~ 
    \begin{subfigure}[t]{0.49\textwidth}
        \centering
        \includegraphics[width=.95\textwidth, trim=0cm 0cm 0cm 0cm,clip]{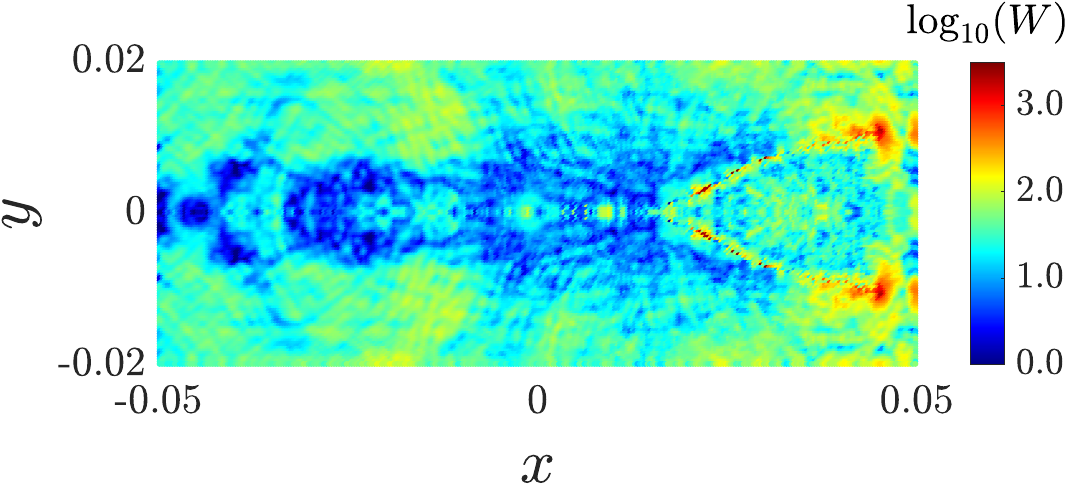}
        \caption{Strain energy density}
    \label{Fig: Example 4b results (b)}
    \end{subfigure}
        \caption{Comparison of the damage and strain energy density fields between the two bond-failure criteria for the case of $\sigma = 2$~MPa and $\alpha = 2$ at the final time in Example~3.} 
    \label{Fig: Example 4b results}
\end{figure*}

\section{Summary}\label{sec: discussion}

This paper presented a mathematical analysis of the critical stretch  and critical energy density bond-failure criteria for bond-based peridynamic fracture models. The analysis employed a model that generalizes the commonly used prototype microelastic brittle (PMB) constitutive model by incorporating an influence function.  
The critical energy density criterion was recast as a bond-dependent critical stretch criterion to rigorously prove that the two criteria are not equivalent in general and, consequently, result in different bond-breaking and fracture behaviors. 
In particular, depending on the choice of the influence function, certain bonds at critical stretch break under the critical stretch criterion but do not break under the critical energy density criterion, and vice versa. 
Moreover, for the critical energy density criterion, the choice of influence function may dictate whether shorter or longer bonds break first; this behavior is absent in the critical stretch criterion. 
Additionally, we prove that the two bond-failure criteria are equivalent if and only if the influence function is $\omega(\|\bm{\xi}\|) = \beta \|\bm{\xi}\|^{-1}$ for some $\beta>0$. 
Our findings were confirmed numerically for two-dimensional problems, with examples  presented for isotropic extension, crack tip evolution, and crack branching.  
It was revealed that the speed at which the crack tip evolves and the branching effect depend heavily on both the bond-failure criterion and the influence function. 

\section{Acknowledgements}

P. Seleson was supported by the Laboratory Directed Research and Development Program of Oak Ridge National Laboratory, managed by UT-Battelle, LLC, for the U. S. Department of Energy. 
P.~R.~Stinga was supported by Simons Foundation grant MP-TSM-00002709. 




\begin{thebibliography}{10}

\bibitem{Aguiar-Patriota-2023} 
A. R. Aguiar, T. V. B. Patriota, 
{Brittle fracture modeling using ordinary state‑based peridynamics with continuous bond‑breakage damage}, 
\textit{J. Peridyn. Nonlocal Model.}  
\textbf{5} 
(2023), 81--120.
\filbreak

\bibitem{Bobaru-Zhang-2015} F.~Bobaru, G.~Zhang, 
{Why do cracks branch? {A} peridynamic investigation of dynamic brittle fracture}, 
\textit{Int. J. Fract.}
\textbf{196} (2015), 59--98. 
\filbreak

\bibitem{Diana-Casolo-2019} V.~Diana and S.~Casolo, 
{A bond-based micropolar peridynamic model with shear deformability: Elasticity, failure properties and initial yield domains},
\textit{Int. J. Solids Struct.} 
\textbf{160} (2019), 201--231. 
\filbreak

\bibitem{Dipasquale-et-al-2022} D.~Dipasquale, G.~Sarego, P.~Prapamonthon, S.~Yooyen, and A.~Shojaei,
{A stress tensor-based failure criterion for ordinary state-based peridynamic models},
\textit{J. Appl. Comput. Mech.} 
\textbf{8} (2022), 617--628. 
\filbreak

\bibitem{Dipasqualeetal2017} D.~Dipasquale, G.~Sarego, M.~Zaccariotto, and U.~Galvanetto, 
{A discussion on failure criteria for ordinary state-based peridynamics}, 
\textit{Eng. Fract. Mech.}
\textbf{186} (2017), 378--398.  
\filbreak


\bibitem{Du-et-al-2018} 
Q.~Du, Y.~Tao, and X.~Tian, 
{A peridynamic model of fracture mechanics with bond-breaking}, 
\textit{J. Elasticity}
\textbf{132} 
(2018), 197--218.
\filbreak

\bibitem{Emmrich-Puhst-2016} 
E. Emmrich and D. Puhst,  
{A short note on modeling damage in peridynamics}, 
\textit{J. Elasticity}
\textbf{123} 
(2016), 245--252.
\filbreak


\bibitem{FosterSillingChen} J.~T.~Foster, S.~A.~Silling, and W.~Chen,
{An energy based criterion for use with peridynamic states},
\textit{Int. J. Multiscale Comput. Eng.}
\textbf{9} (2011), 675--687.  
\filbreak

\bibitem{Gerstle-et-al-2007} W.~Gerstle, N.~Sau, and S.~Silling, 
{Peridynamic modeling of concrete structures},
\textit{Nucl. Eng. Des.} 
\textbf{237} (2007), 1250--1258. 
\filbreak

 \bibitem{Gerstle} W.~Gerstle, N.~Sau, and S.~Silling,
{Peridynamic modeling of plain and reinforced concrete structures},
\textit{In: 18th International Conference on Structural Mechanics in Reactor
Technology (SMiRT 18)}, 
(2005),  54--68. 
\filbreak

\bibitem{Ghajari-et-al-2014} 
M.~Ghajari, L.~Iannucci, and P.~Curtis,   
{A peridynamic material model for the analysis of dynamic crack propagation in orthotropic media},
\textit{Comput. Methods Appl. Mech. Engrg.} 
\textbf{276} 
(2014), 431--452.
\filbreak

\bibitem{Bobaru-Ha-2010} Y.~D.~Ha, F.~Bobaru,
{Studies of dynamic crack propagation and crack branching with peridynamics.}, 
\textit{Int. J. Fract.}
\textbf{162} (2010), 229--244. 


\bibitem{Hattori-et-al-2018} 
G.~Hattori, J.~Trevelyan, and W.~M.~Coombs,  
{A non-ordinary state-based peridynamics framework for anisotropic materials}, 
\textit{Comput. Methods Appl. Mech. Engrg.},
\textbf{339} 
(2018), 416--442.
\filbreak

\bibitem{Ignatiev-et-al-2020} 
M.~Ignatiev, N.~Kazarinov, and Y.~Petrov, 
{Peridynamic modelling of the dynamic crack initiation}, 
\textit{Procedia Struct. Integr.},
\textbf{28}  
(2020), 1650--1654.
\filbreak


\bibitem{Ignatev-Oterkus-2025} 
M.~Ignatev and E.~Oterkus, 
{Remote stress fracture criterion in peridynamics}, 
\textit{Engineering with Computers},
\textbf{41} 
(2025), 3169--3192.
\filbreak


\bibitem{Ignatiev-et-al-2021} 
M.~O.~Ignatiev, Y.~V.~Petrov, and N.~A.~Kazarinov, 
{Simulation of dynamic crack initiation based on the peridynamic numerical model and the incubation time criterion}, 
\textit{Tech. Phys.},
\textbf{66}(3)  
(2021), 422--425.
\filbreak

\bibitem{Ignatiev-et-al-2023} 
M.~O.~Ignatiev, Y.~V.~Petrov, N.~A.~Kazarinov, and E.~Oterkus, 
{Peridynamic formulation of the mean stress and incubation time fracture criteria and its correspondence to the classical Griffith's approach}, 
\textit{Continuum Mech. Thermodyn.},
\textbf{35} 
(2023), 1523--1534.
\filbreak

\bibitem{Karpenko-et-al-2020} O.~Karpenko, S.~Oterkus, and E.~Oterkus,  
{An in-depth investigation of critical stretch based failure criterion in ordinary state-based peridynamics},
\textit{Int. J. Fract.} 
\textbf{226} (2020), 97--119. 
\filbreak

\bibitem{Li-et-al-2025} 
W.-J.~Li, Q.-Z.~Zhu, Y.-L.~Du, and J.-F.~Shao, 
{An extended bond-based peridynamic model with bond transverse deformation effects for quasi-brittle rocks}, 
\textit{Int. J. Rock Mech. Min. Sci.},
\textbf{190} 
(2025), 106099.
\filbreak

\bibitem{Lipton-2016} 
R.~Lipton,  
{Cohesive dynamics and brittle fracture}, 
\textit{J. Elasticity}
\textbf{124} 
(2016), 143--191.
\filbreak

\bibitem{Lipton-2014} R.~Lipton,
{Dynamic brittle fracture as a small horizon limit of peridynamics},
\textit{J. Elasticity},
\textbf{117} (2014), 21--50. 
\filbreak

\bibitem{Lipton-et-al-2018} 
R.~Lipton, E.~Said, and P.~Jha,  
{Free damage propagation with memory}, 
\textit{J. Elasticity}
\textbf{133} 
(2018), 129--153.
\filbreak

\bibitem{Liu-et-al-2022} W.~K.~Liu, S.~Li, and H.~S.~Park,   
{Eighty years of the finite element method: birth, evolution, and future},
\textit{Arch. Comput. Methods Eng.} 
\textbf{29} 
(2022), 4431--4453.
\filbreak

\bibitem{Madenci-et-al-2021} 
E.~Madenci, A.~Barut, and N.~Phan, 
{Bond‑based peridynamics with stretch and rotation kinematics for opening and shearing modes of fracture},
\textit{J. Peridyn. Nonlocal Model.},
\textbf{3} 
(2021), 211--254.
\filbreak

\bibitem{Madenci-Oterkus-2016} E.~Madenci  and S.~Oterkus,
{Ordinary state-based peridynamics for plastic deformation according to von Mises yield criteria with isotropic hardening},
\textit{J. Mech. Phys. Solids} 
\textbf{86} (2016), 192--219. 
\filbreak

\bibitem{Madenci-Oterkus-2014} E.~Madenci and E.~Oterkus,  
{Peridynamic Theory and Its Applications},
Springer, New York, 2014.
\filbreak

 
 \bibitem{O'Grady-Foster-2014} 
 J. O’Grady and J.~Foster, 
{Peridynamic beams: A non-ordinary, state-based model}, 
\textit{Int. J. Solids Struct.}  
 \textbf{51} (2014): 3177–3183.
 \filbreak

\bibitem{Oterkus-Madenci-2012} E.~Oterkus and E.~Madenci, 
{Peridynamic analysis of fiber-reinforced composite materials},
\textit{J. Mech. Mater. Struct.} 
\textbf{7} (2012), 45--84. 
\filbreak

\bibitem{Oterkus-Madenci-2015} S.~Oterkus and E.~Madenci, 
{Peridynamics for antiplane shear and torsional deformations},
\textit{J. Mech. Mater. Struct.} 
\textbf{10}(2) (2015), 167--193. 
\filbreak

\bibitem{Ravi-ChandarKnauss} K.~Ravi-Chandar and W.~G.~Knauss, 
{An experimental investigation into dynamic fracture: {IV.} On the interaction of stress waves with propagating cracks},
\textit{Int. J. Fract.}
\textbf{26} (1984), 189--200. 
\filbreak


\bibitem{Ren-et-al-2018} 
B.~Ren, C.~T.~Wu, P.~Seleson, D.~Zeng, and D.~Lyu, 
{A peridynamic failure analysis of fiber-reinforced composite laminates using finite element discontinuous Galerkin
approximations},
\textit{Int. J. Fract.} 
\textbf{214} 
(2018), 49--68.
\filbreak

\bibitem{Ren-et-al-2022} 
B.~Ren, C.~T.~Wu, P.~Seleson, D.~Zeng, M.~Nishi, and M.~Pasetto, 
{An FEM‑based peridynamic model for failure analysis of unidirectional fiber‑reinforced laminates},
\textit{J. Peridyn. Nonlocal Model.} 
\textbf{4} 
(2022), 139--158.
\filbreak

\bibitem{Ren-et-al-shear}
H.~Ren, X.~Zhuang, and T.~Rabczuk,
{A new peridynamic formulation with shear deformation for elastic solid},
\textit{J. Micromech. Mol. Phys.}
\textbf{1} (2016), 1650016.
\filbreak

\bibitem{Seleson} P.~D.~Seleson,
{Peridynamic multiscale models for the mechanics of materials: constitutive relations, upscaling from atomistic systems, and interface problems},
Thesis (Ph.D.)-The Florida State University (2010), 143pp.
\filbreak

\bibitem{SelesonParks} P.~Seleson and M.~Parks,
{On the role of the influence function in the peridynamic theory},
\textit{Int. J. Multiscale Comput. Eng.}
\textbf{9} (2011), 689--706.
\filbreak

\bibitem{SelesonEtAl} P.~Seleson, M.~Pasetto, Y.~John, J.~Trageser, and S.~Temple Reeve,
{PDMATLAB2D: a peridynamics MATLAB two-dimensional code},
\textit{J. Peridyn. Nonlocal Model.}
\textbf{6} (2024), 149--205.
\filbreak


\bibitem{Shou-et-al-2019} 
Y.~Shou, X.~Zhou, and F.~Berto, 
{3D numerical simulation of initiation, propagation and coalescence of cracks using the extended non-ordinary state-based peridynamics}, 
\textit{Theoretical and Applied Fracture Mechanics} 
 \textbf{101} (2019) ,254–268
 \filbreak

\bibitem{Silling}  S.~A.~Silling,
{Reformulation of elasticity theory for discontinuities and long-range forces},
\textit{J. Mech. Phys. Solids}
\textbf{48} (2000), 175--209.
\filbreak

\bibitem{SillingAskari} S.~A.~Silling and E.~Askari,
{A meshfree method based on the peridynamic model of solid mechanics},
\textit{Comput. Struct.}
\textbf{83} (2005), 1526--1535.
\filbreak

\bibitem{Sillingetal2007} S.~A.~Silling, M.~Epton, O.~Weckner, J.~Xu, and E.~Askari,  
{Peridynamic states and constitutive modeling},
\textit{J. Elasticity}
\textbf{88} (2007), 151--184.
\filbreak

\bibitem{Silling-Lehoucq-2008} S.~A.~Silling and R.~B.~Lehoucq,
{Convergence of peridynamics to classical elasticity theory},
\textit{J. Elasticity},
\textbf{93} (2008), 13--37. 
\filbreak

\bibitem{Song-et-al-2025} 
Z.~Song, G.~Wang, D.~Lu, X.~Zhou, T.~Rabczuk, and X.~Du, 
{A modeling method of failure for concrete considering the stress state in peridynamics},
\textit{Int. J. Solids Struct.},
\textbf{320} 
(2025), 113536.
\filbreak


\bibitem{Trageser-Seleson-2020} J.~Trageser and P.~Seleson,
{Bond-based peridynamics: a tale of two Poisson's ratios},
\textit{J. Peridyn. Nonlocal Model.},
\textbf{2} (2020), 278--288. 
\filbreak

\bibitem{Tupek-etal-2013} M.~R.~Tupek, J.~J.~Rimoli, and R.~Radovitzky, 
{An approach for incorporating classical continuum damage models in state-based peridynamics},
\textit{Comput. Methods Appl. Mech. Engrg.} 
\textbf{263} (2013), 20--26. 
\filbreak

\bibitem{Wang-et-al-2018}
Y.~Wang, X.~Zhou,~Y.~Wang, and Y.~Shou,
{A 3-D conjugated bond-pair-based peridynamic formulation for initiation and propagation of cracks in brittle solids},
\textit{Int. J. Solids Struct.},
\textbf{134} (2018), 89--115.
\filbreak

\bibitem{Wang-et-al-2023} 
W.~Wang, Q.-Z.~Zhu, T.~Ni, B.~Vazic, P.~Newell, and S.~P.~A.~Bordas, 
{An extended peridynamic model equipped with a new bond-breakage criterion for mixed-mode fracture in rock-like materials},
\textit{Comput. Methods Appl. Mech. Engrg.} 
\textbf{411} 
(2023), 116016.
\filbreak

\bibitem{Warren-et-al-2009} T.~L.~Warren, S.~A.~Silling, A.~Askari, O.~Weckner, M.~A.~Epton, and J.~Xu, 
{A non-ordinary state-based peridynamic method to model solid material deformation and fracture},
\textit{Int. J. Solids Struct.},
\textbf{46} (2009), 1186--1195. 
\filbreak

\bibitem{Yang-et-al-2024} 
Z.~Yang, H.~Wang, and M.~Sharma,  
{A peridynamic compensated critical energy density criterion for mixed-mode fracturing in quasi-brittle materials}, 
\textit{Theor. Appl. Fract. Mech.},
\textbf{134} 
(2024), 104736.
\filbreak

\bibitem{Zhang-Qiao-2019} Y.~Zhang and P.~Qiao, 
{A new bond failure criterion for ordinary state-based peridynamic mode II fracture analysis},
\textit{Int. J. Fract.} 
\textbf{215} (2019), 105--128. 
\filbreak

\bibitem{Zhang-Qiao-2018} 
H.~Zhang and P.~Qiao, 
{A state-based peridynamic model for quantitative fracture analysis},
\textit{Int. J. Fract.} 
\textbf{211} 
(2018), 217--235.
\filbreak


\bibitem{Zhou-Wang-2016} 
X.-P.~Zhou and Y.-T.~Wang, 
{Numerical simulation of crack propagation and coalescence in pre-cracked rock-like Brazilian disks using the non-ordinary state-based peridynamics}, 
\textit{Int. J. Rock Mech. Min. Sci.},
\textbf{89} 
(2016), 235--249.
\filbreak



\end{thebibliography}
\end{document}